\documentclass[11pt, a4paper]{article}
\usepackage[T1]{fontenc}
\usepackage{a4wide, url}
\usepackage[english]{babel}
\usepackage{graphicx}
\usepackage{array}
\usepackage{natbib}
\usepackage[textfont=sc,labelfont=bf]{caption}
\usepackage{setspace}
\usepackage{multirow}
\usepackage{amsthm}
\usepackage{amsmath}
\usepackage{amsfonts}
\usepackage{amssymb}
\usepackage{bigints}
\usepackage{pgf} 
\usepackage{bm}
\usepackage{paralist}
\usepackage{pdflscape}
\usepackage{xfrac}
\usepackage[linesnumbered,ruled,vlined]{algorithm2e}
\SetKwInput{KwInput}{Input}                
\SetKwInput{KwOutput}{Output}              
\usepackage{enumitem,kantlipsum}
\usepackage{xr}
\usepackage{yfonts}
\usepackage{xcolor}

\usepackage{comment}
\excludecomment{ext}\excludecomment{comment}

\usepackage[top=1in,bottom=1in,left=.75in,right=.75in]{geometry}

\usepackage{tabu}

\newcolumntype{L}[1]{>{\raggedright\let\newline\\\arraybackslash\hspace{0pt}}m{#1}}
\newcolumntype{C}[1]{>{\centering\let\newline\\\arraybackslash\hspace{0pt}}m{#1}}
\newcolumntype{R}[1]{>{\raggedleft\let\newline\\\arraybackslash\hspace{0pt}}m{#1}}

\DeclareMathAlphabet\mathbfcal{OMS}{cmsy}{b}{n}

\newcommand{\mbf}{\mathbf}
\newcommand{\beq}{\begin{equation}}
\newcommand{\eeq}{\end{equation}}
\newcommand{\bea}{\begin{eqnarray}}
\newcommand{\eea}{\end{eqnarray}}
\newcommand{\ba}{\begin{array}}
\newcommand{\ea}{\end{array}}
\newcommand{\bit}{\begin{itemize}}
\newcommand{\eit}{\end{itemize}}
\newcommand{\ben}{\begin{enumerate}} 
\newcommand{\een}{\end{enumerate}}
\newcommand{\bpm}{\begin{pmatrix}}
\newcommand{\epm}{\end{pmatrix}}
\newcommand{\bbm}{\begin{bmatrix}}
\newcommand{\ebm}{\end{bmatrix}}

\renewcommand{\l}{\left}
\renewcommand{\r}{\right}

\newcommand{\E}[0]{\mathbb{E}}
\newcommand{\Var}[0]{\mathbb{V}\mathrm{ar}}
\newcommand{\Cov}[0]{\mathbb{C}\mathrm{ov}}

\newcommand{\nn}{\nonumber}
\newcommand{\wh}{\widehat}
\newcommand{\wt}{\widetilde}

\newtheorem{ass}{Assumption}
\newtheorem{theorem}{Theorem}

\newtheorem{prop}{Proposition}
\newtheorem{rem}{Remark}
\newtheorem{lem}{Lemma}
\newtheorem{cor}{Corollary}

\graphicspath{{Figures/Submission/}}
\title{\textsc{\large 	
Asymptotic equivalence of
Principal Components \linebreak and  Quasi Maximum Likelihood estimators \linebreak  in Large Approximate Factor Models}}
\date{ }

\begin{document}
\maketitle

\begin{center}\vspace{-1.5cm}
Matteo Barigozzi$^*$ \\[.1cm] 
%
\small This version: \today
\end{center}

\begin{abstract}
This paper investigates the properties of Quasi Maximum Likelihood estimation of an approximate factor model for an $n$-dimensional vector of stationary time series.
We prove that the factor loadings estimated by Quasi Maximum Likelihood are asymptotically equivalent, as $n\to\infty$,  to those estimated via Principal Components. Both estimators are, in turn, also asymptotically equivalent, as $n\to\infty$, to the unfeasible Ordinary Least Squares estimator we would have if the factors were observed. 
We also show that the usual sandwich form of the asymptotic covariance matrix of the Quasi Maximum Likelihood estimator is asymptotically equivalent to the simpler asymptotic covariance matrix of the unfeasible Ordinary Least Squares. All these results hold in the general case in which the idiosyncratic components are cross-sectionally heteroskedastic, as well as serially and cross-sectionally weakly correlated. The intuition behind these results is that as $n\to\infty$ the factors can be considered as observed, thus showing that factor models enjoy a blessing of dimensionality. 
\vspace{0.5cm}


\end{abstract}

\renewcommand{\thefootnote}{$\ast$} 
\thispagestyle{empty}

\footnotetext{ Universit\`a di Bologna, matteo.barigozzi@unibo.it} 

\renewcommand{\thefootnote}{\arabic{footnote}}

\section{Introduction}

Factor models are one of the major dimension reduction techniques used to analyze large panels of time series. Some of their most successful applications are, among many others, in finance (Chamberlain and Rothschild, 1983, \citealp{connor2006common}, \citealp{ait2017using}, \citealp{kim2019factor}),  and macroeconomics (Stock and Watson, 2002b,
\citealp{BBE05}, \citealp{FHLR05}, \citealp{de2008forecasting}, \citealp{Nowcasting}).
 
 Let us assume to observe an $n$-dimensional zero-mean stochastic process over $T$ periods: $\{x_{it},\, i=1,\ldots , n,\, t=1,\ldots, T\}$, such that
\begin{align}
x_{it}&=\bm\lambda_i^\prime \mbf F_t + \xi_{it}, \quad i=1,\ldots , n,\quad t=1,\ldots, T,\label{eq:SDFM1R}
\end{align}
where $\bm\lambda_i=(\lambda_{i1}\cdots\lambda_{ir})^\prime$ and $\mbf F_t=(F_{1t}\cdots F_{rt})^\prime$ are $r$-dimensional  unobserved vectors, called loadings and factors, respectively, and with $r<\min(n,T)$. We also call $\xi_{it}$ the idiosyncratic component and 
$\chi_{it}=\bm \lambda_i^\prime \mbf F_t$ the common component of the $i$th observed variable at time $t$. Both the factors and the idiosyncratic components are allowed to be serially correlated. Furthermore, the idiosyncratic components are also allowed to be cross-sectionally correlated. We call such model an approximate factor model. This is the class of factor models we consider in this paper and studied, e.g., by \citet{Bai03}. It is a restricted version of the generalized dynamic factor model originally proposed by \citet{FHLR00}, where the factors are loaded with lags (the loadings are  linear filters) and not just contemporaneously. Another popular alternative, not considered in this paper, is the factor model studied, e.g., by \citet{lam2012factor}, where the idiosyncratic components are instead assumed to be serially uncorrelated.

There are two main ways to estimate the factor loadings in \eqref{eq:SDFM1R}. First, by Principal Component (PC) analysis \citep{stockwatson02JASA,Bai03,FLM13}, and, second, by Quasi Maximum Likelihood (QML) estimation \citep{baili12,baili16}. In both cases, we can first estimate the loadings and then estimate the factors by linear projection, possibly weighted by the idiosyncratic variances, of the observables onto the estimated loadings. The PC estimator of the loadings is easily implementable since it just requires to compute the $r$ leading eigenvectors and eigenvalues of the sample covariance matrix of the observables. The QML estimator of the loadings does not have a closed form solution, so numerical maximization is required.
It is typically defined as the maximizer of a mis-specified log-likelihood where the  idiosyncratic components are treated as if they were uncorrelated even if the true ones are correlated. 

Although QML is the classical way to estimate a factor model, dating back more than fifty years ago \citep[and references therein]{lawleymaxwell71,MBQML}, in the recent years, PC analysis has gained popularity, given its non-parametric nature and ease of implementation. Nevertheless, QML estimation still retains an important role in factor analysis 
since for many reasons. For example, it allows to easily impose constraints on the loadings \citep[see, e.g.,the applications in][]{CGM16,DCGF2021}, and it fully addresses idiosyncratic cross-sectional heteroskedasticity, while the PC method does not \citep[see the comments in][]{baili12,baili16}.

In this paper, we compare the PC and QML estimators of the loadings and we show that, under a minimal unified set of standard assumptions and identifying constraints, the two estimators are asymptotically equivalent as $n\to\infty$ (see Theorem \ref{th:QML}). 
It is well known that both estimators of the loadings are $\min(n,\sqrt T)$-consistent as $n,T\to\infty$ and are also asymptotically normal if we assume $n^{-1}\sqrt T\to 0$. Standard references for these results are, among others, \citet[Theorem 2]{Bai03} for PC, and \citet[Theorem 1]{baili16} for QML. Proving the asymptotic properties of the PC estimator is long but straightforward. However, things are much more complicated for the QML estimator, essentially because no closed form expression exists for this estimator. Moreover, the existing proofs for the two estimators do not make use of the same assumptions nor of the same identification constraints needed to uniquely identify the loadings. 

Our main result has some important implications. First, if $n^{-1}\sqrt T\to 0$, as $n,T\to\infty$, then the QML estimator has the same asymptotic distribution as the unfeasible OLS estimator we would obtain if the factors were observed. Therefore, we are able to prove asymptotic normality of the QML estimator in an indirect and easy way (see Theorem \ref{cor:QMLcons}). Our approach provides a much simpler alternative to the approach used by \citet[Theorem 1]{baili16}, and it does not require to study also the properties of the QML estimators of the idiosyncratic variances. Second, it nests the case of spherical idiosyncratic components \citep{tippingbishop99}, in which case the equivalence of PC and QML estimator is often quoted, but, to the best of our knowledge, it has never been formally proved, at least under the standard minimal set of assumptions used in this paper. Related to these two  points, we stress that our result is, in fact, more general since it holds without the need of considering QML estimation based on a mis-specified  log-likelihood with a diagonal idiosyncratic covariance. Third, and last, our result paves the way towards studying the asymptotic properties of the QML estimator of a factor model where we also explicitly model the dynamics of the factors, and for which only partial results exist (\citealp{DGRqml}, \citealp{baili16}).

The theoretical analysis of approximate factor models requires a double asymptotic framework. Indeed, only if $n,T\to\infty$ we can consistently estimate the eigenvectors of the covariance matrix of the data which are needed for PC estimation, and we can control for the mis-specifications introduced in the log-likelihood when considering QML estimation. This means that QML estimation of approximate factor models does not fall into the framework of classical QML estimation, where $n$ is fixed. 

In a second contribution, we reconcile the asymptotic distribution of the QML estimator of the loadings, derived in this paper, with the results of classical likelihood-based inference. In particular, we show that, when evaluated in the true value of the parameters or in the value of the QML estimator, the first and second derivatives of the factor model log-likelihood we maximize are asymptotically, as $n,T\to\infty$, equivalent to the first and second derivatives of the log-likelihood we would have if the factors were observed (see Theorem \ref{prop:hessiani}). This result also allows us to compute simple estimators of the asymptotic covariance matrix of the estimated loadings without the need of computing the derivatives  of the log-likelihood which have rather long and  complex expressions.


The paper is organized as follows. In Section \ref{sec:sdfm} we present the model and all assumptions. In Section \ref{sec:PC} we reconsider the PC estimator of the loadings and in Theorem \ref{th:CLTL} we derive its asymptotic properties using a new approach equivalent to the one proposed by \citet{Bai03} but which is more convenient for the present work.  In Section \ref{sec:PCQML} we review the existing results on the QML estimator of the loadings by \citet{baili16}. In Section \ref{sec:PCQML0} we present our first contribution in Theorem \ref{th:QML} where we prove that PC and QML estimators of the loadings are asymptotically equivalent. 
In Section \ref{sec:fisher} we present our second contribution in Theorem \ref{prop:hessiani}, where we prove the asymptotically equivalence of the first and second derivatives of the factor model log-likelihood and the log-likelihood of a model with observed factors. 
In Section \ref{sec:fact} we briefly discuss estimation of factors.
In Section \ref{sec:MC} we provide simulation results confirming our theoretical results.  In Section \ref{sec:conc} we conclude.
In the Supplementary Material we prove all other main and auxiliary theoretical results.


%

\section{Model and Assumptions}\label{sec:sdfm}


Given the $n$-dimensional vector  $\mbf x_t=(x_{1t}\cdots x_{nt})^\prime$,  we can also write model \eqref{eq:SDFM1R} in vector notation: 
\[
\mbf x_t = \bm\Lambda\mbf F_t +\bm \xi_t,\quad t=1,\ldots, T,
\]
where $\bm\xi_{t}=(\xi_{1t}\cdots\xi_{nt})^\prime$ is the $n$-dimensional vector of  idiosyncratic components and  $\bm\Lambda=(\bm\lambda_1\cdots\bm\lambda_n)^\prime$ is the $n\times r$ matrix of factor loadings. We call $\bm\chi_t=\bm\Lambda\mbf F_t$ the vector of common components.

Moreover, we can collected all observations into the $T\times n$ matrix $\bm X=(\mbf x_1\cdots\mbf x_T)^\prime$, and we can write model \eqref{eq:SDFM1R} also in matrix notation:
\beq\label{eq:SDFM1Rmat}
\bm X =\bm F\bm\Lambda^\prime +\bm\Xi,
\eeq
where $\bm \Xi=(\bm\xi_1\cdots\bm \xi_T)^\prime$ is a $T\times n$ matrix of idiosyncratic components, and $\bm F=(\mbf F_1\cdots\mbf F_T)^\prime$ is the $T\times r$ matrix of factors. 

The following assumptions are similar to those made by  \citet[]{Bai03} for PC estimation while slightly differ from those in \citet{baili16} for QML estimation. Throughout, we  highlight the main differences or similarities. 

We start by characterizing the common component by means of the following assumption.

\begin{ass}[\textsc{common component}]\label{ass:common} $\,$
\begin{compactenum}[(a)]
	\item  $\lim_{n\to\infty}\Vert n^{-1}\bm\Lambda^\prime\bm\Lambda-\bm\Sigma_{\Lambda}\Vert=0$, where $\bm\Sigma_{\Lambda}$ is $r\times r$ positive definite, and, for all $i\in\mathbb N$, 
	$\Vert\bm\lambda_i\Vert\le M_\Lambda$ for some finite positive real $M_\Lambda$ independent of $i$.  
\item For all $t\in\mathbb Z$, $\E[\mbf F_{t}]=\mbf 0_r$ and $\bm\Gamma^F=\E[\mbf F_t\mbf F_t^\prime]$ is $r\times r$ positive definite and $\Vert\bm\Gamma^F\Vert\le M_F$ for some finite positive real $M_F$ independent of $t$.

	\item 
	\begin{inparaenum} 
	\item [(i)] For all $t\in\mathbb Z$, $\E[\Vert \mbf F_t\Vert^4]\le K_F$ for some finite positive real $K_F$ independent of $t$; \\
	\item	[(ii)] for all $i,j=1,\ldots, r$, all $s=1,\ldots, T$, and all $T\in\mathbb N$,
\[
\E\l[\l\vert\frac 1{\sqrt{T}}\sum_{t=1}^T\l\{F_{is}F_{jt}-\E[F_{is}F_{jt}]\r\} \r\vert^2\r]
\le C_F
\]
for some finite positive real $C_F$ independent of $i$, $j$, $s$, and $T$.
	\end{inparaenum}
	\item There exists an integer $N$ such that for all $n> N$, $r$ is a finite positive integer, independent of $n$. 

\end{compactenum}
\end{ass}

Part (a) is standard \citep[Assumption B]{Bai03}. It implies that, asymptotically, as $n\to\infty$, the loadings matrix has asymptotically maximum column rank $r$, and that 
for any given $n\in\mathbb N$, each factor has a finite contribution to each of the $n$ observed series (upper bound on $\Vert\bm\lambda_i\Vert$). A similar requirement is in \citet[Assumption B]{baili16}.
We only consider non-random factor loadings for simplicity and in agreement with classical factor analysis where the loadings are the parameters of the model (\citealp[see, e.g.,][]{lawleymaxwell71}).

Part (b) assumes that the factors have zero mean and have a finite full-rank covariance matrix $\bm\Gamma^F$, so they are non-degenerate. In part (c-i) we assume finite 4th order moments of the factors.  These are standard requirements \citep[Assumption A]{Bai03}. 
Part (c-ii) is very general, it implies that the sample covariance matrix of the factors is a $\sqrt T$-consistent estimator of its population counterpart $\bm\Gamma^F$, which has full rank because of part (b). It is immediate to see that it is equivalent to asking for 4th order summable cross-cumulants which a necessary and sufficient condition for consistent estimation (\citealp[pp. 209-211]{hannan}).
This approach is high-level in that it does not make any specific assumption on the dynamics of $\{\mbf F_t\}$. Obviously this implies the usual assumption of convergence in probability  made in this literature: $\mathrm P\text{-}\lim_{T\to\infty}\Vert T^{-1}{\bm F^\prime\bm F}-\bm\Gamma^F\Vert=0$ \citep[Assumption A]{Bai03}. Nothing would change in our proofs if we directly assumed this latter condition instead of part (c-ii). Notice that \citet[Assumption A]{baili16} treat the factors as being deterministic, but essentially make the same assumption as our parts (b) and (c).

Part (d) implies the existence of a finite number of factors. In particular, the number of common factors, $r$, is identified only as $n\to\infty$. Here $N$ is the minimum number of series we need to be able to identify $r$ so that $r\le N$. Hereafter, when we say ``for all $n\in\mathbb N$'' we always mean that $n>N$ so that $r$ can be identified. In practice, we must always work with $n$ such that $r<n$. Moreover, because PC estimation is based on eigenvalues of an $n\times n$ matrix estimated using $T$ observations, then we must also have samples of size $T$ such that $r<T$. Therefore, sometimes it is directly assumed that $r<\min(n,T)$.

To characterize the idiosyncratic component, we make the following assumptions.

\begin{ass}[\textsc{idiosyncratic component}]\label{ass:idio}
  $\,$
\begin{compactenum}[(a)]
	\item For all $i\in\mathbb N$ and all $t\in\mathbb Z$, $\E[\xi_{it}]= 0$ and $\sigma_i^2=\E[\xi_{it}^2]$ is such that $C_\xi\le \sigma_i^2\le C_\xi^\prime$ for some finite positive reals $C_\xi$ and $C_\xi^\prime$  independent of  $i$ and $t$. 
	
	\item For all $i,j\in\mathbb N$, all $t\in\mathbb Z$, and all $k\in\mathbb Z$, $\vert \E[\xi_{it}\xi_{j,t-k}]\vert\le \rho^{\vert k\vert} M_{ij}$, where $\rho$ and $M_{ij}$ are finite positive reals independent of $t$ such that $0\le \rho <1$, $\sum_{j=1,j\ne i}^n M_{ij}\le M_\xi$, and $\sum_{i=1, i\ne j}^n M_{ij}\le M_\xi$ for some finite positive real $M_{\xi}$ independent of  $i$, $j$, and $n$. 
	
\item 
\begin{inparaenum}
\item [(i)]  For all $i=1,\ldots, n$, all $t=1,\ldots, T$, and all $n,T\in\mathbb N$, $\E[\xi_{it}^4]\le Q_\xi$ for some finite positive real $Q_\xi$ independent of $i$ and $t$;\\
\item [(ii)] for all $j=1,\ldots, n$, all $s=1,\ldots, T$, and all $n,T\in\mathbb N$,
\[
\E\l[\l\vert\frac 1{\sqrt{nT}} \sum_{i=1}^n\sum_{t=1}^T\l\{\xi_{is}\xi_{jt}-\E[\xi_{is}\xi_{jt}]\r\} \r\vert^2\r]
\le K_\xi
\]
for some finite positive real $K_\xi$ independent of $j$, $s$, $n$, and $T$.
\end{inparaenum}
 \end{compactenum}
 \end{ass}

By part (a), we have that the idiosyncratic components have zero mean. This, jointly with Assumption \ref{ass:common}(b)  by which $\E[\mbf F_t]=\mbf 0_r$, implies that we are implicitly assuming that each observed series has zero mean, i.e., $\E[ x_{it}]=0$ (this is without loss of generality), and strictly positive variance, i.e., $\mathbb V\text{ar}(x_{it})>0$, for all $i\in\mathbb N$ (\citealp[Assumption C.2]{baili16}). For the case of non-zero mean see Remark \ref{rem:media}.

Part (b) has a twofold purpose. First, it limits the degree of serial correlation of the idiosyncratic components by imposing geometric decay of their autocovariances. Second, it also limits the degree of cross-sectional correlation between idiosyncratic components, which is usually assumed in approximate factor models. 
In particular,  by setting $k=0$, it follows also that for all $i\in\mathbb N$, $\sigma_i^2 \le M_\xi$. Thus, all idiosyncratic components have finite variance. This jointly with Assumptions \ref{ass:common}(a) and \ref{ass:common}(b) implies that 
 each observed time series has finite variance, i.e.,  $\mathbb V\text{ar}(x_{it})<\infty$, for all $i\in\mathbb N$. In Lemma \ref{lem:Gxi} we show that part (b) implies all usual conditions on second order moments typically found in the literature (\citealp[Assumptions C.2, C.3, C.4, and E]{Bai03}, and \citealp[Assumptions C.3, C.4, and E.1]{baili16}). 

Part (c-i) assumes finite 4th order moments of the idiosyncratic components. Part (c-ii) gives summability conditions across the cross-section and time dimensions for the 4th order cumulants of $\{\xi_{it}\}$. Jointly with Assumptions \ref{ass:common}(a) and \ref{ass:idio}(a) it nests the requirement in \citet[Assumption E.2]{baili16}. It implies that the sample (auto)covariances between $\{\xi_{it}\}$ and $\{\xi_{jt}\}$ are $\sqrt T$-consistent estimators of their population counterparts. In particular, by choosing $s=t$ we see that we can consistently estimate the $(i,j)$th entry of the idiosyncratic covariance matrix $\bm\Gamma^\xi$. Notice that, contrary to the existing literature (\citealp[Assumptions C.1 and C.5]{Bai03}, and \citealp[Assumptions C.1 and C.5]{baili16}), there is no need to ask for finite 8th order moments and cumulants.

We then make  a series of  identifying assumptions.
\begin{ass}[\textsc{independence}]
\label{ass:ind} The processes
$\{\xi_{it},\, i\in\mathbb N,\, t\in\mathbb Z\}$ and $\{F_{jt},\, j=1,\ldots, r,\, t\in\mathbb Z\}$ are mutually independent. 
\end{ass}

This assumption obviously implies that the factors and, therefore, the common components are independent of the idiosyncratic components at all leads and lags and across all units. This is compatible for example with a macroeconomic interpretation of factor models, according to which the factors driving the common component are independent of the idiosyncratic components representing measurement errors or local dynamics. This assumption is made for simplicity and could be easily relaxed \citep[Assumption D]{Bai03}. 

 Let $\bm\Gamma^\chi=\E[\bm\chi_t\bm\chi_t^\prime]=\bm\Lambda\bm\Gamma^F \bm\Lambda^\prime$ with $r$ largest eigenvalues $\mu_j^\chi$, $j=1,\ldots, r$, sorted in decreasing order. In Lemma \ref{lem:Gxi}(iv) we prove that,
for all $j=1,\ldots,r$, 
\beq\label{eq:lindiv}
\underline C_j\!\le \lim\inf_{n\to\infty} \frac{ \mu_{j}^\chi}n\le\lim\sup_{n\to\infty} \frac{\mu_{j}^\chi}n\le\! \overline C_j,
\eeq
where $\underline C_j$ and $ \overline C_j$ are finite positive reals. This means we consider only strong factors, i.e., fully pervasive. Furthermore in Lemma \ref{lem:Gxi}(v) we prove that the largest eigenvalue of $\bm\Gamma^\xi$ is such that
\beq\label{eq:evalidio}
 \sup_{n\in\mathbb N}\mu_{1}^\xi  \le C_\xi^\prime+M_\xi,
 \eeq
 where $C_\xi^\prime$ and $M_\xi$ are defined in Assumptions \ref{ass:idio}(a) and \ref{ass:idio}(b), respectively. 
Because of Weyl's inequality (see, e.g., \citealp{MK04}, Theorem 1), conditions \eqref{eq:lindiv}-\eqref{eq:evalidio} and Assumption \ref{ass:ind} imply the eigengap in the eigenvalues $\mu_j^x$, $j=1,\ldots,n$, of $\bm\Gamma^x=\E[\mbf x_t\mbf x_t^\prime]=\bm\Gamma^\chi+\bm\Gamma^\xi$ (see Lemma \ref{lem:Gxi}(vi)):
\begin{align}
&\underline C_r\le \lim\inf_{n\to\infty} \frac{ \mu_{r}^x}n\le\lim\sup_{n\to\infty} \frac{\mu_{r}^x}n\le \overline C_r\;\;\text{ and }\;\;\sup_{n\in\mathbb N}\mu_{r+1}^x \le C^\prime_\xi+M_{\xi}.\nn
\end{align}
This property allows us to identify, asymptotically, as $n\to\infty$, the number of factors $r$ (see, e.g., \citealp{baing02}, \citealp{onatski10}, \citealp{trapani2018randomized}, among many others). And, therefore, as $n\to\infty$ we also identify the common and idiosyncratic components.

\begin{ass}[\textsc{distinct eigenvalues}]
\label{ass:eval}~For all $n\in\mathbb N$ and all $i=1,\ldots, r-1$, $\mu_i^\chi > \mu_{i+1}^\chi$.
\end{ass}

This  assumption is needed in order to identify the eigenvectors in PC estimation.
Note that it implies that $\overline C_j<\underline C_{j-1}$, $j=2,\ldots, r$, in \eqref{eq:lindiv} and that the $r$ eigenvalues of $\bm\Sigma_\Lambda\bm\Gamma^F$ are distinct \citep[Assumption G]{Bai03}. Indeed, these coincide with the non-zero eigenvalues of $\lim_{n\to\infty}n^{-1}{\bm\Gamma^\chi}$ which are given by $\lim_{n\to\infty}n^{-1}{\mbf M^\chi}=(\bm\Sigma_\Lambda)^{1/2}\bm\Gamma^F(\bm\Sigma_\Lambda)^{1/2}$. 

The factors and the loadings can be identified by means of the following assumption.

\begin{ass}[\textsc{Identification}]\label{ass:ident}
$\,$
\begin{inparaenum}[(a)]
\item For all $n\in\mathbb N$, $n^{-1}{\bm\Lambda^\prime\bm\Lambda}$ diagonal.
\item For all $T\in\mathbb N$, $T^{-1}{\bm F^\prime\bm F}=\mbf I_r$.
\end{inparaenum}
\end{ass}

This assumption a standard requirement in PC based exploratory factor analysis (see, e.g, \citealp[identification constraints PC1]{baing13}). It has some important and usueful implications in terms of identification. In particular, in Proposition \ref{prop:K00} we prove that, for all $n\in\mathbb N$, $n^{-1}{\bm\Lambda^\prime\bm\Lambda} = n^{-1}{\mbf M^\chi}$ 
and 
$\bm\Lambda=\mbf V^\chi(\mbf M^\chi)^{1/2}$, where $\mbf M^\chi$ is the $r\times r$ diagonal matrix of the nonzero eigenvalues of $\bm\Gamma^\chi$ (sorted in decreasing order) and $\mbf V^\chi$ is the $n\times r$ matrix having as columns the corresponding normalized eigenvectors.
Notice that part (a) differs from \citet[Condition IC3]{baili16}, where it is assumed that $n^{-1}{\bm\Lambda^\prime(\bm\Sigma^\xi)^{-1}\bm\Lambda}$ is diagonal, for all $n\in\mathbb N$, with $\bm\Sigma^\xi$ being the diagonal matrix with entries the idiosyncratic variances $\sigma_i^2$, $i=1,\ldots, n$. This is a common requirement in QML based exploratory factor analysis. One of the aims of this paper is to reconsider QML estimation under Assumption \ref{ass:ident}(a).
Part (b) is instead also assumed in \citet[Condition IC3]{baili16} but for deterministic factors.

Clearly since Assumption \ref{ass:ident} concerns only the product $\bm\Lambda^\prime\bm\Lambda$, it allows us to identify the loadings only up to right multiplication by a diagonal matrix with entries $\pm 1$, i.e., the columns of $\bm\Lambda$ are identified only up to a sign. We can fix such sign by means of the following assumption.

\begin{ass}[\textsc{Global identification}]\label{ass:sign}
For all $j=1,\ldots, r$, one of the two following conditions holds:
\begin{inparaenum}[(a)]
\item $\lambda_{1j}> 0$; or
\item $F_{j1}> 0$. 
\end{inparaenum}
\end{ass}
Since loadings are considered as deterministic and are estimated as eigenvectors, part (a) is more natural and can be easily imposed just by setting the sign of the  eigenvectors of the sample covariance matrix $\wh{\bm\Gamma}^x$ accordingly. As a consequence of Assumption \ref{ass:sign} both the loadings and the factors are globally identified \citep[see also the comments in][Remark 1]{baing13}.

Finally, in order to derive the asymptotic distribution of the considered estimators of the loadings it is common to assume the following assumption (\citealp[Assumption F.4]{Bai03}, and \citealp[Assumption F.1]{baili16}).

\begin{ass}[\textsc{Central limit theorem}]
\label{ass:CLT}
For all $i\in\mathbb N$, as $T\to\infty$,
	\[
	\frac 1{\sqrt T}\sum_{t=1}^T
	 \mbf F_t \xi_{it}
	\to_d
	\mathcal N\l(\mbf 0_r, \bm\Phi_i
	\r),
	\]
	where $\bm\Phi_i=\lim_{T\to\infty}T^{-1}\sum_{t,s=1}^T
	\E_{}[\mbf F_t\mbf F_s^\prime\xi_{it}\xi_{is}]$.
\end{ass}

There are many ways to derive this condition from more primitive assumptions, for example, by assuming strong mixing factors and idiosyncratic components both with finite $4+\epsilon$ moments. 

\begin{rem}\label{rem:baili}
\upshape{
As discussed above, all assumptions for PC estimation by \citet{Bai03} are implied or equivalent to ours. Regarding QML estimation and the assumptions in \citet{baili16}, we do not need their Assumption D, which requires the QML estimators of the idiosyncratic variances, $\sigma_i^2$, $i=1,\ldots, n$, to be finite and strictly positive, and we do not require the moment condition E.3, which is needed only for estimation of $\sigma_i^2$ and it is not necessary to prove our results. Assumptions E.4, E.5, E.6, F.2, and F.3 in \citet{baili16} are needed only for studying the behavior of the estimated factors so are not needed here. All other assumptions in \citet{baili16} are equivalent or nested into ours, with the crucial exception of their identifying condition on the loadings which, as noted above, differs from ours.}
\end{rem}

\begin{rem}\label{rem:media}
\upshape{
Allowing for data with non-zero mean simply amounts to adding a constant term, $\alpha_i\ne 0$, to \eqref{eq:SDFM1R}, so that $x_{it}=\alpha_i+\bm\lambda_i^\prime\mbf F_t+\xi_{it}$. All estimators described in the following retain the same properties even in this case. Indeed, estimation of such model by PCs simply requires to center data first, i.e., to work with $x_{it}-\bar {x}_i$ where $\bar x_i=T^{-1}\sum_{t=1}^T x_{it}$ is clearly a consistent estimator of $\alpha_i$. As for QML estimation, it is straightforward to see that $\bar x_i$ is precisely the QML estimator of $\alpha_i$ and thus it is enough to work with the likelihood of the centered data \citep[Section 2]{baili12}.
 }
\end{rem}

\section{Principal Component Analysis}\label{sec:PC}

The PC estimators of $\bm\Lambda$ and $\bm F$ are the solutions of the following minimization:
\begin{align}
\min_{\underline{\bm\Lambda},\underline{\bm F}} \frac 1{nT} \text{tr}\l\{
\l(\bm X-\underline{\bm F}\,\underline{\bm\Lambda}^\prime \r)
\l(\bm X-\underline{\bm F}\,\underline{\bm\Lambda}^\prime \r)^\prime
\r\}=\min_{\underline{\bm\Lambda},\underline{\bm F}} \frac 1{nT} \text{tr}\l\{
\l(\bm X-\underline{\bm F}\,\underline{\bm\Lambda}^\prime \r)^\prime
\l(\bm X-\underline{\bm F}\,\underline{\bm\Lambda}^\prime \r)
\r\},\label{eq:minimizza}
\end{align}
where $\underline{\bm\Lambda}$ and $\underline{\bm F}$  indicate generic values of the loadings and the factors, respectively, and it is intended that they also satisfy Assumptions \ref{ass:common}, \ref{ass:eval}, and \ref{ass:ident}. Given the identifying constraints in Assumption \ref{ass:ident}, the solution to \eqref{eq:minimizza} can be found in two steps. There are two equivalent ways to do that:
\begin{inparaenum}
\item [(A)] based on the $n\times n$ matrix $\bm X^\prime\bm X$ solve first for loadings and then get the factors by projecting $\bm X$ onto the estimated loadings;
\item [(B)] based on the $T\times T$ matrix $\bm X\bm X^\prime$ solve first for the factors and then get the loadings by projecting $\bm X$ onto the estimated factors.
\end{inparaenum}

Although the majority of the literature on factor models considers approach B, thus estimating the factors as normalized eigenvectors, and it derives the theory accordingly  \citep[see, e.g.,][]{Bai03}, in the rest of the paper we follow approach A. The main reason for this choice is that approach A does not require to estimate the factors first, which is convenient given that our focus  is on estimation of the loadings. Notice that approach A is also the classical one (see, e.g., \citealp[Chapter 4]{lawleymaxwell71}, \citealp[Chapter 9.3]{mardia1979multivariate}, \citealp[Chapter 7.2]{jolliffe2002principal}). 
Nevertheless, which approach to choose is just a matter of taste and it has no theoretical or practical implications. Indeed, numerically both approaches give the same results and all the following theory can be equivalently derived under approach B. Incidentally, by choosing approach A, we also contribute to PC literature with new proofs, alternative to those in \citet{Bai03}. 

More in detail, consider the $n\times n$ sample covariance matrix (recall that $\E[\mbf x_t]=\mbf 0_n$  by assumption) $\wh{\bm\Gamma}^x= T^{-1}{\bm X^\prime\bm X}$, having  its $r$ largest eigenvalues collected in the $r\times r$ diagonal matrix $\wh{\mbf M}^x$ (sorted in descending order) with corresponding normalized eigenvectors as columns of  the $n\times r$ matrix $\wh{\mbf V}^x$.
For any given $\underline{\bm\Lambda}$ the solution of \eqref{eq:minimizza} for $\underline{\bm F}$ is just the linear projection $\underline{\bm F}=\bm X\underline{\bm\Lambda}(\underline{\bm\Lambda}^\prime\underline{\bm\Lambda})^{-1}$. By substituting this expression in \eqref{eq:minimizza} we have
\begin{align}
\min_{\underline{\bm\Lambda}}\frac 1{nT} \text{tr}\l\{
\bm X\l(\mbf I_n-\underline{\bm\Lambda}(\underline{\bm\Lambda}^\prime\underline{\bm\Lambda})^{-1}\underline{\bm\Lambda}^\prime \r)\bm X^\prime
\r\}&=
\max_{\underline{\bm\Lambda}}\frac 1{nT} \text{tr}\l\{
\bm X^\prime\bm X \underline{\bm\Lambda}(\underline{\bm\Lambda}^\prime\underline{\bm\Lambda})^{-1}\underline{\bm\Lambda}^\prime 
\r\}\nn\\
&=\max_{\underline{\bm\Lambda}}\frac 1{n} \text{tr}
\l\{
(\underline{\bm\Lambda}^\prime\underline{\bm\Lambda})^{-1/2}
\underline{\bm\Lambda}^\prime
\frac{\bm X^\prime\bm X}{T}
\underline{\bm\Lambda}
(\underline{\bm\Lambda}^\prime\underline{\bm\Lambda})^{-1/2}
\r\}
\end{align}
Now since by construction each column of $ \underline{\bm\Lambda}
(\underline{\bm\Lambda}^\prime\underline{\bm\Lambda})^{-1/2}$ is normalized (since we assumed ${\underline{\bm\Lambda}^\prime\underline{\bm\Lambda}}$ to be diagonal), then the above maximizaton,  once solved, should return the $r$ largest eigenvalues of  $\wh{\bm\Gamma}^x$ divided by $n$, i.e., it must give $n^{-1}{\wh{\mbf M}^x}$. In other words, our estimator $\wh{\bm\Lambda}$ must be such that $\wh{\bm\Lambda}
(\wh{\bm\Lambda}^\prime\wh{\bm\Lambda})^{-1/2}$ is the matrix of normalized eigenvectors corresponding the $r$ largest eigenvalues of $(nT)^{-1}{\bm X^\prime\bm X}$, i.e., such that:
\beq\label{eq:evecXX}
(\wh{\bm\Lambda}^\prime\wh{\bm\Lambda})^{-1/2}
\wh{\bm\Lambda}^\prime
\frac{\bm X^\prime\bm X}{nT}
\wh{\bm\Lambda}
(\wh{\bm\Lambda}^\prime\wh{\bm\Lambda})^{-1/2}=\frac{\wh{\mbf M}^x}{n},
\eeq
but also, by definition of eigenvectors, 
\beq\label{eq:evecXX2}
\wh{\mbf V}^{x\prime}
\frac{\bm X^\prime\bm X}{nT}
\wh{\mbf V}^x =\frac{\wh{\mbf M}^x}{n}.
\eeq
Therefore, since we must have $\bm\Lambda^\prime\bm\Lambda$ diagonal with distinct entries by Assumption \ref{ass:ident}(a), from \eqref{eq:evecXX} and \eqref{eq:evecXX2},
\beq\label{eq:estL}
\wh{\bm\Lambda}=\wh{\mbf V}^x(\wh{\mbf M}^x)^{1/2},
\eeq
which is such that $n^{-1}{\wh{\bm\Lambda}^\prime \wh{\bm\Lambda}}=n^{-1}{\wh{\mbf M}^x}$ is diagonal. The factors are then estimated as the linear projections: $\wh{\bm F}=\bm X\wh{\bm\Lambda}(\wh{\bm\Lambda}^\prime\wh{\bm\Lambda})^{-1}=\bm X\wh{\mbf V}^x(\wh{\mbf M}^x)^{-1/2}$, which are the normalized PCs of $\bm X$, such that $T^{-1}{\wh{\bm F}^\prime\wh{\bm F}}=\mbf I_r$. The latter, however, are not needed in the following.

Letting $\wh{\bm\lambda}_i^\prime$ be the $i$th row of the PC estimator of the loadings in \eqref{eq:estL}, we have the following asymptotic results.

\begin{theorem}[]\label{th:CLTL}
Under Assumptions \ref{ass:common} through \ref{ass:CLT}, as  $n,T\to\infty$, 
\begin{compactenum}[(a)]
\item $\min(n,\sqrt T)\Vert\wh{\bm\lambda}_i-{\bm\lambda}_i\Vert=O_{\mathrm P}(1)$, for any given $i=1,\ldots,n$, and, if $n^{-1}\sqrt T\to0$, 
$
\sqrt T(\wh{\bm\lambda}_i-{\bm\lambda}_i) 
\to_d\mathcal N\l(\mbf 0_r, \bm\Phi_i\r),
$
where $\bm\Phi_i$ is defined in Assumption \ref{ass:CLT};
\item $\min(n,\sqrt T)\Vert n^{-1/2}({\wh{\bm\Lambda}-\bm\Lambda})\Vert=O_{\mathrm P}(1)$.
\end{compactenum}
\end{theorem}


An important consequence of Theorem \ref{th:CLTL} is that the PC estimator is asymptotically equivalent to the unfeasible OLS estimator, which we would obtain if the factors were observed, and denoted as $\bm\Lambda^{\text{\tiny OLS}}$, with $i$th row given by ${\bm\lambda}_i^{\text{\tiny \upshape OLS}\prime}$.

\begin{cor}[]\label{cor:PCOLS}
Under Assumptions \ref{ass:common} through \ref{ass:CLT}, as $n,T\to\infty$,
\begin{inparaenum}[(a)]
\item [\upshape (a)]
$
\min(n,\sqrt{nT})\Vert\wh{\bm\lambda}_i - {\bm\lambda}_i^{\text{\tiny \upshape OLS}}\Vert =  O_{\mathrm P}\l(1\r)
$, for any given $i=1,\ldots,n$;
\item [\upshape (b)] $\min(n,\sqrt {nT})\Vert n^{-1/2}({\wh{\bm\Lambda}-{\bm\Lambda}^{\text{\tiny \upshape OLS}}})\Vert=O_{\mathrm P}(1)$.
\end{inparaenum}
\end{cor}


\section{Quasi Maximum Likelihood}\label{sec:PCQML}

Define $\bm {\mathcal X}=\text{vec}(\bm X^\prime)=(\mbf x_{1}^\prime\cdots \mbf x_{T}^\prime)^\prime$ and $\bm{\mathcal Z}=\text{vec}(\bm \Xi^\prime)=(\bm\xi_{1}^\prime\cdots \bm\xi_{T}^\prime)^\prime$  as 
the $nT$-dimensional vectors of observations and idiosyncratic components,  
$\bm {\mathcal F}=\text{vec}(\bm F^\prime)=(\mbf F_{1}^\prime\cdots \mbf F_T^\prime)^\prime$ as the $rT$-dimensional vector of factors, and  $\bm {\mathfrak L}=\mbf I_T\otimes \bm\Lambda$ as the $nT\times rT$ matrix containing all factor loadings replicated $T$ times. Then, by vectorizing the transposed of \eqref{eq:SDFM1Rmat} we have:
\begin{align}\label{eq:DFM_matrix}
\bm{\mathcal X}= \bm {\mathfrak L} \bm{\mathcal F} + \bm{\mathcal Z}.
\end{align}
Let $\bm\Omega^x=\E_{}[\bm {\mathcal X}\bm {\mathcal X}^\prime]$ and $\bm\Omega^\xi=\E_{}[\bm{\mathcal Z}\bm{\mathcal Z}^\prime]$, be the $nT\times nT$ covariance matrices of the $nT$-dimensional vectors of data and idiosyncratic  components, respectively. Let also $\bm\Omega^F=  \E_{}[\bm {\mathcal F}\bm {\mathcal F}^{\prime}]$ be the $rT\times rT$ covariance matrix of the $rT$-dimensional factor vector. Then, because of Assumption \ref{ass:ind},
\[
\bm\Omega^x = \bm {\mathfrak L}\,\bm\Omega^F\bm {\mathfrak L}^\prime + \bm\Omega^\xi.
\]
In principle, the parameters that need to be estimated are then given by the vector
$$
\bm\varphi=(\mathrm{vec}(\bm\Lambda)^\prime, \mathrm{vech}(\bm\Omega^{\xi})^\prime,\mathrm{vech}(\bm\Omega^{F})^\prime)^\prime,
$$ and
the Gaussian quasi-log-likelihood computed in a generic value of the parameters, denoted as $\underline{\bm{\varphi}}$, is given by (omitting the constant term for simplicity)
\begin{align}
\ell^*(\bm {\mathcal X};\underline{\bm{\varphi}})
&=-\frac 12 \log\det\l(\underline{\bm\Omega}^x\r) -\frac 12  \bm {\mathcal X}^\prime \l(\underline{\bm\Omega}^x\r)^{-1}\bm {\mathcal X}\label{eq:LL}.
\end{align}
In general, maximization of \eqref{eq:LL} is an unfeasible task since the parameters vector ${\bm{\varphi}}$ to be estimated has $\simeq (nT)^2$ elements. The common practice is then to consider simpler mis-specified log-likelihoods which depend on fewer parameters \citep{baili12,baili16}.  

First of all, consistently with the fact that we do not assume any parametric model describing the dynamics of the factors, we can consider a simpler log-likelihood with $\bm\Omega^F=\mbf I_{nr}$, where we imposed also the identification constraint $\bm\Gamma^F=\mbf I_r$ implied by Assumption \ref{ass:ident}(b).
A second simplification consists in considering a mis-specified log-likelihood of an approximate factor model where also the idiosyncratic components are treated as serially uncorrelated. As a result of these two mis-specifications the log-likelihood \eqref{eq:LL} is reduced to:
\beq
\ell(\bm{\mathcal X};\underline{\bm\varphi})= -\frac T 2 \log\det (\underline{\bm\Lambda}\,\underline{\bm\Lambda}^\prime+\underline{\bm\Gamma}^\xi)-\frac 12\sum_{t=1}^T \mbf x_t^\prime(\underline{\bm\Lambda}\,\underline{\bm\Lambda}^\prime+\underline{\bm\Gamma}^\xi)^{-1}\mbf x_t.\label{eq:LL0}
\eeq
 The parameters to be estimated are then reduced:
$\bm\varphi=(\mathrm{vec}(\bm\Lambda)^\prime, \mathrm{vech}(\bm\Gamma^\xi)^\prime)^\prime$.
Nevertheless, estimation of $\bm\varphi$ by means of maximization of \eqref{eq:LL0} seems still hopeless since in general ${\bm\Gamma}^\xi$ contains $n(n+1)/2$ distinct elements. 

Some further mis-specification of the log-likelihood \eqref{eq:LL0}, based on regularizing $\bm\Gamma^\xi$, is then usually introduced in order to reduce the number of parameters to be estimated. A possibility in this sense is explored, for example, by \citet{bailiao16} who propose to maximize \eqref{eq:LL0} subject to an $\ell_1$ penalty imposed on the off-diagonal entries of $\bm\Gamma^\xi$. This approach forces sparsity, thus reducing the number of parameters to be estimated, but at the same  time it makes estimation dependent on the chosen penalization level, which affects also the rate of consistency. 

An even simpler approach consists in estimating only the diagonal entries of $\bm\Gamma^\xi$. Specifically, letting $\bm\Sigma^\xi=\text{dg}(\sigma_1^2\cdots\sigma_n^2)$ be the diagonal matrix with entries the diagonal entries of $\bm\Gamma^\xi$, we can focus on maximization of  the further mis-specified log-likelihood:
\beq
\ell_{\text{\tiny E}}(\bm{\mathcal X};\underline{\bm\varphi})= -\frac T 2 \log\det (\underline{\bm\Lambda}\,\underline{\bm\Lambda}^\prime+\underline{\bm\Sigma}^\xi)-\frac 12\sum_{t=1}^T \mbf x_t^\prime(\underline{\bm\Lambda}\,\underline{\bm\Lambda}^\prime+\underline{\bm\Sigma}^\xi)^{-1}\mbf x_t.\label{eq:LL00}
\eeq
The parameters to be estimated are then reduced to: $\bm\varphi=(\mathrm{vec}(\bm\Lambda)^\prime, \sigma_1^2,\cdots, \sigma_n^2)^\prime$, which are just $nr+n$. This is a feasible task given that we have $nT$ data points.

The log-likelihood \eqref{eq:LL00} is the one considered in classical factor analysis, where, however, $n$ is assumed to be fixed and small, and, moreover, the true idiosyncratic covariance matrix is assumed to be diagonal, i.e., the likelihood is not mis-specified (see, e.g., \citealp[Chapter 2]{lawleymaxwell71}, and \citealp{RT82}). In the high-dimensional case, i.e., when we allow $n\to\infty$, maximization of \eqref{eq:LL00} has been studied by \citet{baili12,baili16} under a variety of possible identifying constraints. In particular, while \citet{baili12} consider the case of no idiosyncratic serial or cross-correlation, i.e., they assume $\bm\Omega^\xi=\mbf I_T\otimes \bm\Sigma^\xi$, which is diagonal, thus considering \eqref{eq:LL00} as a correctly specified log-likelihood,  \citet{baili16} allow instead for idiosyncratic serial and cross-sectional correlations, as we do in this paper, thus treating \eqref{eq:LL00} as a mis-specified log-likelihood. 

In the rest of this section we briefly review the properties of the QML estimator of the loadings. We refer to \citet{MBQML} for a full review of QML estimation of factor models.
We start with the simplest case in which we consider a further mis-specification of the log-likelihood \eqref{eq:LL00}, where we treat the idiosyncratic components as if they were homoskedastic, thus in the log-likelihood we replace $\bm\Sigma^\xi$ with an even simpler covariance matrix ${\sigma}^2\mbf I_n$, with ${\sigma}^2>0$ and finite.
In this case \citet{tippingbishop99} prove that the QML estimator of the loadings matrix $\bm\Lambda$, has a closed form given by
\beq\label{eq:TB99}
\wh{\bm\Lambda}^{\text{\tiny QML,E}_0} = \wh{\mbf V}^x \l(\wh{\mbf M}^x-\wh{\sigma}^{2\text{\tiny QML,E}_0}\mbf I_r\r)^{1/2}=
 \wh{\mbf V}^x \l(\wh{\mbf M}^x-\l\{\frac 1{n-r}\sum_{j=r+1}^n\wh{\mu}^x_j\r\}\mbf I_r\r)^{1/2}
,
\eeq
with $\wh{\sigma}^{2\text{\tiny QML,E}_0}$ being the QML estimator of $\sigma^2$. Intuitively, since under our assumptions we should have $\wh{\mbf M}^x=O_{\mathrm P}(n)$, as $n\to\infty$, then  $\wh{\bm\Lambda}^{\text{\tiny QML,E}_0}$  seems to coincide asymptotically  with the PC estimator given in \eqref{eq:estL}. This is a well known fact and it is often quoted in the literature (see, e.g., \citealp{DGRqml}): in the case of spherical idiosyncratic components the PC and QML estimators are asymptotically equivalent. However, to the best of our knowledge no formal proof exists, at least under the present set of assumptions. In fact, the proof would essentially require to prove that the estimated $(r+1)$th largest eigenvalue is such that $\wh{\mu}^x_{r+1}=O_{\mathrm P}(1)$, which is not an easy task in a high-dimensional setting, because, although we know that under our assumptions ${\mu}^x_{r+1}=O(1)$ (see Lemma \ref{lem:Gxi}(v)), in general $\wh{\mu}^x_{r+1}$ is not a consistent estimator of ${\mu}^x_{r+1}$ (see, e.g., \citealp[Lemma 2.2]{trapani2018randomized}).

Things become more complicated if we allow for heteroskedasticity. Indeed, no closed form solution exists for the QML estimator maximizing the mis-specified log-likelihood \eqref{eq:LL00} and numerical maximization is required instead (\citealp[see, e.g.,][Section 8]{baili12}, for a proposed algorithm). Still it is possible to derive its asymptotic properties.
Let us denote as $\wh{\bm\lambda}_i^{\text{\tiny \upshape QML,E}}$, $i=1,\ldots, n$,  the QML estimator of the $i$th row of $\bm\Lambda$.  In the classical fixed $n$ case, such estimator retains the classical propertied of the QML estimators, so it is $\sqrt T$-consistent and asymptotically normal. However, since the first and second derivatives of  \eqref{eq:LL00} are very complex, the asymptotic covariance matrix has also a very complicated form, which in turn makes its estimation not at all easy (see, e.g., \citealp[Theorem 12.3]{AR56}, \citealp[Theorem 2F]{AFP87}, and \citealp[Theorems 1, 2, and 3]{AA88}). If, instead, we study the properties of the QML estimator when allowing also for $n\to\infty$, things become, perhaps surprisingly, simpler. This is shown in the following theorem proved by \citet[Theorem 1]{baili16}.

\begin{theorem}[]\label{th:CLTLBAI} 
Assume $n^{-1}{\bm\Lambda^\prime(\bm\Sigma^\xi)^{-1}\bm\Lambda}$ to be diagonal  for all $n\in\mathbb N$ and $T^{-1}{\bm F^\prime\bm F}=\mbf I_r$ for all $T\in\mathbb N$.
Then, under Assumptions \ref{ass:common}, \ref{ass:idio}, \ref{ass:ind}, and \ref{ass:CLT}, if $n^{-1}\sqrt T\to0$, as $n,T\to\infty$, for any given $i=1,\ldots,n$,
$
\sqrt T(\wh{\bm\lambda}_i^{\text{\tiny \upshape QML,E}} -\bm\lambda_i)\to_d\mathcal N\l(\mbf 0_r, \bm\Phi_i\r),
$
where $\bm\Phi_i$ is defined in Assumption \ref{ass:CLT}.
\end{theorem}

The proof of this result is based on asymptotic expansions of a set of conditions derived from first order conditions computed for the log-likelihood \eqref{eq:LL00}. It is a very long proof, aimed at showing that $\wh{\bm\lambda}_i^{\text{\tiny \upshape QML,E}}$  is asymptotically equivalent to the unfeasible OLS we would obtain if we knew the factors, which in turn has a known asymptotic distribution  \citep[see the online supplement of][]{baili16}. The proof requires a series of technical assumptions different from, but nested into, ours (see Remark \ref{rem:baili}). However, as noticed already in Section \ref{sec:sdfm}, it is important to stress that Theorem \ref{th:CLTLBAI} holds under an identification constraint  for the loadings which differs from our Assumption \ref{ass:ident}(a). At the end of the next section we prove that this theorem holds also under our assumptions (see Theorem \ref{cor:QMLcons} and Corollary \ref{cor:sett}).

\begin{rem}\label{rem:baili12}
\upshape{
Notice that the fact that we consider a mis-specified log-likelihood with serially and cross-sectionally uncorrelated idiosyncratic components, does not affect consistency of the estimated loadings, but only their asymptotic covariance. And in particular, notice that what matters for the asymptotic covariance is just the idiosyncratic serial correlation. Indeed, if the idiosyncratic components were cross-sectionally uncorrelated but serially correlated the asymptotic covariance in Theorem \ref{th:CLTLBAI} would be the same, while if they were serially uncorrelated the asymptotic covariance would be $ \sigma_i^2 \mbf I_r$ (due to the identification $\bm\Gamma^F=\mbf I_r$), regardless of the presence or not of cross-sectional correlation. This  result, which is a special case of Theorem \ref{th:CLTLBAI}, does not require any constraint between the rates of divergence of $n$ and $T$ and it is proved by \citet[Theorem 5.2]{baili12}. }
\end{rem}

\begin{rem}\label{rem:DGR}
\upshape{
The case of QML estimation when $\bm\Omega^F$ depends explicitly on additional parameters capturing the autocorrelations in the factors is considered in \citet{DGRqml}. However, in that case QML estimation requires the use of the EM algorithm jointly with the Kalman smoother and the results of this paper do not directly apply unless we first prove convergence of the EM to the same QML estimator considered here.
%
}
\end{rem}

\section{The PC and QML estimators are asymptotically equivalent}\label{sec:PCQML0}
Given the discussion in the previous section, we might argue that the PC and QML estimators have the same asymptotic properties. However, Theorems \ref{th:CLTL} and \ref{th:CLTLBAI} are derived under a similar but different sets of identifying assumptions and  the two results cannot be directly compared.  Even in the spherical case the proof seems to be not so easy due to the unknown properties of the smallest $N-r$ sample eigenvalues.

It is then natural to ask the following question. Can we prove in a simple way that the PC and QML estimators are asymptotically equivalent under 
the same set of assumptions given in this paper, which are the standard PC assumptions? Moreover, if we prove such equivalence and 
since the PC estimator of the loadings does not depend on the idiosyncratic covariance matrix, it is natural to ask also the following additional question. Can we prove the asymptotic equivalence of the two estimators when considering QML based on the log-likelihood \eqref{eq:LL0}  of an approximate factor model, i.e., without constraining the idiosyncratic covariance to be diagonal or even homoskedastic? In this section we answer both questions.

Denote as $\wh{\bm\Lambda}^{\text{\tiny QML}}$ the QML estimator of the loadings matrix maximizing the mis-specified log-likleihood \eqref{eq:LL0}, where the idiosyncratic components are treated as serially uncorrelated, but their covariance matrix is unrestricted so it is correctly specified.
Let also $\wh{\bm\lambda}_i^{\text{\tiny \upshape QML}\prime}$, $i=1,\ldots, n$, be the $i$th row of $\wh{\bm\Lambda}^{\text{\tiny QML}}$. Then, we state our main result.

\begin{theorem}\label{th:QML}
Under Assumptions \ref{ass:common} through \ref{ass:sign} and assuming also that $\bm\Gamma^\xi$ is positive definite, as $n,T\to\infty$, 
\begin{inparaenum}[(a)]
\item [\upshape (a)] $n\Vert
n^{-1/2}({\wh{\bm\Lambda}^{\text{\tiny \upshape QML}}-\wh{\bm\Lambda}})\Vert= O_{\mathrm {P}}(1)$;
\item [\upshape (b)]
$n \Vert\wh{\bm\lambda}_i^{\text{\tiny \upshape QML}}-\wh{\bm\lambda}_i \Vert = O_{\mathrm {P}}(1)$, for any given $i=1,\ldots,n$.
\end{inparaenum}
\end{theorem}

Consistency and asymptotic normality of the QML estimator of the loadings maximizing the log-likelihood \eqref{eq:LL0} immediately follow.

\begin{theorem}\label{cor:QMLcons}
Under Assumptions \ref{ass:common} through \ref{ass:sign} and assuming also that $\bm\Gamma^\xi$ is positive definite, as $n,T\to\infty$:
\begin{compactenum}[(a)]
\item 
$\min(n,\sqrt T)\Vert\wh{\bm\lambda}_i^{\text{\tiny \upshape QML}}-{\bm\lambda}_i\Vert = O_{\mathrm {P}}(1)$, for any given $i=1,\ldots,n$,, and, if $n^{-1}\sqrt T\to 0$
$
\sqrt T(\wh{\bm\lambda}_i^{\text{\tiny \upshape QML}}-\bm\lambda_i)\to_d\mathcal N\l(\mbf 0_r, \bm\Phi_i\r),
$
where $\bm\Phi_i$ is defined in Assumption \ref{ass:CLT};
\item $\min(n,\sqrt T)\Vert n^{-1/2}({\wh{\bm\Lambda}^{\text{\tiny \upshape QML}}-\bm\Lambda})\Vert=O_{\mathrm P}(1)$.
\end{compactenum}
\end{theorem}



The results in Theorems \ref{th:QML} and \ref{cor:QMLcons} hold when considering the log-likelihood \eqref{eq:LL0} which depends on any generic idiosyncratic covariance matrix provided that it is positive definite and Assumption \ref{ass:idio} is satisfied. Therefore, from Theorem \ref{th:QML} it follows that if we replace in the log-likelihood the full idiosyncratic covariance with a diagonal covariance matrix as in \citet{baili12,baili16} or if we also impose homoskedasticity as in \citet{tippingbishop99}, we still get QML estimators that are asymptotically equivalent to the PC estimator.

\begin{cor}\label{cor:sett}
Under Assumptions \ref{ass:common} through \ref{ass:sign}, as $n,T\to\infty$, 
\begin{inparaenum}[(a)]
\item [\upshape (a)] $n\Vert \wh{\bm\lambda}_i - \wh{\bm\lambda}_i^{\text{\tiny \upshape QML,E}}\Vert=O_{\mathrm P}(1)$;
\item [\upshape (b)] $n\Vert \wh{\bm\lambda}_i - \wh{\bm\lambda}_i^{\text{\tiny \upshape QML,E}_0}\Vert=O_{\mathrm P}(1)$.
\end{inparaenum}
\end{cor}

From these results and Theorem \ref{cor:QMLcons}, we directly have a proof of Theorem \ref{th:CLTLBAI} by \citet{baili16}, which now holds under our identifying constraint in Assumption \ref{ass:ident}(a), and we also prove the often quoted statement that  the PC and QML estimators are asymptotically equivalent under sphericity of idiosyncratic components.

\begin{rem}\label{rem:lobuttiamo}
\upshape{
The appealing feature of the log-likelihood \eqref{eq:LL0}, is that it does not depend on the factors, and, under the identifying constraint $T^{-1}{\bm F'\bm F}=\mbf I_r$ for all $T\in\mathbb N$, it does not depend on the second moments of the factors either. Still, if we follow the classical approach to consider the vector of factors  $\bm{\mathcal F}$ as an $rT$-dimensional sequence of deterministic constants \citep{lawleymaxwell71}, then a log-likelihood alternative to \eqref{eq:LL00} can be considered, namely:
\beq\label{eq:ZLL}
\ell_{\text{\tiny E}}(\bm{\mathcal X};\underline{\bm\varphi},\underline{\bm{\mathcal F}})= -\frac T2\log\det(\underline{\bm\Gamma}^\xi)-\frac 12\sum_{t=1}^T
(\mbf x_t-\underline{\bm\Lambda}\,\underline{\mbf F}_t)^\prime
(\underline{\bm\Gamma}^\xi)^{-1}
(\mbf x_t-\underline{\bm\Lambda}\,\underline{\mbf F}_t).
\eeq
It is straightforward to see that, for given factors, the loadings maximizing \eqref{eq:ZLL} are given by their OLS estimator, while, for given loadings, the factors maximizing \eqref{eq:ZLL} are given by their GLS estimator. Hence, full maximization of \eqref{eq:ZLL} requires knowing, or estimating, the factors too. Although seemingly simpler this approach presents at least two major drawbacks with respect to the approach followed in the proof of Theorem \ref{th:QML} (see also the comments in 
\citealp[p.440]{baili12}, and \citealp[pp.129-130]{AR56}).

First, by iterating, between OLS and GLS we can think of finding a solution. This approach is similar to the one recently considered by \citet{pz23}, but the convergence properties of such algorithm are not proved nor discussed, so it is unclear how this estimator is related to the maximizer of the full likelihood \eqref{eq:LL0}. Moreover, this strategy would require a positive definite estimator of $\bm\Gamma^\xi$ in order to compute the GLS estimator of the factors, a hard task when $n$ is large. This, in general, requires again mis-specifying or regularizing the log-likelihood \eqref{eq:ZLL}, e.g., by replacing $\bm\Gamma^\xi$ with the diagonal $\bm\Sigma^\xi$. The asymptotic properties of the estimator of the loadings will then explicitly  depend on the properties of an estimator of the idiosyncratic covariance, or at least of its diagonal elements. As shown in Theorem \ref{cor:QMLcons}, this is not the case for the QML estimator maximizing the log-likelihood \eqref{eq:LL0}.


Second, if the factors are treated as random variables, as, e.g., in the popular Factor Augmented VAR models \citep{BBE05}, then they cannot be considered as constant parameters, which means that \eqref{eq:ZLL} is not the full log-likelihood of the data $\bm{\mathcal X}$ but it is just the conditional log-likelihood of $\bm{\mathcal X}$ given the factors, so, in principle, not all information is used to estimate the loadings. Indeed, if the factors are random variables then the log-likelihood \eqref{eq:LL00} is decomposed as
\beq
\ell_{\text{\tiny E}}(\bm{\mathcal X};\underline{\bm\varphi})=\ell_{\text{\tiny E}}(\bm{\mathcal X}|\bm{\mathcal F};\underline{\bm\varphi}) +\ell_{\text{\tiny E}}(\bm{\mathcal F};\underline{\bm\varphi})-\ell_{\text{\tiny E}}(\bm{\mathcal F}|\bm{\mathcal X};\underline{\bm\varphi}),\label{eq:bayesiid}
\eeq 
where $\ell_{\text{\tiny E}}(\bm{\mathcal X}|\bm{\mathcal F};\underline{\bm\varphi})$ coincides with \eqref{eq:ZLL} but  it does not coincide with $\ell_{\text{\tiny E}}(\bm{\mathcal X};\underline{\bm\varphi})$ anymore. 

This last point has  both a theoretical and an applied implication. From the theory point of view, to directly show from \eqref{eq:bayesiid} that the QML estimator of the loadings maximizing $\ell_{\text{\tiny E}}(\bm{\mathcal X};\underline{\bm\varphi})$ is asymptotically equivalent to the unfeasible OLS maximizing $\ell_{\text{\tiny E}}(\bm{\mathcal X}|\bm{\mathcal F};\underline{\bm\varphi})$ would require showing that $\ell_{\text{\tiny E}}(\bm{\mathcal F};\underline{\bm\varphi})$ and $\ell_{\text{\tiny E}}(\bm{\mathcal F}|\bm{\mathcal X};\underline{\bm\varphi})$ are asymptotically negligible. This is the argument sketched by \citet{BT11}, but to make it formal is not an easy task. Consider the simplest case in which the factors are treated as serially uncorrelated, then, while $\ell_{\text{\tiny E}}(\bm{\mathcal F};\underline{\bm\varphi})$ does not depend on any parameter and can be discarded, the expression of $\ell_{\text{\tiny E}}(\bm{\mathcal F}|\bm{\mathcal X};\underline{\bm\varphi})$ will depend on the conditional moments (mean and covariance) of the factors given $\bm{\mathcal X}$. These in turn have simple expressions only if we are willing to assume Gaussianity, in which case the conditional mean is just a linear projection. Otherwise computation of those moments is not straightforward. The proof of Theorem \ref{th:QML} relies instead only on $\ell_{\text{\tiny E}}(\bm{\mathcal X};\underline{\bm\varphi})$, so, as noticed above, it does not require knowing the factors or their conditional moments.

Finally, from a practical point of view, we could use the right-hand-side of \eqref{eq:bayesiid} to compute the QML estimator by means of the Expectation Maximization (EM) algorithm, which simplifies our task since it allows us to discard $\ell_{\text{\tiny E}}(\bm{\mathcal F}|\bm{\mathcal X};\underline{\bm\varphi})$ \citep{wu83}. However, once again we would still have to compute the conditional moments of the factors when in the E-step we need compute the conditional expectation of $\ell_{\text{\tiny E}}(\bm{\mathcal X}|\bm{\mathcal F};\underline{\bm\varphi}) +\ell_{\text{\tiny E}}(\bm{\mathcal F};\underline{\bm\varphi})$ given $\bm{\mathcal X}$. 
This introduces a correction in the estimation of the loadings and the resulting estimator is not given by a simple OLS anymore. 

}
\end{rem}

\section{Asymptotic covariance matrices of the QML and PC estimators}\label{sec:fisher}

In this section,  we focus on the mis-specified log-likelihood \eqref{eq:LL00}, which is commonly used in empirical work \citep{baili12,baili16}. For such log-likelihood, denote the Fisher information and the population Hessian matrices for $\bm\lambda_i$, $i=1,\ldots, n$, as:
\begin{align}
\bm{\mathcal I}_i(\bm{\mathcal X};\bm\varphi)&=\lim_{T\to\infty} \E\l[
\l(\frac 1{\sqrt T}
\l.\frac{\partial \ell_{\text{\tiny E}}(\bm{\mathcal X};\underline{\bm\varphi})}{\partial \underline{\bm\lambda}_i^\prime}\r\vert_{\underline{\bm\varphi}={{\bm\varphi}}}
\r)
\l(\frac 1{\sqrt T}\l.\frac{\partial \ell_{\text{\tiny E}}(\bm{\mathcal X};\underline{\bm\varphi})}{\partial \underline{\bm\lambda}_i}\r\vert_{\underline{\bm\varphi}={{\bm\varphi}}}
\r)\,
\r],\nn\\
\bm{\mathcal H}_{i}(\bm{\mathcal X};\bm\varphi)&=\lim_{T\to\infty} \E\l[
\frac 1T  \l.\frac{\partial^2 \ell_{\text{\tiny E}}(\bm{\mathcal X};\underline{\bm\varphi})}{\partial \underline{\bm\lambda}_i^\prime\partial \underline{\bm\lambda}_i}\r\vert_{\underline{\bm\varphi}={{\bm\varphi}}}
\r].\nn
\end{align}
From QML theory the asymptotic covariance of the QML estimator should be given by:
\beq\label{eq:QMLADVar}
\text{A$\!\mathbb V$ar}_0(\sqrt T \wh{\bm\lambda}_i^{\text{\tiny QML,E}})= \{\bm{\mathcal H}_i(\bm{\mathcal X};\bm\varphi)\}^{-1}\bm{\mathcal I}_i(\bm{\mathcal X};\bm\varphi)\{\bm{\mathcal H}_i(\bm{\mathcal X};\bm\varphi)\}^{-1}.
\eeq
This would be the matrix to estimate if we were to conduct QML based inference on the loadings. 

In general, estimation of \eqref{eq:QMLADVar} is very difficult given the complex expressions of the Hessian and Fisher information matrices (see also (A.15) and (A.42) in the Supplementary Material). Moreover, from Theorem \ref{cor:QMLcons} we know that, in fact, if $n^{-1}\sqrt T\to 0$, as $n,T\to\infty$, then the asymptotic covariance of the QML estimator is
\beq\label{eq:QMLADVar1}
\text{A$\!\mathbb V$ar}_1(\sqrt T \wh{\bm\lambda}_i^{\text{\tiny QML,E}})= \bm \Phi_i=\lim_{T\to\infty}\frac 1{T}\sum_{s,t=1}^T\E\l[\mbf F_t\mbf F_s^\prime \xi_{it}\xi_{is}\r],
\eeq
where we used the definition of $\bm\Phi_i$ in Assumption \ref{ass:CLT}. So what is the relation between the asymptotic covariances \eqref{eq:QMLADVar} and \eqref{eq:QMLADVar1}?

First, notice that \eqref{eq:QMLADVar1} coincides with the asymptotic covariance of the unfeasible OLS estimator, which, in turn, is the QML estimator maximizing the log-likelihood for an exact factor model conditional on observing the factors, i.e., 
\beq
\ell_{\text{\tiny E}}(\bm{\mathcal X}|\bm{\mathcal F};\underline{\bm\varphi})=-\frac T2\log\det(\underline{\bm\Sigma}^\xi)-\frac 12\sum_{t=1}^T (\mbf x_t-\underline{\bm\Lambda}\mbf F_t)^\prime (\underline{\bm\Sigma}^\xi)^{-1}(\mbf x_t-\underline{\bm\Lambda}\mbf F_t).\label{eq:LLXF00}
\eeq
It is easy to compute the Fisher information and the population Hessian matrices of \eqref{eq:LLXF00} for $\bm\lambda_i$: 
(see also (A.37) and (A.47) in the Supplementary Material)
\begin{align}
\bm{\mathcal I}_i(\bm{\mathcal X}|\bm{\mathcal F};\bm\varphi)&=\lim_{T\to\infty} \E\l[
\l(\frac 1{\sqrt T}
\l.\frac{\partial \ell_{\text{\tiny E}}(\bm{\mathcal X}|\bm{\mathcal F};\underline{\bm\varphi})}{\partial \underline{\bm\lambda}_i^\prime}\r\vert_{\underline{\bm\varphi}={{\bm\varphi}}}
\r)
\l(\frac 1{\sqrt T}\l.\frac{\partial \ell_{\text{\tiny E}}(\bm{\mathcal X}|\bm{\mathcal F};\underline{\bm\varphi})}{\partial \underline{\bm\lambda}_i}\r\vert_{\underline{\bm\varphi}={{\bm\varphi}}}
\r)\,
\r]
=\frac 1{\sigma_i^4}\bm\Phi_i,\nn\\
\bm{\mathcal H}_{i}(\bm{\mathcal X}|\bm{\mathcal F};\bm\varphi)&=\lim_{T\to\infty} \E\l[
\frac 1T \l.\frac{\partial^2 \ell_{\text{\tiny E}}(\bm{\mathcal X}|\bm{\mathcal F};\underline{\bm\varphi})}{\partial \underline{\bm\lambda}_i^\prime\partial \underline{\bm\lambda}_i}\r\vert_{\underline{\bm\varphi}={{\bm\varphi}}}
\r] =-\frac 1{\sigma_i^2}\mbf I_r,\nn
\end{align}
where we imposed orthonormality of the factors as required by Assumption \ref{ass:ident}(b). Clearly,
\beq\label{eq:QMLADVar2}
\bm\Phi_i=\l\{\bm{\mathcal H}_{i}(\bm{\mathcal X}|\bm{\mathcal F};\bm\varphi)\r\}^{-1}
\bm{\mathcal I}_i(\bm{\mathcal X}|\bm{\mathcal F};\bm\varphi)
\l\{\bm{\mathcal H}_{i}(\bm{\mathcal X}|\bm{\mathcal F};\bm\varphi)\r\}^{-1}.
\eeq
It follows that for the asymptotic covariances \eqref{eq:QMLADVar} and \eqref{eq:QMLADVar1} to be equivalent it must be that the Fisher information and Hessian matrices of the full log-likelihood $\ell_{\text{\tiny E}}(\bm{\mathcal X};\underline{\bm\varphi})$ in \eqref{eq:LL0} and of the conditional log-likelihood $\ell_{\text{\tiny E}}(\bm{\mathcal X}|\bm{\mathcal F};\underline{\bm\varphi})$ in \eqref{eq:LLXF00} are asymptotically equivalent when computed in the true values of the parameters $\bm\varphi$. This is proved in the following theorem.



%

\begin{theorem}\label{prop:hessiani}
Under Assumptions \ref{ass:common} through \ref{ass:sign}, as $n,T\to\infty$, for any given $i=1,\ldots,n$,
\begin{compactenum}[(a)]
\item 
\[
\frac 1{\sqrt T}\l\Vert \l.\frac{\partial \ell_{\text{\tiny \upshape E}}(\bm{\mathcal X};\underline{\bm\varphi})}{\partial \underline{\bm\lambda}_i^\prime}\r\vert_{\underline{\bm\varphi}={\bm\varphi}}
-\l.\frac{\partial \ell_{\text{\tiny \upshape E}}(\bm{\mathcal X}|\bm{\mathcal F};\underline{\bm\varphi})}{\partial \underline{\bm\lambda}_i^\prime}\r\vert_{\underline{\bm\varphi}={\bm\varphi}}
\r\Vert=O_{\mathrm P}\l(\max\l(\frac 1{\sqrt n},\frac {\sqrt T}{n}\r)\r);
\]
\item 
\[
\frac 1T\l\Vert \l.\frac{\partial^2 \ell_{\text{\tiny \upshape E}}(\bm{\mathcal X};\underline{\bm\varphi})}{\partial \underline{\bm\lambda}_i^\prime\partial \underline{\bm\lambda}_i}\r\vert_{\underline{\bm\varphi}={\bm\varphi}}
-\l.\frac{\partial^2 \ell_{\text{\tiny \upshape E}}(\bm{\mathcal X}|\bm{\mathcal F};\underline{\bm\varphi})}{\partial \underline{\bm\lambda}_i^\prime\partial \underline{\bm\lambda}_i}\r\vert_{\underline{\bm\varphi}={\bm\varphi}}
\r\Vert=O_{\mathrm P}\l(\max\l(\frac 1n,\frac 1{\sqrt {nT}}\r)\r);
\]
\item 
\[
\frac 1{\sqrt T}\l\Vert \l.\frac{\partial \ell_{\text{\tiny \upshape E}}(\bm{\mathcal X};\underline{\bm\varphi})}{\partial \underline{\bm\lambda}_i^\prime}\r\vert_{\underline{\bm\varphi}=\wh{\bm\varphi}^{\text{\tiny \upshape QML,E}}}
-\l.\frac{\partial \ell_{\text{\tiny \upshape  E}}(\bm{\mathcal X}|\bm{\mathcal F};\underline{\bm\varphi})}{\partial \underline{\bm\lambda}_i^\prime}\r\vert_{\underline{\bm\varphi}=\wh{\bm\varphi}^{\text{\tiny \upshape QML,E}}}
\r\Vert=O_{\mathrm P}\l(\max\l(\frac 1{\sqrt n},\frac {\sqrt T}{n}\r)\r);
\]
\item 
\[
\frac 1T\l\Vert \l.\frac{\partial^2 \ell_{\text{\tiny \upshape E}}(\bm{\mathcal X};\underline{\bm\varphi})}{\partial \underline{\bm\lambda}_i^\prime\partial \underline{\bm\lambda}_i}\r\vert_{\underline{\bm\varphi}=\wh{\bm\varphi}^{\text{\tiny \upshape QML,E}}}
-\l.\frac{\partial^2 \ell_{\text{\tiny \upshape E}}(\bm{\mathcal X}|\bm{\mathcal F};\underline{\bm\varphi})}{\partial \underline{\bm\lambda}_i^\prime\partial \underline{\bm\lambda}_i}\r\vert_{\underline{\bm\varphi}=\wh{\bm\varphi}^{\text{\tiny \upshape QML,E}}}
\r\Vert=O_{\mathrm P}\l(\max\l(\frac 1n,\frac 1{\sqrt {nT}}\r)\r).
\]

\end{compactenum}
\end{theorem}

From Theorem \ref{prop:hessiani} it is clear that if $n^{-1}\sqrt T\to 0$, as $n,T\to\infty$, then we can consistently estimate $\text{A$\!\mathbb V$ar}_0(\sqrt T \wh{\bm\lambda}_i^{\text{\tiny QML,E}})$ in \eqref{eq:QMLADVar} by means of any consistent estimator of $\text{A$\!\mathbb V$ar}_1(\sqrt T \wh{\bm\lambda}_i^{\text{\tiny QML,E}})=\bm\Phi_i$, as the classical HAC estimator:
\begin{align}
\wh{\bm\Phi}_i =\l( \frac 1T\sum_{t=1}^T \wh{\mbf F}_t\wh{\mbf F}_t^\prime \wh{\xi}_{it}^2  \r)&+\l\{\sum_{k=1}^{m_T} \l(1-\frac k{M_T+1}\r)\l[
\l( \frac 1T\sum_{t=k+1}^T \wh{\mbf F}_t\wh{\mbf F}_{t-k}^\prime \wh{\xi}_{it}\wh{\xi}_{it-k}  \r)\r.\r.\nn\\
&+\l.\l.\l( \frac 1T\sum_{t=k+1}^T \wh{\mbf F}_t\wh{\mbf F}_{t-k}^\prime \wh{\xi}_{it}\wh{\xi}_{it-k}  \r)^\prime\r]
\r\},\nn
\end{align}
where $\wh{\mbf F}_t=(\wh{\mbf M}^x)^{-1/2}\wh{\mbf V}^{x\prime} \mbf x_t$ and $\wh{\xi}_{it}=x_{it}-\wh{\bm\lambda}_i^\prime\wh{\mbf F}_t$ are the PC estimators of the factors and idiosyncratic components, respectively. Consistency of this estimator is proved by \citet[Theorem 6]{Bai03}. 

\begin{rem}
\upshape{
The sandwich form of the asymptotic covariances \eqref{eq:QMLADVar} and \eqref{eq:QMLADVar2} comes from the fact that in the log-likelihood we treat the idiosyncratic components as serially uncorrelated, while, in fact, they might be autocorrelated. Indeed, if we assumed $\E[\xi_{it}\xi_{is}]=0$ for all $s,t=1,\ldots, T$ with $s\ne t$ and all $i=1,\ldots, T$, then $\bm \Phi_i = \sigma_i^2\mbf I_r$ (see also Remark \ref{rem:baili12}). It follows that: $\bm{\mathcal I}_i(\bm{\mathcal X}|\bm{\mathcal F};\bm\varphi)={\sigma_i^{-2}}\mbf I_r=-\bm{\mathcal H}_{i}(\bm{\mathcal X}|\bm{\mathcal F};\bm\varphi)$.
And, by virtue of Theorem \ref{prop:hessiani}, $\bm{\mathcal I}_i(\bm{\mathcal X};\bm\varphi)$ must coincide asymptotically with $-\bm{\mathcal H}_{i}(\bm{\mathcal X};\bm\varphi)$, thus, from \eqref{eq:QMLADVar}, we have 
$\text{A$\!\mathbb V$ar}_0(\sqrt T \wh{\bm\lambda}_i^{\text{\tiny QML,E}})=\{\bm{\mathcal I}_i(\bm{\mathcal X};\bm\varphi)\}^{-1}$.
%
}
\end{rem}

\section{Estimation of the factors}\label{sec:fact}
Up to this point we did not say anything about estimating the factors. This is because, first, the factors are typically estimated once we have estimated the loadings, and, second, there is not a clear definition of what a QML estimator of the factors is. 

If we treat the factors as random variables, then they are not parameters and they do not have a QML estimator. We know that their optimal (in mean-squared sense) estimator is their conditional mean given $\bm{\mathcal X}$, which, under Gaussianity can be estimated as the linear projection. If, consistently with the discussion in Section \ref{sec:PCQML} about QML estimation, we mis-specify the second order structure of the factors, by replacing $\bm\Omega^F$ with $\mbf I_{rT}$, and of the idiosyncratic components, by replacing $\bm\Omega^\xi$ with $\bm\Sigma^\xi$,
then such linear projection is given by:
\beq\label{eq:FLP}
\mbf F_t^{\text{\tiny LP}}=\bm\Lambda^\prime \left(\bm\Lambda\bm\Lambda^\prime+\bm\Sigma^\xi\right)^{-1} \mbf x_t=
\left(\bm\Lambda^\prime(\bm\Sigma^\xi)^{-1}\bm\Lambda+\mbf I_r\right)^{-1}\bm\Lambda^\prime(\bm\Sigma^\xi)^{-1}\mbf x_t, \quad t=1,\ldots, T,
\eeq
where we used Woodbury formula. Alternatively, it is common to use the OLS or GLS estimators:
\beq\label{eq:FLS}
\mbf F_t^{\text{\tiny OLS}}=\left(\bm\Lambda^\prime\bm\Lambda\right)^{-1}\bm\Lambda^\prime\mbf x_t, \;\text{ or }\; 
\mbf F_t^{\text{\tiny GLS}}=\left(\bm\Lambda^\prime(\bm\Sigma^\xi)^{-1}\bm\Lambda\right)^{-1}\bm\Lambda^\prime(\bm\Sigma^\xi)^{-1}\mbf x_t, \quad t=1,\ldots, T.
\eeq
If instead we treat the factors as constant parameters then, as discussed in Remark \ref{rem:lobuttiamo}, we can see the GLS estimator as the QML estimator maximizing the joint log-likelihood \eqref{eq:ZLL} of factors and data.

The three estimators defined above are unfeasible unless we first compute estimates of the parameters. By virtue of our result in Theorem \ref{th:QML} there is asymptotically no difference if we use $\wh{\bm\Lambda}$ or $\wh{\bm\Lambda}^{\text{\tiny QML}}$, and an estimator of $\bm\Sigma^\xi$ can easily be computed either from the residuals of PC estimation or by QML as suggested by \citet{baili12}. Once we use these estimated parameters in \eqref{eq:FLP} and \eqref{eq:FLS}, we have the estimators $\wh {\mbf F}_t^{\text{\tiny LP}}$, $\wh{\mbf F}_t^{\text{\tiny OLS}}$, and $\wh{\mbf F}_t^{\text{\tiny GLS}}$. Notice that, as shown in Section \ref{sec:PC}, $\wh{\mbf F}_t^{\text{\tiny OLS}}$ is nothing else but the PC estimator of the factors. The GLS has also been studied by \citet{BT11} and \citet{choi12} when using different estimators of $\bm\Sigma^\xi$.

By construction, the OLS and GLS in \eqref{eq:FLS} are both less efficient than the linear projection in \eqref{eq:FLP}. Moreover, the GLS is always more efficient than the OLS if we could use an estimator of the full idiosyncratic covariance matrix, but, since this is in general unfeasible and we typically estiamate only its diagonal as in \eqref{eq:FLS}, then we can just conjecture that  the more sparse is the true covariance matrix the more likely the GLS is to be more efficient. 

Finally, these estimators are all $\min(\sqrt n, T)$-consistent. For the OLS we refer to \citet[Theorem 1]{Bai03}. For the GLS we refer to \citet[Theorem 2]{baili16}. Moreover, it is  straightforward to see that, by Lemmas \ref{lem:Gxi}(v) and \ref{lem:FTLN}(i), we have $\Vert \wh{\mbf F}_t^{\text{\tiny GLS}}-\wh{\mbf F}_t^{\text{\tiny LP}} \Vert=O_{\mathrm P}(n^{-1})$.

\begin{rem}\label{rem:filtro}
\upshape{
If we explicitly model the dynamics of $\mbf F_t$ then the expression of $\wh{\mbf F}_t^{\text{\tiny LP}}$ in \eqref{eq:FLP} is replaced by the Kalman smoother \citep{DGRfilter}. This, in fact, can be shown to be asymptotically equivalent, as $n\to\infty$, to the GLS estimator in \eqref{eq:FLS} \citep{baili16,PR22}. 
}
\end{rem}

\begin{rem}\upshape{
In the case of deterministic factors, we could also write the mis-specified log-likelihood \eqref{eq:LL00} of an exact factor model as:
\beq
\ell_{\text{\tiny E}}(\bm{\mathcal X};\underline{\bm\varphi})= -\frac 1 2\sum_{i=1}^n \log\det (\underline{\bm F}\,\underline{\bm \lambda}_i\underline{\bm \lambda}_i^\prime\underline{\bm F}^\prime+\underline{\sigma}_i^2\mbf I_T)-\frac 12\sum_{i=1}^n \bm x_i^\prime(\underline{\bm F}\,\underline{\bm \lambda}_i\underline{\bm \lambda}_i^\prime\underline{\bm F}^\prime+\underline{\sigma}_i^2\mbf I_T)^{-1}\bm x_i,\label{eq:LL00F}
\eeq
where $\bm x_i=(x_{i1}\cdots x_{iT})'$, thereby exchanging the role of $n$ and $T$. 
Then, we can conjecture that the QML estimator of the factors maximizing \eqref{eq:LL00F} will be asymptotically equivalent, this time as $T\to\infty$, to their PC estimator. 
This is approach is also considered in \citet{fortin2023latent}. However, since, as noted above, the PC estimator of the factors is asymptotically equivalent to the OLS, it is not the most efficient estimator because it does not account for possible cross-sectional heteroskedasticity of  idiosyncratic components. 
}\end{rem}

\section{Monte Carlo study}\label{sec:MC}

Throughout, we let $n\in\{20, 50, 100,200\}$, $T=100$, and $r=2$, and, for all $i=1,\ldots, n$ and $t=1,\ldots, T$, we simulate the data according to 
\begin{align}
x_{it}&={\bm\ell}_{i}^\prime {\bm f}_t + \phi_i \xi_{it}, \quad {\bm f}_t ={\bm A}  {\bm f}_{t-1} + \mbf u_t,\quad \xi_{it}=\alpha_i\xi_{it-1}+e_{it},\nn
\end{align}
where $\bm\ell_i$ and $\bm f_t$ are $r$-dimensional vectors. Specifically,
\begin{inparaenum}
	\item []  $\bm\ell_i$ has entries ${\ell}_{ij}{\sim}iid\,\mathcal{N}(1,1)$, $i=1,\ldots, N$, $j=1,\ldots, r$; 
	\item [] ${\bm A}=0.9 \check{\bm A} \Vert{\bm A}\Vert^{-1}$, where $\check{\bm A}$ is $r\times r$ with entries $\check{a}_{jj}{\sim} iid\,\mathcal U[0.5,0.8]$ for all $j$, and  $\check{ a}_{jk}{\sim} iid\, \mathcal U[0,0.3]$ if $j\ne k$;
	\item [] $u_{jt}{\sim}(0,1)$, $j=1,\ldots, r$, $t=1,\ldots, t$, with either a Gaussian or an Asymmetric Laplace distribution and with $\Cov(u_{it},u_{jt})=0$ if $i\ne j$,
	and $\Cov(u_{it},u_{jt-k})=0$ for all $i,j$ and if $k\ne 0$;
	\item [] $e_{it}{\sim}(0,\sigma_{ei}^2)$, with either a Gaussian or an Asymmetric Laplace distribution and with
	$\sigma_{ei}^2\sim \mathcal U[0.5, 1.5]$ for all $i$,
	$\Cov(e_{it},e_{jt})=\tau^{\vert i-j\vert }\mathbb I(\vert i-j\vert \le 10)$ with $\tau\in\{0,0.5\}$ if $i\ne j$,
	and $\Cov(e_{it},e_{jt-k})=0$ for all $i,j$ and if $k\ne 0$;
	\item  []  $\delta_i{\sim}iid\,\mathcal{U}(0,\delta)$, and $\delta\in\{0,0.5\}$;
	\item [] $\phi_i=\{{\theta_i (\sum_{t=1}^T \chi_{it}^2)/(\sum_{t=1}^T \xi_{it}^2)}\}^{1/2}$,
	and 
	$\theta_i{\sim}iid\,\mathcal{U}(0.25,0.5)$.
\end{inparaenum}
In the case of the Asymmetric Laplace distribution, all the innovations have location 0, asymmetry index $\kappa$, with $\kappa \sim \mathcal{U}(.9,1.1)$ and scale index $\lambda=\sqrt{{1+\kappa^4}/{\kappa^2}}$, so that the variance is 1. The parameters  $\tau$ and $\delta$ control the degrees of
cross-sectional and serial idiosyncratic correlation in the idiosyncratic components. The the noise-to-signal ratio for series $i$ is given by $\theta_i$. 

Finally, in order to satisfy Assumptions \ref{ass:ident}(a) and \ref{ass:ident}(b) we proceed as follows. Given the common components is generated as $\chi_{it}={\bm\ell}_{i}^\prime {\bm f}_t$, let $\bm\chi_t=(\chi_{1t}\cdots \chi_{nt})^\prime$ and compute $\check{\bm\Gamma}^\chi=T^{-1}\sum_{t=1}^T \bm\chi_t\bm\chi_t^\prime$. Collect its $r$ non-zero eigenvalues into the $r\times r$ diagonal matrix $\check{\mbf M}^{\chi}$ and the corresponding normalized eigenvectors as the columns of the $n\times r$ matrix $\check{\mbf V}^{\chi}$, with sign fixed such that it has non-negative entries in the first row. The loadings are then simulated as $\bm \Lambda=\check{\mbf V}^{\chi} (\check{\mbf M}^{\chi})^{1/2}$ and the  factors as $\mbf F_t = (\check{\mbf M}^{\chi})^{-1/2}\check{\mbf V}^{\chi\prime}\bm\chi_t$.

We simulate the model described above $B=500$ times, and at each replication we estimate the loadings via PC and QML, where the latter is defined as the maximizer of the log-likelihood \eqref{eq:LL00} and it is computed numerically in the way proposed by \citet{baili12,baili16}. In this way, we obtain two estimates of the loadings matrix $\wh{\bm\Lambda}^{(b)}$ and $\wh{\bm\Lambda}^{\text{\tiny QML,E}(b)}$, respectively, $b=1,\ldots, B$. We also compute the unfeasible OLS estimator ${\bm\Lambda}^{\text{\tiny OLS}(b)}$ by regressing $\mbf x_t$ onto the simulated factors.

 \begin{table}[t!]
 \caption{MSEs}\label{tab:MSEG}
 \centering
\vskip .1cm
\scriptsize
 \begin{tabular}{cccc | ccc | ccc}
 \hline
 \hline
 \multicolumn{10}{c}{\textsc {Gaussian innovations}}\\
 \hline
 &&&&&&& \\[-8pt]
 $n$ & $T$ &$\tau$ &$\delta$ &  $\text{MSE}_1^{\text{\tiny OLS}}$ & $\text{MSE}_1^{\text{\tiny PC}}$ & $\text{MSE}_1^{\text{\tiny QML}}$&  $\text{MSE}_2^{\text{\tiny OLS}}$ & $\text{MSE}_2^{\text{\tiny PC}}$ & $\text{MSE}_2^{\text{\tiny QML}}$\\
  &&&&&&& \\[-8pt]
 \hline
 &&&&&&& \\[-8pt]
 20 & 100 & 0 & 0 &	0.0103		&		0.0123		&		0.0116		&		0.0099		&		0.0153		&		0.0118		\\
&&&&\tiny(0.0032)	&	\tiny(0.0060)	&	\tiny(0.0053)	&	\tiny(0.0034)	&	\tiny(0.0069)	&	\tiny(0.0053)	\\
 50 & 100 & 0 & 0 &	0.0102		&		0.0109		&		0.0108		&		0.0102		&		0.0116		&		0.0107		\\
&&&&\tiny(0.0021)	&	\tiny(0.0023)	&	\tiny(0.0023)	&	\tiny(0.0022)	&	\tiny(0.0025)	&	\tiny(0.0023)	\\
 100 & 100 & 0 & 0 &	0.0100		&		0.0103		&		0.0103		&		0.0100		&		0.0105		&		0.0104		\\
&&&&\tiny(0.0015)	&	\tiny(0.0016)	&	\tiny(0.0015)	&	\tiny(0.0015)	&	\tiny(0.0016)	&	\tiny(0.0016)	\\
 200 & 100 & 0 & 0 &	0.0101		&		0.0102		&		0.0102		&		0.0101		&		0.0104		&		0.0103		\\
&&&&\tiny(0.0011)	&	\tiny(0.0011)	&	\tiny(0.0011)	&	\tiny(0.0011)	&	\tiny(0.0011)	&	\tiny(0.0011)	\\
 \hline
 &&&&&&& \\[-8pt]
  20 & 100 & 0.5 & 0.5 &	0.0180		&		0.0239		&		0.0230		&		0.0134		&		0.0270		&		0.0276		\\
&&&&\tiny(0.0081)	&	\tiny(0.0115)	&	\tiny(0.0110)	&	\tiny(0.0055)	&	\tiny(0.0188)	&	\tiny(0.0217)		\\
  50 & 100 & 0.5 & 0.5 &	0.0183		&		0.0201		&		0.0199		&		0.0135		&		0.0166		&		0.0160		\\
&&&&\tiny(0.0054)	&	\tiny(0.0060)	&	\tiny(0.0060)	&	\tiny(0.0037)	&	\tiny(0.0050)	&	\tiny(0.0047)		\\
  100 & 100 & 0.5 & 0.5 &	0.0182		&		0.0190		&		0.0189		&		0.0134		&		0.0146		&		0.0145		\\
&&&&\tiny(0.0038)	&	\tiny(0.0041)	&	\tiny(0.0041)	&	\tiny(0.0027)	&	\tiny(0.0031)	&	\tiny(0.0030)		\\
  200 & 100 & 0.5 & 0.5 &	0.0184		&		0.0187		&		0.0187		&		0.0135		&		0.0141		&		0.0141		\\
&&&&\tiny(0.0027)	&	\tiny(0.0028)	&	\tiny(0.0028)	&	\tiny(0.0021)	&	\tiny(0.0022)	&	\tiny(0.0022)		\\
\hline
\hline
 \multicolumn{10}{c}{\textsc {Asymmetric Laplace innovations}}\\
 \hline
 &&&&&&& \\[-8pt]
 $n$ & $T$ &$\tau$ &$\delta$ &  $\text{MSE}_1^{\text{\tiny OLS}}$ & $\text{MSE}_1^{\text{\tiny PC}}$ & $\text{MSE}_1^{\text{\tiny QML}}$&  $\text{MSE}_2^{\text{\tiny OLS}}$ & $\text{MSE}_2^{\text{\tiny PC}}$ & $\text{MSE}_2^{\text{\tiny QML}}$\\
  &&&&&&& \\[-8pt]
 \hline
 &&&&&&& \\[-8pt]
  20 & 100 & 0 & 0 & 	0.0401		&		0.0499		&		0.0440		&		0.0408		&		0.0654		&		0.0389		\\
&&&&\tiny(0.0136)	&	\tiny(0.0211)	&	\tiny(0.0180)	&	\tiny(0.0144)	&	\tiny(0.0254)	&	\tiny(0.0145)	\\
  50 & 100 & 0 & 0 &	0.0410		&		0.0442		&		0.0423		&		0.0407		&		0.0480		&		0.0390		\\
&&&&\tiny(0.0091)	&	\tiny(0.0102)	&	\tiny(0.0097)	&	\tiny(0.0090)	&	\tiny(0.0112)	&	\tiny(0.0087)	\\
  100 & 100 & 0 & 0 &	0.0412		&		0.0429		&		0.0418		&		0.0410		&		0.0439		&		0.0391		\\
&&&&\tiny(0.0068)	&	\tiny(0.0073)	&	\tiny(0.0071)	&	\tiny(0.0068)	&	\tiny(0.0077)	&	\tiny(0.0066)	\\
  200 & 100 & 0 & 0 &	0.0409		&		0.0420		&		0.0435		&		0.0411		&		0.0421		&		0.0421		\\
&&&&\tiny(0.0046)	&	\tiny(0.0050)	&	\tiny(0.0078)	&	\tiny(0.0051)	&	\tiny(0.0054)	&	\tiny(0.0075)	\\
\hline
&&&&&&& \\[-8pt]
  20 & 100 & 0.5 & 0.5 &	0.0755		&		0.1169		&		0.1016		&		0.0572		&		0.1257		&		0.1050		\\
&&&&\tiny(0.0353)	&	\tiny(0.1469)	&	\tiny(0.1405)	&	\tiny(0.0270)	&	\tiny(0.1577)	&	\tiny(0.1378)	\\
  50 & 100 & 0.5 & 0.5 &	0.0745		&		0.0864		&		0.0797		&		0.0549		&		0.0726		&		0.0587		\\
&&&&\tiny(0.0252)	&	\tiny(0.0292)	&	\tiny(0.0271)	&	\tiny(0.0179)	&	\tiny(0.0264)	&	\tiny(0.0191)	\\
  100 & 100 & 0.5 & 0.5 &	0.0742		&		0.0814		&		0.0765		&		0.0566		&		0.0640		&		0.0555		\\
&&&&\tiny(0.0183)	&	\tiny(0.0216)	&	\tiny(0.0194)	&	\tiny(0.0136)	&	\tiny(0.0162)	&	\tiny(0.0128)	\\
  200 & 100 & 0.5 & 0.5 &	0.0760		&		0.0800		&		0.0774		&		0.0548		&		0.0594		&		0.0542		\\
&&&&\tiny(0.0160)	&	\tiny(0.0178)	&	\tiny(0.0169)	&	\tiny(0.0106)	&	\tiny(0.0128)	&	\tiny(0.0108)	\\
\hline
\hline
 \end{tabular}
 \end{table}
 
  \begin{table}[t!]
 \caption{Comparison between PC and QML estimators.}\label{tab:MSEREL}
 \centering
\vskip .1cm
\scriptsize
 \begin{tabular}{cccc | cc | cc}
 \hline
 \hline
 \multicolumn{8}{c}{\textsc {Gaussian innovations}}\\
 \hline
&&&&&&& \\[-8pt]
 $n$ & $T$ &$\tau$ &$\delta$ & $D_1$ & $D_2$ &${\text{MSE}_1^{\text{\tiny REL}}}$ & ${\text{MSE}_2^{\text{\tiny REL}}}$ \\
  &&&&&&& \\[-8pt]
 \hline
 &&&&&&& \\[-8pt]
 20 & 100 & 0 & 0 &		5.06$\times 10^{-4}$		&		2.79$\times 10^{-3}$		&	1.06	&	1.30	\\
&&&&\tiny(1.72$\times 10^{-4}$)	&	\tiny(2.61$\times 10^{-3}$)	&		&		\\
 50 & 100 & 0 & 0 &	6.07$\times 10^{-5}$		&		4.50$\times 10^{-4}$		&	1.01	&	1.08	\\
&&&&\tiny(1.59$\times 10^{-5}$)	&	\tiny(2.01$\times 10^{-4}$)	&		&		\\
 100 & 100 & 0 & 0 &	1.23$\times 10^{-5}$		&		6.28$\times 10^{-5}$		&	1.00	&	1.01	\\
&&&&\tiny(3.54$\times 10^{-6}$)	&	\tiny(3.00$\times 10^{-5}$)	&		&		\\
 200 & 100 & 0 & 0 &	4.54$\times 10^{-6}$		&		2.05$\times 10^{-5}$		&	1.00	&	1.00	\\
&&&&\tiny(1.02$\times 10^{-6}$)	&	\tiny(6.96$\times 10^{-6}$)	&		&		\\
\hline
&&&&&&& \\[-8pt]
 20 & 100 & 0.5 & 0.5 &			5.72$\times 10^{-4}$		&		5.86$\times 10^{-3}$		&	1.04	&	0.98	\\
&&&&\tiny	(3.27$\times 10^{-4}$)	&	(1.15$\times 10^{-2}$)	&		&		\\
 50 & 100 & 0.5 & 0.5 &		7.55$\times 10^{-5}$		&		5.63$\times 10^{-4}$		&	1.01	&	1.04	\\
&&&&\tiny	(3.65$\times 10^{-5}$)	&	(4.88$\times 10^{-4}$)	&		&		\\
 100 & 100 & 0.5 & 0.5 &		1.77$\times 10^{-5}$		&		7.93$\times 10^{-5}$		&	1.00	&	1.01	\\
&&&&\tiny	(7.26$\times 10^{-6}$)	&	(5.36$\times 10^{-5}$)	&		&		\\
 200 & 100 & 0.5 & 0.5 &		6.06$\times 10^{-6}$		&		2.34$\times 10^{-5}$		&	1.00	&	1.00	\\
&&&&\tiny	(1.95$\times 10^{-6}$)	&	(1.02$\times 10^{-5}$)	&		&		\\
\hline
\hline
 \multicolumn{8}{c}{\textsc {Asymmetric Laplace innovations}}\\
 \hline
 &&&&&&& \\[-8pt]
 $n$ & $T$ &$\tau$ &$\delta$ & $D_1$ & $D_2$ &${\text{MSE}_1^{\text{\tiny REL}}}$ & ${\text{MSE}_2^{\text{\tiny REL}}}$ \\
  &&&&&&& \\[-8pt]
 \hline
 &&&&&&& \\[-8pt]
  20 & 100 & 0 & 0 &		4.35$\times 10^{-3}$		&		2.09$\times 10^{-2}$		&	1.13	&	1.68	\\
&&&&\tiny(1.67$\times 10^{-3}$)	&(1.68$\times 10^{-2}$)	&		&		\\
  50 & 100 & 0 & 0 &		1.13$\times 10^{-3}$		&		5.67$\times 10^{-3}$		&	1.05	&	1.23	\\
&&&&\tiny(3.33$\times 10^{-4}$)	&(3.19$\times 10^{-3}$)	&		&		\\
  100 & 100 & 0 & 0 &		6.23$\times 10^{-4}$		&		2.64$\times 10^{-3}$		&	1.03	&	1.12	\\
&&&&\tiny(1.48$\times 10^{-4}$)	&(1.28$\times 10^{-3}$)	&		&		\\
  200 & 100 & 0 & 0 &		2.20$\times 10^{-3}$		&		2.07$\times 10^{-3}$		&	0.97	&	1.00	\\
&&&&\tiny(3.81$\times 10^{-3}$)	&(2.41$\times 10^{-3}$)	&		&		\\
\hline
&&&&&&& \\[-8pt]
  20 & 100 & 0.5 & 0.5 &			9.93$\times 10^{-3}$		&		4.12$\times 10^{-2}$		&	1.15	&	1.20	\\
&&&&\tiny(6.77$\times 10^{-3}$)	&(5.68$\times 10^{-2}$)	&		&		\\
50 & 100 & 0.5 & 0.5 &		4.68$\times 10^{-3}$		&		9.16$\times 10^{-3}$		&	1.08	&	1.24	\\
&&&&\tiny(2.07$\times 10^{-3}$)	&(9.07$\times 10^{-3}$)	&		&		\\
100 & 100 & 0.5 & 0.5 &		3.13$\times 10^{-3}$		&		4.56$\times 10^{-3}$		&	1.06	&	1.15	\\
&&&&\tiny(1.05$\times 10^{-3}$)	&(2.92$\times 10^{-3}$)	&		&		\\
200 & 100 & 0.5 & 0.5 &		2.63$\times 10^{-3}$		&		3.41$\times 10^{-3}$		&	1.03	&	1.10	\\
&&&&\tiny(1.78$\times 10^{-3}$)	&(2.73$\times 10^{-3}$)	&		&		\\
\hline
\hline
 \end{tabular}
 \end{table}

 In Table \ref{tab:MSEG}  we report the Mean-Squared-Error (MSE) for each column, $j=1,\ldots, r$, of the considered estimators, averaged over the $B$ replications (with standard deviations in parenthesis):
\begin{align}
&\text{MSE}_j^{\text{\tiny OLS}} =\frac 1B\sum_{b=1}^B\l\{ \frac 1n \sum_{i=1}^n \l({\bm\Lambda}^{\text{\tiny OLS}(b)}_{ij}-\bm\Lambda_{ij} \r)^2\r\},\nn\\
& \text{MSE}_j^{\text{\tiny PC}} =\frac 1B\sum_{b=1}^B\l\{ \frac 1n \sum_{i=1}^n \l(\wh{\bm\Lambda}^{(b)}_{ij}-\bm\Lambda_{ij} \r)^2\r\},\quad \text{MSE}_j^{\text{\tiny QML}} =\frac 1B\sum_{b=1}^B\l\{ \frac 1n \sum_{i=1}^n \l(\wh{\bm\Lambda}^{\text{\tiny QML}(b)}_{ij}-\bm\Lambda_{ij} \r)^2\r\}.\nn
\end{align}
Results show that QML and PC estimator have similar MSEs and both improve as $n$ increases to the point that when $n=200$ their MSEs are comparable with the  one of the unfeasible OLS, i.e., the estimators behave as if the factors were observed. In most cases the QML estimator has a smaller MSE even when the true distribution is not Gaussian.

Finally, in Table \ref{tab:MSEREL} for each column of the loadings, $j=1,\ldots, r$, we report the distance between the QML and PC estimators measured as (with standard deviations in parenthesis):
\begin{align}
&\text{D}_j =\frac 1B\sum_{b=1}^B\l\{ \frac 1n \sum_{i=1}^n \l(\wh{\bm\Lambda}^{\text{\tiny QML}(b)}_{ij}-\wh{\bm\Lambda}_{ij}^{(b)} \r)^2\r\}.\nn
\end{align}
And we also report the relative MSE of the PC estimator with respect to the MSE of the QML estimator: $\text{MSE}_j^{\text{\tiny REL}}={\text{MSE}_j^{\text{\tiny PC}}}/{\text{MSE}_j^{\text{\tiny QML}}}$. Results clearly show that as $n$ grows the PC and QML estimators become almost indistinguishable.

\section{Concluding remarks}\label{sec:conc}

To compute in practice the QML estimator of the loadings there are at least two main issues. First, in finite samples the QML estimator of the loadings has no closed form and depends also on the estimator of the idiosyncratic covariance for which no closed form exists either. This is true even if we use the the log-likelihood \eqref{eq:LL00} of an exact factor model. Second, the convergence properties of the various available EM algorithms used to compute the QML estimator \citep{RT82,baili12,baili16,pz23} have never been fully investigated. On the one hand, it is easy to prove, that at each iteration of an EM algorithm the log-likelihood evaluated in the  parameters estimated at that iteration is larger than at the previous iteration \citep{wu83}, but, on the other hand, no formal proof exists of convergence of those algorithms to a global maximum of the likelihood, at least to the best of our knowledge. 

The results of this paper offer a possible solution by showing that, if we are just interested in the factor loadings and we do not need to estimate the idiosyncratic variances, then we can simply use the PC estimator of the loadings and of its asymptotic covariance matrix to approximate the corresponding QML estimator and  its asymptotic covariance matrix. Once this is done, the factors can be estimated via OLS as in PC analysis. 

As a consequence, we might think that there is no apparent advantage in directly computing the QML estimator of the loadings and of the idiosyncratic variances. 

Nevertheless, QML estimation has at least three advantages. First, it allows us to easily impose restrictions on the parameters of the model. Second, having also the QML estimator of the idiosyncratic variances allows us to compute estimators of the factors as the GLS, which are possibly more efficient. Third, QML estimation, as presented in this paper, is a first step towards estimating a model where we explicitly model the dynamics of the factors, something we cannot do with PC analysis. This last point, already briefly discussed  in Remarks \ref{rem:DGR} and \ref{rem:filtro}, is the subject of our ongoing research \citep{BLqml}.

\bibliographystyle{chicago} 
\bibliography{BL_Biblio}

\newcommand{\noop}[1]{}
\begin{thebibliography}{}

\bibitem[\protect\citeauthoryear{A\"it-Sahalia and Xiu}{A\"it-Sahalia and
  Xiu}{2017}]{ait2017using}
A\"it-Sahalia, Y. and D.~Xiu (2017).
\newblock Using principal component analysis to estimate a high dimensional
  factor model with high-frequency data.
\newblock {\em Journal of Econometrics\/}~{\em 201}, 384--399.

\bibitem[\protect\citeauthoryear{Amemiya, Fuller, and Pantula}{Amemiya
  et~al.}{1987}]{AFP87}
Amemiya, Y., W.~A. Fuller, and S.~G. Pantula (1987).
\newblock The asymptotic distributions of some estimators for a factor analysis
  model.
\newblock {\em Journal of Multivariate Analysis\/}~{\em 22}, 51--64.

\bibitem[\protect\citeauthoryear{Anderson and Amemiya}{Anderson and
  Amemiya}{1988}]{AA88}
Anderson, T.~W. and Y.~Amemiya (1988).
\newblock The asymptotic normal distribution of estimators in factor analysis
  under general conditions.
\newblock {\em The Annals of Statistics\/}~{\em 16}, 759--771.

\bibitem[\protect\citeauthoryear{Anderson and Rubin}{Anderson and
  Rubin}{1956}]{AR56}
Anderson, T.~W. and H.~Rubin (1956).
\newblock Statistical inference in factor analysis.
\newblock In {\em Proceedings of the third Berkeley symposium on mathematical
  statistics and probability}, Volume~5, pp.\  111--150.

\bibitem[\protect\citeauthoryear{Bai}{Bai}{2003}]{Bai03}
Bai, J. (2003).
\newblock Inferential theory for factor models of large dimensions.
\newblock {\em Econometrica\/}~{\em 71}, 135--171.

\bibitem[\protect\citeauthoryear{Bai and Li}{Bai and Li}{2012}]{baili12}
Bai, J. and K.~Li (2012).
\newblock Statistical analysis of factor models of high dimension.
\newblock {\em The Annals of Statistics\/}~{\em 40}, 436--465.

\bibitem[\protect\citeauthoryear{Bai and Li}{Bai and Li}{2016}]{baili16}
Bai, J. and K.~Li (2016).
\newblock Maximum likelihood estimation and inference for approximate factor
  models of high dimension.
\newblock {\em The Review of Economics and Statistics\/}~{\em 98}, 298--309.

\bibitem[\protect\citeauthoryear{Bai and Liao}{Bai and Liao}{2016}]{bailiao16}
Bai, J. and Y.~Liao (2016).
\newblock Efficient estimation of approximate factor models via penalized
  maximum likelihood.
\newblock {\em Journal of Econometrics\/}~{\em 191}, 1--18.

\bibitem[\protect\citeauthoryear{Bai and Ng}{Bai and Ng}{2002}]{baing02}
Bai, J. and S.~Ng (2002).
\newblock Determining the number of factors in approximate factor models.
\newblock {\em Econometrica\/}~{\em 70}, 191--221.

\bibitem[\protect\citeauthoryear{Bai and Ng}{Bai and Ng}{2013}]{baing13}
Bai, J. and S.~Ng (2013).
\newblock Principal components estimation and identification of static factors.
\newblock {\em Journal of Econometrics\/}~{\em 176}, 18--29.

\bibitem[\protect\citeauthoryear{Barigozzi}{Barigozzi}{2024}]{MBQML}
Barigozzi, M. (2024).
\newblock Quasi maximum likelihood estimation of high-dimensional factor
  models.
\newblock {\em Oxford Research Encyclopedia of Economics and Finance\/}.
\newblock forthcoming.

\bibitem[\protect\citeauthoryear{Barigozzi and Luciani}{Barigozzi and
  Luciani}{2024}]{BLqml}
Barigozzi, M. and M.~Luciani (2024).
\newblock Quasi maximum likelihood estimation and inference of large
  approximate dynamic factor models via the {EM} algorithm.
\newblock Technical Report arXiv:1910.03821.

\bibitem[\protect\citeauthoryear{Bernanke, Boivin, and Eliasz}{Bernanke
  et~al.}{2005}]{BBE05}
Bernanke, B.~S., J.~Boivin, and P.~S. Eliasz (2005).
\newblock Measuring the effects of monetary policy: A {F}actor-{A}ugmented
  {V}ector {A}utoregressive ({FAVAR}) approach.
\newblock {\em The Quarterly Journal of Economics\/}~{\em 120}, 387--422.

\bibitem[\protect\citeauthoryear{Breitung and Tenhofen}{Breitung and
  Tenhofen}{2011}]{BT11}
Breitung, J. and J.~Tenhofen (2011).
\newblock {GLS} estimation of dynamic factor models.
\newblock {\em Journal of the American Statistical Association\/}~{\em 106},
  1150--1166.

\bibitem[\protect\citeauthoryear{Chamberlain and Rothschild}{Chamberlain and
  Rothschild}{1983}]{chamberlainrothschild83}
Chamberlain, G. and M.~Rothschild (1983).
\newblock Arbitrage, factor structure, and mean-variance analysis on large
  asset markets.
\newblock {\em Econometrica\/}~{\em 51}, 1281--1304.

\bibitem[\protect\citeauthoryear{Choi}{Choi}{2012}]{choi12}
Choi, I. (2012).
\newblock Efficient estimation of factor models.
\newblock {\em Econometric Theory\/}~{\em 28}, 274--308.

\bibitem[\protect\citeauthoryear{Connor, Korajczyk, and Linton}{Connor
  et~al.}{2006}]{connor2006common}
Connor, G., R.~A. Korajczyk, and O.~Linton (2006).
\newblock The common and specific components of dynamic volatility.
\newblock {\em Journal of Econometrics\/}~{\em 132}, 231--255.

\bibitem[\protect\citeauthoryear{Coroneo, Giannone, and Modugno}{Coroneo
  et~al.}{2016}]{CGM16}
Coroneo, L., D.~Giannone, and M.~Modugno (2016).
\newblock Unspanned macroeconomic factors in the yield curve.
\newblock {\em Journal of Business and Economic Statistics\/}~{\em 34},
  472--485.

\bibitem[\protect\citeauthoryear{De~Mol, Giannone, and Reichlin}{De~Mol
  et~al.}{2008}]{de2008forecasting}
De~Mol, C., D.~Giannone, and L.~Reichlin (2008).
\newblock Forecasting using a large number of predictors: Is bayesian shrinkage
  a valid alternative to principal components?
\newblock {\em Journal of Econometrics\/}~{\em 146}, 318--328.

\bibitem[\protect\citeauthoryear{Delle~Chiaie, Ferrara, and
  Giannone}{Delle~Chiaie et~al.}{2021}]{DCGF2021}
Delle~Chiaie, S., L.~Ferrara, and D.~Giannone (2021).
\newblock Common factors of commodity prices.
\newblock {\em Journal of Applied Econometrics\/}.
\newblock forthcoming.

\bibitem[\protect\citeauthoryear{Doz, Giannone, and Reichlin}{Doz
  et~al.}{2011}]{DGRfilter}
Doz, C., D.~Giannone, and L.~Reichlin (2011).
\newblock A two-step estimator for large approximate dynamic factor models
  based on {K}alman filtering.
\newblock {\em Journal of Econometrics\/}~{\em 164}, 188--205.

\bibitem[\protect\citeauthoryear{Doz, Giannone, and Reichlin}{Doz
  et~al.}{2012}]{DGRqml}
Doz, C., D.~Giannone, and L.~Reichlin (2012).
\newblock A quasi maximum likelihood approach for large approximate dynamic
  factor models.
\newblock {\em The Review of Economics and Statistics\/}~{\em 94\/}(4),
  1014--1024.

\bibitem[\protect\citeauthoryear{Fan, Liao, and Mincheva}{Fan
  et~al.}{2013}]{FLM13}
Fan, J., Y.~Liao, and M.~Mincheva (2013).
\newblock Large covariance estimation by thresholding principal orthogonal
  complements.
\newblock {\em Journal of the Royal Statistical Society: Series B (Statistical
  Methodology)\/}~{\em 75}, 603--680.

\bibitem[\protect\citeauthoryear{Forni, Hallin, Lippi, and Reichlin}{Forni
  et~al.}{2000}]{FHLR00}
Forni, M., M.~Hallin, M.~Lippi, and L.~Reichlin (2000).
\newblock The {G}eneralized {D}ynamic {F}actor {M}odel: Identification and
  estimation.
\newblock {\em The Review of Economics and Statistics\/}~{\em 82}, 540--554.

\bibitem[\protect\citeauthoryear{Forni, Hallin, Lippi, and Reichlin}{Forni
  et~al.}{2005}]{FHLR05}
Forni, M., M.~Hallin, M.~Lippi, and L.~Reichlin (2005).
\newblock The {G}eneralized {D}ynamic {F}actor {M}odel: One sided estimation
  and forecasting.
\newblock {\em Journal of the American Statistical Association\/}~{\em 100},
  830--840.

\bibitem[\protect\citeauthoryear{Fortin, Gagliardini, and Scaillet}{Fortin
  et~al.}{2023}]{fortin2023latent}
Fortin, A.-P., P.~Gagliardini, and O.~Scaillet (2023).
\newblock Latent factor analysis in short panels.
\newblock Technical Report arXiv:2306.14004.

\bibitem[\protect\citeauthoryear{Giannone, Reichlin, and Small}{Giannone
  et~al.}{2008}]{Nowcasting}
Giannone, D., L.~Reichlin, and D.~Small (2008).
\newblock Nowcasting: The real-time informational content of macroeconomic
  data.
\newblock {\em Journal of Monetary Economics\/}~{\em 55}, 665--676.

\bibitem[\protect\citeauthoryear{Hannan}{Hannan}{1970}]{hannan}
Hannan, E.~J. (1970).
\newblock {\em Multiple time series}.
\newblock John Wiley \& Sons.

\bibitem[\protect\citeauthoryear{Jolliffe}{Jolliffe}{2002}]{jolliffe2002principal}
Jolliffe, I.~T. (2002).
\newblock {\em Principal component analysis}.
\newblock Springer.

\bibitem[\protect\citeauthoryear{Kim and Fan}{Kim and
  Fan}{2019}]{kim2019factor}
Kim, D. and J.~Fan (2019).
\newblock Factor {GARCH-I}t{\^o} models for high-frequency data with
  application to large volatility matrix prediction.
\newblock {\em Journal of Econometrics\/}~{\em 208}, 395--417.

\bibitem[\protect\citeauthoryear{Lam and Yao}{Lam and
  Yao}{2012}]{lam2012factor}
Lam, C. and Q.~Yao (2012).
\newblock Factor modeling for high-dimensional time series: inference for the
  number of factors.
\newblock {\em The Annals of Statistics\/}~{\em 40}, 694--726.

\bibitem[\protect\citeauthoryear{Lawley and Maxwell}{Lawley and
  Maxwell}{1971}]{lawleymaxwell71}
Lawley, D.~N. and A.~E. Maxwell (1971).
\newblock {\em Factor Analysis as a Statistical Method}.
\newblock Butterworths, London.

\bibitem[\protect\citeauthoryear{Mardia, Kent, and Bibby}{Mardia
  et~al.}{1979}]{mardia1979multivariate}
Mardia, K., J.~Kent, and J.~Bibby (1979).
\newblock {\em Multivariate analysis}.

\bibitem[\protect\citeauthoryear{Merikoski and Kumar}{Merikoski and
  Kumar}{2004}]{MK04}
Merikoski, J.~K. and R.~Kumar (2004).
\newblock Inequalities for spreads of matrix sums and products.
\newblock {\em Applied Mathematics E-Notes\/}~{\em 4}, 150--159.

\bibitem[\protect\citeauthoryear{Onatski}{Onatski}{2010}]{onatski10}
Onatski, A. (2010).
\newblock Determining the number of factors from empirical distribution of
  eigenvalues.
\newblock {\em The Review of Economics and Statistics\/}~{\em 92}, 1004--1016.

\bibitem[\protect\citeauthoryear{Rubin and Thayer}{Rubin and
  Thayer}{1982}]{RT82}
Rubin, D.~B. and D.~T. Thayer (1982).
\newblock {EM} algorithms for {ML} factor analysis.
\newblock {\em Psychometrika\/}~{\em 47}, 69--76.

\bibitem[\protect\citeauthoryear{Ruiz and Poncela}{Ruiz and
  Poncela}{2022}]{PR22}
Ruiz, E. and P.~Poncela (2022).
\newblock Factor extraction in {D}ynamic {F}actor {M}odels: Using {K}alman
  {F}ilter and {P}rincipal {C}omponents in practice.
\newblock {\em Foundations and Trends in Econometrics\/}.
\newblock forthcoming.

\bibitem[\protect\citeauthoryear{Stock and Watson}{Stock and
  Watson}{2002a}]{stockwatson02JASA}
Stock, J.~H. and M.~W. Watson (2002a).
\newblock Forecasting using principal components from a large number of
  predictors.
\newblock {\em Journal of the American Statistical Association\/}~{\em 97},
  1167--1179.

\bibitem[\protect\citeauthoryear{Stock and Watson}{Stock and
  Watson}{2002b}]{stockwatson02JBES}
Stock, J.~H. and M.~W. Watson (2002b).
\newblock Macroeconomic forecasting using diffusion indexes.
\newblock {\em Journal of Business and Economic Statistics\/}~{\em 20},
  147--162.

\bibitem[\protect\citeauthoryear{Tipping and Bishop}{Tipping and
  Bishop}{1999}]{tippingbishop99}
Tipping, M.~E. and C.~M. Bishop (1999).
\newblock Probabilistic principal component analysis.
\newblock {\em Journal of the Royal Statistical Society: Series B (Statistical
  Methodology)\/}~{\em 61}, 611--622.

\bibitem[\protect\citeauthoryear{Trapani}{Trapani}{2018}]{trapani2018randomized}
Trapani, L. (2018).
\newblock A randomized sequential procedure to determine the number of factors.
\newblock {\em Journal of the American Statistical Association\/}~{\em 113},
  1341--1349.

\bibitem[\protect\citeauthoryear{Wu}{Wu}{1983}]{wu83}
Wu, J. C.~F. (1983).
\newblock On the convergence properties of the {EM} algorithm.
\newblock {\em The Annals of Statistics\/}~{\em 11}, 95--103.

\bibitem[\protect\citeauthoryear{Yu, Wang, and Samworth}{Yu
  et~al.}{2015}]{yu15}
Yu, Y., T.~Wang, and R.~J. Samworth (2015).
\newblock A useful variant of the {D}avis-{K}ahan theorem for statisticians.
\newblock {\em Biometrika\/}~{\em 102}, 315--323.

\bibitem[\protect\citeauthoryear{Zadrozny}{Zadrozny}{2023}]{pz23}
Zadrozny, P.~A. (2023).
\newblock Gaussian maximum likelihood estimation of a static-form factor model
  using the {E}xpectation-{M}aximization algorithm.
\newblock mimeo, Bureau of Labor Statistics.

\end{thebibliography}

\begin{appendix}
\numberwithin{equation}{section}
\numberwithin{prop}{section}
\small
\section{Proof of main results}

\subsection{Proof of Theorem \ref{th:CLTL}}\label{app:A1}
From \eqref{eq:evecXX} and \eqref{eq:evecXX2}, and since by \eqref{eq:estL} we have $\wh{\bm\Lambda}^\prime\wh{\bm\Lambda}=\wh{\mbf M}^x$ and $\wh{\mbf V}^x=\wh{\bm\Lambda}(\wh{\mbf M}^x)^{-1/2}$ which is well defined because of Lemma \ref{lem:MO1}(iv), we get
\beq\label{eq:start}
\frac{\bm X^\prime\bm X}{nT}\wh{\bm\Lambda}=\wh{\bm\Lambda}\frac{\wh{\mbf M}^x}{n}.
\eeq
Then, substituting $\bm X^\prime\bm X=(\bm\Lambda\bm F^\prime+\bm\Xi^\prime)^\prime(\bm F\bm\Lambda^\prime+\bm\Xi)$ into \eqref{eq:start}
\begin{align}\label{eq:start2}
\frac{\bm\Lambda\bm F^\prime\bm F\bm\Lambda^\prime\wh{\bm\Lambda}}{nT}
+\frac{\bm\Lambda\bm F^\prime\bm \Xi\wh{\bm\Lambda}}{nT}
+\frac{\bm\Xi^\prime\bm F\bm\Lambda^\prime\wh{\bm\Lambda}}{nT}
+\frac{\bm\Xi^\prime\bm \Xi\wh{\bm\Lambda}}{nT}=\wh{\bm\Lambda}\frac{\wh{\mbf M}^x}{n}.
\end{align}
Define
\beq\label{eq:acca}
\wh{\mbf H}=
\l(\frac{\bm F^\prime\bm F}{T}\r)
\l(\frac{\bm\Lambda^\prime\wh{\bm\Lambda}}{n}\r)
\l(\frac{\wh{\mbf M}^x}{n}\r)^{-1}=
\l(\frac{\bm\Lambda^\prime\wh{\bm\Lambda}}{n}\r)
\l(\frac{\wh{\mbf M}^x}{n}\r)^{-1},
\eeq
by Assumption \ref{ass:ident}(b). Notice that, as $n,T\to\infty$,  $\wh{\mbf H}$ is well defined because $\l(\frac{\wh{\mbf M}^x}{n}\r)^{-1}$ is well defined because of Lemma \ref{lem:MO1}(iv). From \eqref{eq:start2} and \eqref{eq:acca}
\begin{align}
\wh{\bm\Lambda}&-\bm\Lambda\wh{\mbf H}= \l(\frac{\bm\Lambda\bm F^\prime\bm \Xi\wh{\bm\Lambda}}{nT}
+\frac{\bm\Xi^\prime\bm F\bm\Lambda^\prime\wh{\bm\Lambda}}{nT}
+\frac{\bm\Xi^\prime\bm \Xi\wh{\bm\Lambda}}{nT}\r)\l(\frac{\wh{\mbf M}^x}{n}\r)^{-1}\nn\\
=&\,
\l(\frac{\bm\Lambda\bm F^\prime\bm \Xi{\bm\Lambda}}{nT}
+\frac{\bm\Xi^\prime\bm F\bm\Lambda^\prime{\bm\Lambda}}{nT}
+\frac{\bm\Xi^\prime\bm \Xi{\bm\Lambda}}{nT}\r)\bm{\mathcal H} \l(\frac{\wh{\mbf M}^x}{n}\r)^{-1}\!\!\!\!\!+
\l(\frac{\bm\Lambda\bm F^\prime\bm \Xi}{nT}
+\frac{\bm\Xi^\prime\bm F\bm\Lambda^\prime}{nT}
+\frac{\bm\Xi^\prime\bm \Xi}{nT}\r)(\wh{\bm\Lambda}-\bm\Lambda\bm{\mathcal H}) \l(\frac{\wh{\mbf M}^x}{n}\r)^{-1}.\label{eq:start4}
\end{align}
Taking the $i$th row of  \eqref{eq:start4}
\begin{align}
\wh{\bm\lambda}_i^\prime-{\bm\lambda}_i^\prime\wh{\mbf H}=&\,
\l(
\underbrace{\frac 1{nT}{\bm\lambda}_i^\prime\sum_{t=1}^T\sum_{j=1}^n\mbf F_t\xi_{jt}{\bm\lambda}_j^\prime}_{\text{(1.a)}}
+\underbrace{\frac 1{nT} \sum_{t=1}^T \xi_{it}\mbf F_t^\prime\sum_{j=1}^n\bm\lambda_j\bm\lambda_j^\prime}_{\text{(1.b)}}
+\underbrace{\frac 1{nT} \sum_{t=1}^T\sum_{j=1}^n\xi_{it}\xi_{jt} \bm\lambda_j^\prime}_{\text{(1.c)}}
\r)\bm{\mathcal H} \l(\frac{\wh{\mbf M}^x}{n}\r)^{-1}\nn\\
&+
\l(
\underbrace{\frac 1{nT}{\bm\lambda}_i^\prime\sum_{t=1}^T\sum_{j=1}^n\mbf F_t\xi_{jt}(\wh{\bm\lambda}_j^\prime-\bm\lambda_j^\prime\bm{\mathcal H})}_{\text{(1.d)}}
+\underbrace{\frac 1{nT} \sum_{t=1}^T \xi_{it}\mbf F_t^\prime\sum_{j=1}^n\bm\lambda_j(\wh{\bm\lambda}_j^\prime-\bm\lambda_j^\prime\bm{\mathcal H})}_{\text{(1.e)}}\r.\nn\\
&\;\;\l.+\underbrace{\frac 1{nT} \sum_{t=1}^T\sum_{j=1}^n\xi_{it}\xi_{jt} (\wh{\bm\lambda}_j^\prime-\bm\lambda_j^\prime\bm{\mathcal H})}_{\text{(1.f)}}\r) \l(\frac{\wh{\mbf M}^x}{n}\r)^{-1}.\label{eq:sviluppoLambda}
\end{align}
From Proposition \ref{prop:load} we see that, under Assumptions \ref{ass:common} through \ref{ass:eval}, the terms $\text{\upshape (1.a)}$, $\text{\upshape (1.c)}$, $\text{\upshape (1.d)}$, $\text{\upshape (1.e)}$, and $\text{\upshape (1.f)}$ are all $o_{\mathrm P}\l(\frac 1{\sqrt T}\r)$. In particular,  from \eqref{eq:sviluppoLambda}, for any $i=1,\ldots, n$, we get
\begin{align}
\wh{\bm\lambda}_i-\wh{\mbf H}^\prime{\bm\lambda}_i = \l(\frac{\wh{\mbf M}^x}{n}\r)^{-1} \bm{\mathcal H}^\prime\l(\frac{\bm\Lambda^\prime\bm\Lambda}{n}\r) \l(\frac 1T\sum_{t=1}^T \mbf F_t\xi_{it}\r)
+ O_{\mathrm {P}}\l(\max\l(\frac 1 n, \frac 1{\sqrt{nT}}\r)\r)\label{eq:finaleL}.
\end{align}
Consistency follows immediately since the first term in \eqref{eq:finaleL} is $O_{\mathrm P}\l(\frac 1{\sqrt T}\r)$ because of Assumption \ref{ass:CLT} (see also \eqref{eq:2a3bis}  in the proof of Theorem \ref{prop:hessiani}) and since  $\l\Vert \l(\frac{\wh{\mbf M}^x}{n}\r)^{-1} \bm{\mathcal H}^\prime \r\Vert = O_{\mathrm {P}}(1)$ because of Lemma \ref{lem:MO1}(iv) and \ref{lem:HO1bis}.

Then, from \eqref{eq:finaleL}, 
by Proposition \ref{prop:KKK}(b):
\begin{align}\label{eq:CLTL}
\sqrt T(\wh{\bm\lambda}_i-\wh{\mbf H}^\prime{\bm\lambda}_i) &= \l(\frac{\wh{\mbf M}^x}{n}\r)^{-1} \bm{\mathcal H}^\prime\l(\frac{\bm\Lambda^\prime\bm\Lambda}{n}\r) \l(\frac 1{\sqrt T}\sum_{t=1}^T \mbf F_t\xi_{it}\r)
+ o_{\mathrm {P}}(1)\nn\\
&=  \l(\frac{\wh{\mbf M}^x}{n}\r)^{-1}\l(\frac{\wh{\bm\Lambda}^\prime\bm\Lambda}{n}\r) \l(\frac 1{\sqrt T}\sum_{t=1}^T \mbf F_t\xi_{it}\r)
+ o_{\mathrm {P}}(1).
\end{align}
Now, from Proposition \ref{prop:H} we have that, as $n,T\to\infty$, 
$\min(\sqrt{n},\sqrt T)\Vert \wh{\mbf H}-\bm J\Vert = o_{\mathrm P}(1)$, and, by using the definition of $\wh{\mbf H}$ in \eqref{eq:acca}, it follows that \eqref{eq:CLTL} is equivalent to (note that $\Vert\bm\lambda_i\Vert=O(1)$ by Assumption \ref{ass:common}(a))
\begin{align}\label{eq:CLTLalt}
\sqrt T(\wh{\bm\lambda}_i-\bm J\bm\lambda_i)&=\sqrt T(\wh{\bm\lambda}_i-\wh{\mbf H}^\prime\bm\lambda_i)+\sqrt T(\wh{\mbf H}^\prime-\bm J)\bm\lambda_i+o_{\mathrm P}(1)\nn\\
 &=  \l(\frac{\wh{\mbf M}^x}{n}\r)^{-1} \l(\frac{\wh{\bm\Lambda}^\prime\bm\Lambda}{n}\r) \l(\frac 1{\sqrt T}\sum_{t=1}^T \mbf F_t\xi_{it}\r)+ o_{\mathrm {P}}(1) \nn\\
&=\wh{\mbf H}^\prime  \l(\frac 1{\sqrt T}\sum_{t=1}^T \mbf F_t\xi_{it}\r)+o_{\mathrm P}(1)\nn\\
&= \bm J\l(\frac 1{\sqrt T}\sum_{t=1}^T \mbf F_t\xi_{it}\r)+o_{\mathrm P}(1)\nn\\
&\to_d\mathcal N\l(\mbf 0_r, \bm\Phi_i\r),
\end{align}
where we used Slutsky's theorem and Assumption \ref{ass:CLT}. Notice that, $\bm J$ plays no role in the covariance since it is diagonal and $\bm J^2=\mbf I_r$. Because of Assumption \ref{ass:sign} the sign indeterminacy on the left-hand-side can be easily fixed so that $\bm J=\mbf I_r$. By substituting $\mbf I_r$ in place of $\wh{\mbf H}$ in \eqref{eq:finaleL} and using Assumption \ref{ass:CLT} it follows also that:
\[
\l\Vert \wh{\bm\lambda}_i-{\bm\lambda}_i  \r\Vert = O_{\mathrm {P}}\l(\max\l(\frac 1 n, \frac 1{\sqrt{T}}\r)\r).
\]
This completes the proof of part (a).\smallskip

For part (b), from \eqref{eq:start4}
\begin{align}
\l\Vert\frac{\wh{\bm\Lambda}-\bm\Lambda\wh{\mbf H}}{\sqrt n}\r\Vert\le&\,
\l(\l\Vert\frac{\bm F^\prime\bm \Xi}{\sqrt nT}\r\Vert\,\l\Vert\frac{\bm\Lambda}{\sqrt n}\r\Vert^2
+\l\Vert \frac{\bm\Xi^\prime\bm F}{\sqrt {n}T}\r\Vert\,\l\Vert \frac{\bm\Lambda^\prime{\bm\Lambda}}{n}\r\Vert
+\l\Vert\frac{\bm\Xi^\prime\bm \Xi\bm\Lambda}{n^{3/2}T}\r\Vert\r)\l\Vert\bm{\mathcal H} \r\Vert \,\l\Vert\l(\frac{\wh{\mbf M}^x}{n}\r)^{-1}\r\Vert\nn\\
&+
\l(\l\Vert\frac{\bm F^\prime\bm \Xi}{\sqrt nT}\r\Vert\, \l\Vert\frac{\bm\Lambda}{\sqrt n}\r\Vert
+\l\Vert\frac{\bm\Xi^\prime\bm F}{\sqrt nT}\r\Vert\, \l\Vert\frac{\bm\Lambda}{\sqrt n}\r\Vert
+\l\Vert\frac{\bm\Xi^\prime\bm \Xi}{nT}\r\Vert\r)\l\Vert\frac{\wh{\bm\Lambda}-\bm\Lambda\bm{\mathcal H}}{\sqrt n}\r\Vert \l\Vert\l(\frac{\wh{\mbf M}^x}{n}\r)^{-1}\r\Vert,\nn
\end{align}
and the proof of part (b) follows from 
Proposition \ref{prop:L} and Lemma
\ref{lem:FTLN}(i),
\ref{lem:LLN},
\ref{lem:aiuto}(i),
\ref{lem:aiuto}(ii),
\ref{lem:aiuto}(iii),
\ref{lem:MO1}(iv),
\ref{lem:HO1bis}. This proves part (b) and completes the proof. 
\hfill $\Box$

\subsection{Proof of Corollary \ref{cor:PCOLS}}\label{app:A2}
\noindent
From \eqref{eq:finaleL} and the second last line of \eqref{eq:CLTLalt} in the proof of Theorem \ref{th:CLTL} and by imposing the identification constraint of orthonormal factors in Assumption \ref{ass:ident}(b) and $\bm J=\mbf I_r$ by Assumption \ref{ass:sign}, we have
\begin{align}
\wh{\bm\lambda}_i-\bm\lambda_i &= \frac 1{ T}\sum_{t=1}^T \mbf F_t\xi_{it}+ O_{\mathrm P}\l(\max\l(\frac 1n,\frac 1{\sqrt {nT}}\r)\r),\label{eq:defOLS2}
\end{align}
where the rate of the last term comes from \eqref{eq:finaleL} in the proof of Theorem \ref{th:CLTL}(a). By definition of OLS, and again imposing Assumption \ref{ass:ident}(b),  we have: 
\begin{align}
{\bm\lambda}_i^{\text{\tiny \upshape OLS}}-\bm\lambda_i
= \frac 1{ T}\sum_{t=1}^T  \mbf F_t\xi_{it}.\label{eq:defOLS}
\end{align}
By comparing \eqref{eq:defOLS2} and \eqref{eq:defOLS} we complete the proof of part (a). Part (b) follows also directly from Theorem \ref{th:CLTL}(b). This completes the proof.
\hfill $\Box$

\subsection{Proof of Theorem \ref{th:QML}}

In principle, we could try to replicate the proofs by \citet{baili16} under our identifying constraints and using only our assumptions. However, there is a much simpler and intuitive way to proceed. 

Consider the log-likelihood \eqref{eq:LL0}. The parameters to be estimated are given by ${\bm\varphi}=(\mathrm{vec}({\bm\Lambda})^\prime, \mathrm{vech}({\bm\Gamma}^\xi)^\prime)^\prime$. Let also $\wh{\bm\varphi}^{\text{\tiny \upshape QML}}=(\mathrm{vec}(\wh{\bm\Lambda}^{\text{\tiny \upshape QML}})^\prime, \mathrm{vech}(\wh{\bm\Gamma}^{\xi,\text{\tiny \upshape QML}})^\prime)^\prime$ denote  the maximizer of \eqref{eq:LL0} and 
$\underline{\bm\varphi}=(\mathrm{vec}(\underline{\bm\Lambda})^\prime, \mathrm{vech}(\underline{\bm\Gamma}^\xi)^\prime)^\prime$ denote a generic value of the parameters.
Whenever we consider $\underline{\bm\varphi}$ it is intended that its elements satisfy Assumptions \ref{ass:common} through \ref{ass:sign}.

Then, the elements of $\wh{\bm\varphi}^{\text{\tiny \upshape QML}}$ are such that:
\beq\label{intuito}
\wh{\bm\Lambda}^{\text{\tiny \upshape QML}}\wh{\bm\Lambda}^{\text{\tiny \upshape QML}\prime}+\wh{\bm\Gamma}^{\xi,\text{\tiny \upshape QML}}=\wh{\bm\Gamma}^x,
\eeq
where $\wh{\bm\Gamma}^x=T^{-1} {\bm X^\prime\bm X}$. To see that \eqref{intuito} defines the global maximum of the log-likelihood we proceed in two steps.

First, notice that the first order conditions derived from the log-likelihood \eqref{eq:LL0} are satisfied when \eqref{intuito} holds:
\begin{align}
&\l.\frac{\partial \ell(\bm{\mathcal X};\underline{\bm\varphi})}{\partial\underline{\bm\Lambda}}\r\vert_{\underline{\bm\Lambda}=\wh{\bm\Lambda}^{\text{\tiny \upshape QML}}}\nn\\
&= T
\l(\wh{\bm\Lambda}^{\text{\tiny \upshape QML}}\wh{\bm\Lambda}^{\text{\tiny \upshape QML}\prime}+\wh{\bm\Gamma}^{\xi,\text{\tiny \upshape QML}}\r)^{-1}
\wh{\bm\Gamma}^x
\l(\wh{\bm\Lambda}^{\text{\tiny \upshape QML}}\wh{\bm\Lambda}^{\text{\tiny \upshape QML}\prime}+\wh{\bm\Gamma}^{\xi,\text{\tiny \upshape QML}}\r)^{-1}
\wh{\bm\Lambda}^{\text{\tiny \upshape QML}}
-T\l(\wh{\bm\Lambda}^{\text{\tiny \upshape QML}}\wh{\bm\Lambda}^{\text{\tiny \upshape QML}\prime}+\wh{\bm\Gamma}^{\xi,\text{\tiny \upshape QML}}\r)^{-1}
\wh{\bm\Lambda}^{\text{\tiny \upshape QML}}\nn\\
&= T
\l(\wh{\bm\Lambda}^{\text{\tiny \upshape QML}}\wh{\bm\Lambda}^{\text{\tiny \upshape QML}\prime}+\wh{\bm\Gamma}^{\xi,\text{\tiny \upshape QML}}\r)^{-1}
\wh{\bm\Lambda}^{\text{\tiny \upshape QML}}
-T\l(\wh{\bm\Lambda}^{\text{\tiny \upshape QML}}\wh{\bm\Lambda}^{\text{\tiny \upshape QML}\prime}+\wh{\bm\Gamma}^{\xi,\text{\tiny \upshape QML}}\r)^{-1}
\wh{\bm\Lambda}^{\text{\tiny \upshape QML}}=\mbf 0_{r\times n},\nn\\
&\l.\frac{\partial \ell(\bm{\mathcal X};\underline{\bm\varphi})}{\partial\underline{\bm\Gamma}^\xi}\r\vert_{\underline{\bm\Gamma}^\xi=\wh{\bm\Gamma}^{\xi,\text{\tiny \upshape QML}}}\nn\\
&=
\frac T2   \l(\wh{\bm\Lambda}^{\text{\tiny \upshape QML}}\wh{\bm\Lambda}^{\text{\tiny \upshape QML}\prime}+\wh{\bm\Gamma}^{\xi,\text{\tiny \upshape QML}}\r)^{-1} \wh{\bm\Gamma}^x  
 \l(\wh{\bm\Lambda}^{\text{\tiny \upshape QML}}\wh{\bm\Lambda}^{\text{\tiny \upshape QML}\prime}+\wh{\bm\Gamma}^{\xi,\text{\tiny \upshape QML}}\r)^{-1}
- \frac T2  \l(\wh{\bm\Lambda}^{\text{\tiny \upshape QML}}\wh{\bm\Lambda}^{\text{\tiny \upshape QML}\prime}+\wh{\bm\Gamma}^{\xi,\text{\tiny \upshape QML}}\r)^{-1}\nn\\
&=
\frac T2   
 \l(\wh{\bm\Lambda}^{\text{\tiny \upshape QML}}\wh{\bm\Lambda}^{\text{\tiny \upshape QML}\prime}+\wh{\bm\Gamma}^{\xi,\text{\tiny \upshape QML}}\r)^{-1}
- \frac T2  \l(\wh{\bm\Lambda}^{\text{\tiny \upshape QML}}\wh{\bm\Lambda}^{\text{\tiny \upshape QML}\prime}+\wh{\bm\Gamma}^{\xi,\text{\tiny \upshape QML}}\r)^{-1}=\mbf 0_{n\times n}.\nn
\end{align}
Notice also that the conditions given in \citet[Equations (2.7)-(2.8)]{baili12}, which are derived from the first order conditions above, are also satisfied. Namely, it holds that:
\begin{align}
&\wh{\bm\Lambda}^{\text{\tiny \upshape QML}\prime}\l(\wh{\bm\Lambda}^{\text{\tiny \upshape QML}}\wh{\bm\Lambda}^{\text{\tiny \upshape QML}\prime}+\wh{\bm\Gamma}^{\xi,\text{\tiny \upshape QML}}\r)^{-1}\l\{\wh{\bm\Gamma}^x-\wh{\bm\Lambda}^{\text{\tiny \upshape QML}}\wh{\bm\Lambda}^{\text{\tiny \upshape QML}\prime}-\wh{\bm\Gamma}^{\xi,\text{\tiny \upshape QML}}\r\}=\mbf 0_{r\times n},\nn\\
&\l(\wh{\bm\Lambda}^{\text{\tiny \upshape QML}}\wh{\bm\Lambda}^{\text{\tiny \upshape QML}\prime}+\wh{\bm\Gamma}^{\xi,\text{\tiny \upshape QML}}\r)^{-1}=
\l(\wh{\bm\Lambda}^{\text{\tiny \upshape QML}}\wh{\bm\Lambda}^{\text{\tiny \upshape QML}\prime}+\wh{\bm\Gamma}^{\xi,\text{\tiny \upshape QML}}\r)^{-1}
\wh{\bm \Gamma}^x
\l(\wh{\bm\Lambda}^{\text{\tiny \upshape QML}}\wh{\bm\Lambda}^{\text{\tiny \upshape QML}\prime}+\wh{\bm\Gamma}^{\xi,\text{\tiny \upshape QML}}\r)^{-1},\nn\\
&\wh{\bm\Lambda}^{\text{\tiny \upshape QML}\prime}
\l(\wh{\bm\Lambda}^{\text{\tiny \upshape QML}}\wh{\bm\Lambda}^{\text{\tiny \upshape QML}\prime}+\wh{\bm\Gamma}^{\xi,\text{\tiny \upshape QML}}\r)^{-1}
\wh{\bm\Lambda}^{\text{\tiny \upshape QML}}=
\wh{\bm\Lambda}^{\text{\tiny \upshape QML}\prime}
\l(\wh{\bm\Lambda}^{\text{\tiny \upshape QML}}\wh{\bm\Lambda}^{\text{\tiny \upshape QML}\prime}+\wh{\bm\Gamma}^{\xi,\text{\tiny \upshape QML}}\r)^{-1}
\wh{\bm \Gamma}^x
\l(\wh{\bm\Lambda}^{\text{\tiny \upshape QML}}\wh{\bm\Lambda}^{\text{\tiny \upshape QML}\prime}+\wh{\bm\Gamma}^{\xi,\text{\tiny \upshape QML}}\r)^{-1}
\wh{\bm\Lambda}^{\text{\tiny \upshape QML}}.\nn
\end{align}

Second, given  the log-likelihood \eqref{eq:LL0}, for any  $\underline{\bm\varphi}$, we have:
\begin{align}
\ell(\bm{\mathcal X};\wh{\bm\varphi}^{\text{\tiny \upshape QML}})&-\ell(\bm{\mathcal X};\underline{\bm\varphi})\nn\\
&=-\frac{T}2\log\frac{\det\l( \wh{\bm\Lambda}^{\text{\tiny \upshape QML}}\wh{\bm\Lambda}^{\text{\tiny \upshape QML}\prime}+\wh{\bm\Gamma}^{\xi,\text{\tiny \upshape QML}}\r)}
{\det\l(\underline{\bm\Lambda}\,\underline{\bm\Lambda}^\prime+\underline{\bm\Gamma}^\xi \r)}-\frac{nT}{2}
+\frac T2
\text{tr}\l\{
\wh{\bm\Gamma}^x
\l(\underline{\bm\Lambda}\,\underline{\bm\Lambda}^\prime+\underline{\bm\Gamma}^\xi \r)^{-1}
\r\}\label{eq:DL}\\
&=-\frac{T}2\log\frac{\det\l( \wh{\bm\Lambda}^{\text{\tiny \upshape QML}}\wh{\bm\Lambda}^{\text{\tiny \upshape QML}\prime}+\wh{\bm\Gamma}^{\xi,\text{\tiny \upshape QML}}\r)}
{\det\l(\underline{\bm\Lambda}\,\underline{\bm\Lambda}^\prime+\underline{\bm\Gamma}^\xi \r)}-\frac{nT}{2}
+\frac T2
\text{tr}\l\{
\l(
\wh{\bm\Lambda}^{\text{\tiny \upshape QML}}\wh{\bm\Lambda}^{\text{\tiny \upshape QML}\prime}+\wh{\bm\Gamma}^{\xi,\text{\tiny \upshape QML}}
\r)
\l(\underline{\bm\Lambda}\,\underline{\bm\Lambda}^\prime+\underline{\bm\Gamma}^\xi \r)^{-1}
\r\},\nn
\end{align}
because of \eqref{intuito}. Now, denote as $\zeta_j$, $j=1,\ldots, n$, the $n$ roots of
\[
\det\l\{\l(\wh{\bm\Lambda}^{\text{\tiny \upshape QML}}\wh{\bm\Lambda}^{\text{\tiny \upshape QML}\prime}+\wh{\bm\Gamma}^{\xi,\text{\tiny \upshape QML}}\r)-\zeta\l(
\underline{\bm\Lambda}\,\underline{\bm\Lambda}^\prime+\underline{\bm\Gamma}^\xi 
\r) \r\}=0,
\]
which are all real since $\l(\wh{\bm\Lambda}^{\text{\tiny \upshape QML}}\wh{\bm\Lambda}^{\text{\tiny \upshape QML}\prime}+\wh{\bm\Gamma}^{\xi,\text{\tiny \upshape QML}}\r)\l(
\underline{\bm\Lambda}\,\underline{\bm\Lambda}^\prime+\underline{\bm\Gamma}^\xi 
\r)^{-1}
$ is a symmetric matrix. Then, \eqref{eq:DL} reads:
\beq\label{eq:DL0}
\ell(\bm{\mathcal X};\wh{\bm\varphi}^{\text{\tiny \upshape QML}})-\ell(\bm{\mathcal X};\underline{\bm\varphi}) = \frac T2\sum_{j=1}^n \l\{-\log\zeta_j-1+\zeta_j\r\}\ge 0,
\eeq
since $x\le e^{x-1}$ and so $-\log x-1+x\ge 0$. Therefore, from \eqref{eq:DL0} we see that \eqref{intuito} defines indeed the global maximum of the log-likleihood \eqref{eq:LL0}.

Now, consider the Singular Value Decomposition (SVD) of the true loadings:
\beq\label{svdQML0}
\frac{{\bm\Lambda}}{\sqrt n}={\mbf V}{\mbf D}{\mbf U},
\eeq
where ${\mbf V}$ is $n\times r$ and such that ${\mbf V}^\prime{\mbf V}=\mbf I_r$ for all $n\in\mathbb N$, 
${\mbf D}$ is $r\times r$ diagonal with strictly positive entries, and ${\mbf U}$  is $r\times r$ such that
${\mbf U}{\mbf U}^\prime={\mbf U}^{\prime}{\mbf U}=\mbf I_r$. 

%

We know that under Assumption \ref{ass:ident} and \ref{ass:sign}, $\bm\Lambda$ is globally identified. Let us show that  ${\mbf V}$, ${\mbf D}$, and ${\mbf U}$ in \eqref{svdQML0} are also globally identified. First, notice that, given Assumption \ref{ass:ident}(a) which requires $n^{-1}{{\bm\Lambda}^\prime{\bm\Lambda}}$ to be diagonal, in order to estimate ${\bm\Lambda}$ we need to estimate $nr-{r(r-1)}/{2}$ parameters. Then, by looking at the 
right-hand-side of \eqref{svdQML0} we see that to estimate ${\mbf V}$ we need to estimate $nr-\frac{r(r+1)}{2}$ parameters and to estimate ${\mbf D}$ we need to estimate $r$ parameters, thus ${\mbf V}{\mbf D}$ depends on $nr-{r(r+1)}/{2}+r=nr-{r(r-1)}/{2}$ parameters as ${\bm\Lambda}$. However, in principle the 
right-hand-side of \eqref{svdQML0} depends also on ${\mbf U}$ which in turn requires estimating ${r(r+1)}/2$ parameters more. But if we impose Assumption \ref{ass:ident}(a) also to the right-hand-side of \eqref{svdQML0} we have that $n^{-1}{{\bm\Lambda}^\prime{\bm\Lambda}}={\mbf U}^\prime{\mbf D}^2{\mbf U}$ has to be diagonal, and since ${\mbf D}$ is diagonal, without loss of generality we can set ${\mbf U}=\mbf J$, a diagonal $r\times r$ matrix with entries $\pm 1$. 

Furthermore, from Proposition \ref{prop:K00}(a) we also see that we must have
$n^{-1}{{\bm\Lambda}^\prime{\bm\Lambda}}=n^{-1}{\mbf M^\chi}=\mbf D^2$. 
Hence,
\beq\label{eq:defDD}
\mbf D = \l(\frac{\mbf M^\chi}n\r)^{1/2},
\eeq
and by Lemma \ref{lem:Gxi}(iv) the entries of $\mbf D$, denoted as $d_j$, $j=1,\ldots,r$, are such that
\beq\label{eq:CDDC}
\sqrt{\underline C_j}\!\le \lim\inf_{n\to\infty} d_j \le\lim\sup_{n\to\infty} d_j \le\! \sqrt{\overline C_j},
\eeq
where $\underline C_j$ and $\overline C_j$ are finite positive reals. Last, from \eqref{svdQML0} and Proposition \ref{prop:K00}(b), we must have
$n^{-1/2}{{\bm\Lambda}}={\mbf V}{\mbf D}\mbf J = \mbf V^\chi \l(n^{-1}{\mbf M^\chi}\r)^{1/2}$, and by \eqref{eq:defDD} it follows that 
\beq\label{eq:defVV}
{\mbf V} = \mbf V^\chi \mbf J.
\eeq
This shows that under Assumption \ref{ass:ident}(a), the parameters in ${\mbf D}$ are globally identified while the parameters in $\mbf V$ are uniquely identified up to a right-multiplication by $\mbf J$, which can be pinned down by means of Assumption \ref{ass:sign}, thus achieving global identification of $\mbf V$ as well. Given this discussion, hereafter, we can directly set $\mbf U=\mbf I_r$.

It is clear that the problem of QML estimation of the loadings can be rewritten as a problem of QML estimation of their SVD in \eqref{svdQML0}, namely of $\mbf V$ and $\mbf D$. To this end, we introduce also the SVDs of the QML estimator of the loadings and of a generic value of the loadings:
\begin{align}\label{eq:varieSVD}
&\frac{\wh{\bm\Lambda}^{\text{\tiny \upshape QML}}}{\sqrt n}=\wh{\mbf V}^{\text{\tiny \upshape QML}}\wh{\mbf D}^{\text{\tiny \upshape QML}}\wh{\mbf U}^{\text{\tiny \upshape QML}},\qquad\frac{\underline{\bm\Lambda}}{\sqrt n}=\underline{\mbf V}\,\underline{\mbf D}\,\underline{\mbf U},
\end{align}
where 
$\wh{\mbf V}^{\text{\tiny \upshape QML}}$ and $\underline{\mbf V}$ have the same properties as ${\mbf V}$, 
$\wh{\mbf D}^{\text{\tiny \upshape QML}}$ and $\underline{\mbf D}$ have the same properties as ${\mbf D}$, 
and
$\wh{\mbf U}^{\text{\tiny \upshape QML}}$ and $\underline{\mbf U}$ have the same properties as ${\mbf U}$. 

Because we set $\mbf U=\mbf I_r$, it follows that we can set $\wh{\mbf U}^{\text{\tiny \upshape QML}} = \mbf I_r$, and we are left with the task of finding $\wh{\mbf V}^{\text{\tiny \upshape QML}}$ and $\wh{\mbf D}^{\text{\tiny \upshape QML}}$. From \eqref{intuito} and since we must have $\wh{\mbf V}^{\text{\tiny \upshape QML}\prime}\wh{\mbf V}^{\text{\tiny \upshape QML}}=\mbf I_r$, it follows that
\beq
\frac{\wh{\bm\Gamma}^x}{n}-\wh{\mbf V}^{\text{\tiny \upshape QML}}\l(\wh{\mbf D}^{\text{\tiny \upshape QML}}\r)^2 \wh{\mbf V}^{\text{\tiny \upshape QML}\prime}-\frac{\wh{\bm\Gamma}^{\xi,\text{\tiny \upshape QML}}}{n}=\mbf 0_{n\times n},\nn
\eeq
which is equivalent to
\beq\label{intuito00}
\frac{\wh{\mbf V}^{\text{\tiny \upshape QML}\prime}\wh{\bm\Gamma}^x\wh{\mbf V}^{\text{\tiny \upshape QML}}}{n}-\l(\wh{\mbf D}^{\text{\tiny \upshape QML}}\r)^2 -\frac{\wh{\mbf V}^{\text{\tiny \upshape QML}\prime}\wh{\bm\Gamma}^{\xi,\text{\tiny \upshape QML}}\wh{\mbf V}^{\text{\tiny \upshape QML}}}{n}=\mbf 0_{r\times r}.
\eeq
Now, let $\wh{\bm v}_j^{\text{\tiny \upshape QML}}$, $j=1,\ldots, r$, be the $j$th column of $\wh{\mbf V}^{\text{\tiny \upshape QML}}$ and let 
$\wh{d}_j^{\,\text{\tiny \upshape QML}}$, $j=1,\ldots, r$, be the $j$th diagonal entry of $\wh{\mbf D}^{\text{\tiny \upshape QML}}$. Then, from \eqref{intuito00}, for all $j=1,\ldots, r$, we have
\beq
\frac{\wh{\bm v}_j^{\text{\tiny \upshape QML}\prime}\wh{\bm\Gamma}^x\wh{\bm v}_j^{\text{\tiny \upshape QML}}}{n}-\l(\wh{ d}_j^{\,\text{\tiny \upshape QML}}\r)^2 -\frac{\wh{\bm v}_j^{\text{\tiny \upshape QML}\prime}\wh{\bm\Gamma}^{\xi,\text{\tiny \upshape QML}}\wh{\bm v}_j^{\text{\tiny \upshape QML}}}{n}=0.
\eeq
 Then, trivially, the QML estimators are such that:
\beq\label{eq:minim1}
\l(\wh{\bm v}_j^{\text{\tiny \upshape QML}}, \wh{d}_j^{\,\text{\tiny \upshape QML}},\wh{\bm\Gamma}^{\xi,\text{\tiny \upshape QML}}\r)=\arg\!\!\!\!\min_{\underline{\bm v}_j,\underline{d}_j,\underline{\bm\Gamma}^\xi} 
\l(\frac{\underline{\bm v}_j^{\prime}\wh{\bm\Gamma}^x\underline{\bm v}_j}{n}-\underline{ d}_j^2 -
\frac{\underline{\bm v}_j^{\prime}\underline{\bm\Gamma}^{\xi}\underline{\bm v}_j}n
\r)^2=\arg\!\!\!\!\min_{\underline{\bm v}_j,\underline{d}_j,\underline{\bm\Gamma}^\xi} 
\mathcal L_{0j}\l(\underline{\bm v}_j, \underline d_j,\underline{\bm\Gamma}^{\xi} \r), \;\text{say,}
\eeq
where $\underline{\bm v}_j$, $j=1,\ldots, r$, is the $j$th column of $\underline{\mbf V}$ and  
$\underline{d}_j$, $j=1,\ldots, r$, is the $j$th diagonal entry of $\underline{\mbf D}$.

Define also the estimators $\wt{\bm v}_j$ and $\wt{d}_j$ such that:
\beq\label{eq:minim2}
\l(\wt{\bm v}_j, \wt{d}_j\r)=\arg\!\min_{\underline{\bm v}_j,\underline{d}_j}
\l(\frac{\underline{\bm v}_j^{\prime}\wh{\bm\Gamma}^x\underline{\bm v}_j}{n}-\underline{ d}_j^2 
\r)^2=\arg\!\min_{\underline{\bm v}_j,\underline{d}_j}\mathcal L_{1j}\l(\underline{\bm v}_j, \underline d_j\r), \;\text{say.}
\eeq
Consistently with \eqref{svdQML0} and \eqref{eq:varieSVD}, these define an estimator of $\bm\Lambda$ by means of its SVD:
\beq
\frac{\wt{\bm\Lambda}}{\sqrt n} = \wt{\mbf V}\wt{\mbf D}\wt{\mbf U},
\eeq
where $\wt{\mbf V}$ has columns $\wt{\bm v}_j$, $j=1,\ldots, r$, and it has the same properties as ${\mbf V}$ (because $\underline{\mbf V}$ does), and
$\wt{\mbf D}$ has entries  $\wt{ d}_j$, $j=1,\ldots, r$,  and it has the same properties as ${\mbf D}$ (because $\underline{\mbf D}$ does). We set $\wt{\mbf U}=\mbf I_r$ consistently with the fact that $\mbf U=\mbf I_r$.

Then, it is easily seen that
\beq\label{eq:DIFFL0L1}
\mathcal L_{0j}\l(\underline{\bm v}_j, \underline d_j,\underline{\bm\Gamma}^{\xi} \r)= \mathcal L_{1j}\l(\underline{\bm v}_j, \underline d_j\r)+\l(\frac{\underline{\bm v}_j^{\prime}\underline{\bm\Gamma}^{\xi}\underline{\bm v}_j}n\r)^2-2\l(\frac{\underline{\bm v}_j^{\prime}\underline{\bm\Gamma}^{\xi}\underline{\bm v}_j}n\r)\sqrt{\mathcal L_{1j}\l(\underline{\bm v}_j, \underline d_j\r)}.
\eeq
Now, we know that for any generic value of the idiosyncratic covariance matrix satisfying Assumption \ref{ass:idio}(b), it holds that
\beq\label{eq:DIFFL0L2}
\frac{\underline{\bm v}_j^{\prime}\underline{\bm\Gamma}^{\xi}\underline{\bm v}_j}n\le \max_{\bm w\, :\, \bm w^\prime\bm w=1}\frac{\bm w'\underline{\bm\Gamma}^{\xi}\bm w}{n}=\frac{\underline {\mu}^\xi_1}{n}\le \frac{M_{2\xi}}n,
\eeq
since $\underline{\bm v}_j^{\prime}\underline{\bm v}_j=1$ and because of  Lemma \ref{lem:Gxi}(v) which implies also that $M_{2\xi}$ is independent of $n$. Moreover, for any generic $\underline{\bm \Lambda}$ satisfying Assumption \ref{ass:common}(a):
\beq\label{eq:DIFFL0L3}
\sqrt{\mathcal L_{1j}\l(\underline{\bm v}_j, \underline d_j\r)} = \l\vert \frac{\underline{\bm v}_j^{\prime}\wh{\bm\Gamma}^x\underline{\bm v}_j}{n}-\underline{ d}_j^2  \r\vert \le 
 \frac{\underline{\bm v}_j^{\prime}\wh{\bm\Gamma}^x\underline{\bm v}_j}{n}+ \underline{ d}_j^2   
\le \frac{\wh{\mu}_1^x}{n} + \underline{ d}_j^2
= O_{\mathrm P}(1),
\eeq
since $\underline d_j= O(1)$ because $\Vert n^{-1/2}{\underline{\bm \Lambda}}\Vert= O(1)$ by Lemma \ref{lem:FTLN}(i), and $\wh{\mu}_1^x = O_{\mathrm P}(n)$ because of Lemma \ref{lem:MO1}(iii). Therefore, from \eqref{eq:DIFFL0L1}, \eqref{eq:DIFFL0L2}, and \eqref{eq:DIFFL0L3}, we have
\beq\label{eq:DIFFL0L4}
\l\vert\mathcal L_{0j}\l(\underline{\bm v}_j, \underline d_j,\underline{\bm\Gamma}^{\xi} \r)-\mathcal L_{1j}\l(\underline{\bm v}_j, \underline d_j \r)\r\vert = O_{\mathrm P}\l(\frac 1n\r),
\eeq
which holds for any $\underline{\bm\Gamma}^{\xi}$ satisfying Assumption \ref{ass:idio}(b). By continuity of these loss functions, from \eqref{eq:DIFFL0L4} it follows that their minima satisfy:
\begin{align}\label{eq:contLoss}
\l\Vert \wh{\bm v}_j^{\text{\tiny \upshape QML}}-\wt{\bm v}_j  \r\Vert = O_{\mathrm P}\l(\frac 1n\r),\qquad \l\vert \wh{d}_j^{\,\text{\tiny \upshape QML}}-\wt{d}_j  \r\vert = O_{\mathrm P}\l(\frac 1n\r),
\end{align}
for all $j=1,\ldots, r$.

Let us now find $\wt{\bm v}_j $ and $\wt{d}_j $. From \eqref{eq:minim2}, it is clear that the solutions must be such that:
%
\beq\label{eq:vGvD2}
\frac{\wt{\bm v}_j^{\prime}\wh{\bm\Gamma}^x\wt{\bm v}_j}{n}=\wt{ d}_j^{\,2},
\eeq
which means that $\wt{ d}_j^{\,2}$ must be an eigenvalue of $n^{-1}{\wh{\bm\Gamma}^x}$ and $\wt{\bm v}_j$ is the corresponding normalized eigenvector.  Obviously, the solution in \eqref{eq:vGvD20} defines a global minimum of the loss $\mathcal L_1(\underline{\bm v}_j,\underline d_j)$ since  $\mathcal L_1(\wt{\bm v}_j,\wt d_j)=0$ while $\mathcal L_1(\underline{\bm v}_j,\underline d_j)> 0$ for any other value $\underline{\bm v}_j\ne \wt{\bm v}_j$ and $\underline d_j\ne \wt d_j$. 

Now, let us show that indeed it must be that $\wt{ d}_j^{\,2}=n^{-1}{\wh{\mu}_j^x}$, $j=1,\ldots, r$, i.e., they have to be the $r$ largest eigenvalues of $n^{-1}{\wh{\bm\Gamma}^x}$. First, by Lemma \ref{lem:covarianze}(i) and Weyl's inequality, for all $k=1,\ldots, r$, as $n,T\to\infty$,
\beq\label{eq:trapani}
\l\vert\frac{\wh{\mu}_k^x}n-\frac{{\mu}_k^x}n\r\vert\le \l\Vert\frac{\wh{\bm\Gamma}^x}n-\frac{{\bm\Gamma}^x}n\r\Vert= O_{\mathrm P}\l(\frac 1{\sqrt T}\r).
\eeq
Therefore, if for any given $j=1,\ldots,r $ we were to choose $\wt{ d}_j^{\,2}=n^{-1}{\wh{\mu}_k^x}$ for, say, $k=r+1$, then, from \eqref{eq:trapani}, we would have
\begin{align}
&\lim\inf_{n\to\infty} \wt{ d}_j^{\,2} =\lim\inf_{n\to\infty} \frac{\wh{\mu}_{r+1}^x}n =\lim\inf_{n\to\infty}  \frac{{\mu}_{r+1}^x}n+O_{\mathrm P}\l(\frac 1{\sqrt T}\r)
\ge \lim\inf_{n\to\infty}  \frac{{\mu}_{n}^{\xi}}n+O_{\mathrm P}\l(\frac 1{\sqrt T}\r)= O_{\mathrm P}\l(\frac 1{\sqrt T}\r),\nn\\
&\lim\sup_{n\to\infty} \wt{ d}_j^{\,2} =\lim\sup_{n\to\infty} \frac{\wh{\mu}_{r+1}^x}n =\lim\sup_{n\to\infty}  \frac{{\mu}_{r+1}^x}n+O_{\mathrm P}\l(\frac 1{\sqrt T}\r)
\le \lim_{n\to\infty}  \frac{M_\xi}n+O_{\mathrm P}\l(\frac 1{\sqrt T}\r)= O_{\mathrm P}\l(\frac 1{\sqrt T}\r),\nn
\end{align}
since we assumed $\bm\Gamma^{\xi}$ to be positive definite, so $\mu^\xi_n>0$ for all $n\in\mathbb N$, and by Lemma \ref{lem:Gxi}(vi). Therefore, as $n,T\to\infty$, this choice for $\wt{ d}_j$ cannot be a consistent estimator of $d_j$ (nor an approximation of $\wh{d}_j^{\,\text{\tiny QML}}$), since while $\wt{ d}_j\to0$, as $n,T\to\infty$, it must be that $d_j>0$ for all $n\in\mathbb N$, as required in \eqref{eq:CDDC}.

So from \eqref{eq:vGvD2} and the above reasoning it follows that, for all $j=1,\ldots, r$,
\beq\label{eq:vGvD20}
\wt{ d}_j^{\,2}= \frac{\wh{\mu}_j^x}{n},\qquad \wt{\bm v}_j = \wh{\bm v}_j^x.
\eeq
where $ \wh{\bm v}_j^x$ is the eigenvector of  $n^{-1}{\wh{\bm\Gamma}^x}$ corresponding to its $j$th largest eigenvalue. 

Then, from \eqref{eq:vGvD20}, first  we have
\begin{align}
&\l\Vert \wh{\mbf V}^{\text{\tiny \upshape QML}}-\wh{\mbf  V}^x  \r\Vert=\l\Vert \wh{\mbf V}^{\text{\tiny \upshape QML}}-\wt{\mbf  V}  \r\Vert = O_{\mathrm P}\l(\frac 1n\r),\label{eq:contLos2}\\
&\l\Vert \wh{\mbf D}^{\,\text{\tiny \upshape QML}}-\l(\frac{\wh{\mbf M}^x}{n}\r)^{1/2}  \r\Vert=
\l\Vert \wh{\mbf D}^{\,\text{\tiny \upshape QML}}-\wt{\mbf D}  \r\Vert = O_{\mathrm P}\l(\frac 1n\r),\label{eq:contLos3}
\end{align}
because of   \eqref{eq:contLoss}, and, second, we have
\beq\label{eq:thisisPC}
\frac{\wt{\bm\Lambda}}{\sqrt n} =  \wh{\mbf V}^x\l(\frac{\wh{\mbf M}^x}{n}\r)^{1/2}=\frac{\wh{\bm\Lambda}}{\sqrt n} ,
\eeq
which shows that the estimator $\wt{\bm\Lambda}$  
 is the PC estimator  defined in \eqref{eq:estL}. Therefore, from \eqref{eq:contLos2}, \eqref{eq:contLos3}, and \eqref{eq:thisisPC}, and by using the SVD of the QML estimator in \eqref{eq:varieSVD}, it follows that:
\begin{align}
\l\Vert
\frac{\wh{\bm\Lambda}^{\text{\tiny QML}}}{\sqrt n}-\frac{\wh{\bm\Lambda}}{\sqrt n}
\r\Vert =&\, 
\l\Vert 
\wh{\mbf V}^{\text{\tiny \upshape QML}}\wh{\mbf D}^{\text{\tiny \upshape QML}}
-  \wh{\mbf V}^x\l(\frac{\wh{\mbf M}^x}{n}\r)^{1/2}
\r\Vert\nn\\
\le&\, \l\Vert
\wh{\mbf V}^{\text{\tiny \upshape QML}}-  \wh{\mbf V}^x
\r\Vert\,
\l\Vert
\l(\frac{\wh{\mbf M}^x}{n}\r)^{1/2}
\r\Vert
+
\l\Vert
\wh{\mbf D}^{\text{\tiny \upshape QML}}-\l(\frac{\wh{\mbf M}^x}{n}\r)^{1/2}
\r\Vert\, 
 \l\Vert
 \wh{\mbf V}^x
 \r\Vert\nn\\
 &+\l\Vert
\wh{\mbf V}^{\text{\tiny \upshape QML}}-  \wh{\mbf V}^x
\r\Vert\,
\l\Vert
\wh{\mbf D}^{\text{\tiny \upshape QML}}-\l(\frac{\wh{\mbf M}^x}{n}\r)^{1/2}
\r\Vert = O_{\mathrm P}\l(\frac 1n\r),\nn
\end{align}
because $ \Vert\wh{\mbf V}^x \Vert=1$ (it is a matrix of normalized eigenvectors) and $\Vert (n^{-1}{\wh{\mbf M}^x})^{1/2}\Vert=O_{\mathrm P}(1)$ by Lemma \ref{lem:MO1}(iii). This proves part (a).

Finally, consider the $r$ dimensional $i$th rows, $i=1,\ldots, n$, of ${\mbf V}^\chi$, $\wh{\mbf V}^x$, $\underline {\mbf V}$, and $\wh{\mbf V}^{\text{\tiny \upshape QML}}$, denoted as ${\mbf v}_i^{\chi\prime}$, $\wh{\mbf v}_i^{x\prime}$, $\underline{\mbf v}_i^\prime$ , and $\wh{\mbf v}_i^{\text{\tiny \upshape QML}\prime}$, respectively. From Lemma \ref{lem:covarianzerighe}(ii) and \ref{lem:covarianzerighe}(iii), we know that $\sqrt n\Vert {\mbf v}_i^{\chi\prime}\Vert = O(1)$ and $\sqrt n\Vert\wh{\mbf v}_i^{x\prime}\Vert=O_{\mathrm P}(1)$, which means that we must have $\sqrt n\Vert \underline {\mbf v}_i^\prime\Vert= O(1)$ since any generic parameter we consider is assumed to satisfy the same assumptions as the corresponding true parameters (${\mbf v}_i^{\chi\prime}$ in this case). Therefore, since the search of the elements of  $\wh{\mbf V}^{\text{\tiny \upshape QML}}$ is made over all $\underline {\mbf V}$ satisfying the same assumptions as ${\mbf V}^\chi$, it follows that we must have also $\sqrt n\Vert \wh{\mbf v}_i^{\text{\tiny \upshape QML}\prime}\Vert=O_{\mathrm P}(1)$. 

To see that this is the case, notice also that because of Assumption  \ref{ass:ident}(a) the QML estimator must be such that $n^{-1}{\wh{\bm \Lambda}^{\text{\tiny \upshape QML}\prime}\wh{\bm \Lambda}^{\text{\tiny \upshape QML}}}$ is diagonal and positive definite. Thus, we must have $\Vert n^{-1/2}{\wh{\bm \Lambda}^{\text{\tiny \upshape QML}}}\Vert = O_{\mathrm P}(1)$ (see also Lemma \ref{lem:FTLN}(i)). From the SVD in \eqref{eq:varieSVD} when $\wh{\mbf U}^{\text{\tiny \upshape QML}}=\mbf I_r$ and since it must be that $\wh{\mbf V}^{\text{\tiny \upshape QML}\prime}\wh{\mbf V}^{\text{\tiny \upshape QML}}=\mbf I_r$ for all $n\in\mathbb N$, then it follows that $\Vert \wh{\mbf V}^{\text{\tiny \upshape QML}}\Vert =O_{\mathrm P}(1)$ and $\Vert \wh{\mbf D}^{\text{\tiny \upshape QML}}\Vert =O_{\mathrm P}(\sqrt n)$. Moreover, since by Assumption \ref{ass:common}(a)   the $i$th row of
$\wh{\bm \Lambda}^{\text{\tiny \upshape QML}}$ must be such that $\Vert\wh{\bm \lambda}_i^{\text{\tiny \upshape QML}}\Vert=O_{\mathrm P}(1)$ and $\wh{\bm \lambda}_i^{\text{\tiny \upshape QML}\prime}= \wh{\mbf v}_i^{\text{\tiny \upshape QML}\prime} \wh{\mbf D}^{\text{\tiny \upshape QML}}$, then it must be that $\sqrt n\Vert \wh{\mbf v}_i^{\text{\tiny \upshape QML}\prime}\Vert=O_{\mathrm P}(1)$.

From this reasoning and \eqref{eq:contLos2} it follows that:
\beq\label{eq:rigacontLoss}
\sqrt n \l\Vert \wh{\mbf v}_i^{\text{\tiny \upshape QML}}-\wh{\mbf  v}_i^x 
\r\Vert = O_{\mathrm P}\l(\frac 1n\r).
\eeq
Now, by taking the $i$th row of the PC estimator in \eqref{eq:thisisPC} and the SVD of the QML estimator in \eqref{eq:varieSVD}, for any $i=1,\ldots, n$, we get
\begin{align}
\l\Vert
\wh{\bm\lambda}_i^{\text{\tiny QML}\prime}-\wh{\bm\lambda}_i^{\prime}
\r\Vert =&\, 
\l\Vert
\wh{\mbf v}_i^{\text{\tiny \upshape QML}\prime} \sqrt n \wh{ \mbf D}^{\text{\tiny \upshape QML}}-
 \wh{\mbf  v}_i^{x\prime}(\wh{\mbf M}^x)^{1/2} 
\r\Vert\nn\\
\le&\,
\l\Vert
 \wh{\mbf v}_i^{\text{\tiny \upshape QML}} - \wh{\mbf  v}_i^x 
\r\Vert \, 
\l\Vert
\l({\wh{\mbf M}^x}\r)^{1/2}
\r\Vert + 
\l\Vert
 \sqrt n \wh{ \mbf D}^{\text{\tiny \upshape QML}}-\l({\wh{\mbf M}^x}\r)^{1/2}
\r\Vert\,
\l\Vert
\wh{\mbf  v}_i^x
\r\Vert\nn\\
&+\l\Vert
 \wh{\mbf v}_i^{\text{\tiny \upshape QML}} - \wh{\mbf  v}_i^x 
\r\Vert \,
\l\Vert
 \sqrt n \wh{ \mbf D}^{\text{\tiny \upshape QML}}-\l({\wh{\mbf M}^x}\r)^{1/2}
\r\Vert= O_{\mathrm P}\l(\frac 1n\r),\nn
\end{align}
because of \eqref{eq:contLos3}, \eqref{eq:rigacontLoss}, and since $\Vert
\wh{\mbf  v}_i^x
\Vert=O_{\mathrm P}(n^{-1/2})$ by Lemma \ref{lem:covarianzerighe}(ii) and \ref{lem:covarianzerighe}(iii), and $\Vert ({\wh{\mbf M}^x})^{1/2}\Vert=O_{\mathrm P}(\sqrt n)$ by Lemma \ref{lem:MO1}(iii). This proves part (b) and completes the proof. \hfill $\Box$\\

\subsection{Proof of Theorem \ref{cor:QMLcons}}\label{app:A3}
From Theorem \ref{th:QML}(b) and  Corollary \ref{cor:PCOLS} it follows that:
\begin{align}\label{eq:QMLPCOLS}
\Vert\wh{\bm\lambda}_i^{\text{\tiny \upshape QML}}-{\bm\lambda}_i\Vert&=\Vert\wh{\bm\lambda}_i-{\bm\lambda}_i\Vert+O_{\mathrm {P}}\l(\frac 1{n}\r)\nn\\
&=\Vert\wh{\bm\lambda}_i-{\bm\lambda}_i^{\text{\tiny OLS}}\Vert+\Vert{\bm\lambda}_i^{\text{\tiny OLS}}-{\bm\lambda}_i\Vert+O_{\mathrm {P}}\l(\frac 1{n}\r)\nn\\
&=\Vert{\bm\lambda}_i^{\text{\tiny OLS}}-{\bm\lambda}_i\Vert+O_{\mathrm {P}}\l(\max\l(\frac 1n,\frac 1{\sqrt {nT}}\r)\r)\nn\\
&=O_{\mathrm {P}}\l(\max\l(\frac 1n,\frac 1{\sqrt {nT}}, \frac 1{\sqrt {T}}\r)\r),
\end{align}
since the unfeasible OLS estimator is $\sqrt T$-consistent. Moreover, if $\sqrt T/n\to 0$ as $n,T\to\infty$, 
by imposing the identification constraint of orthonormal factors in Assumption \ref{ass:ident}(b),
we have
\begin{align}
\sqrt T(\wh{\bm\lambda}_i^{\text{\tiny \upshape QML}}-\bm\lambda_i)&= 
\sqrt T(\wh{\bm\lambda}_i^{\text{\tiny \upshape OLS}}-\bm\lambda_i)+ o_{\mathrm P}(1)
=\frac 1{\sqrt T}\sum_{t=1}^T \mbf F_t\xi_{it}+ o_{\mathrm P}(1)\to_d\mathcal N\l(\mbf 0_r, \bm\Phi_i\r).\nn
\end{align}
by Slutsky's theorem and Assumption \ref{ass:CLT}. This completes the proof of part (a). Part (b) follows similarly from Theorem \ref{th:QML}(a). This completes the proof.
\hfill $\Box$

\subsection{Proof of Corollary \ref{cor:sett}}

The proof of part (a) follows by using the log-likelihood \eqref{eq:LL00} of an exact factor model in place of the log-likelihood \eqref{eq:LL0}, then, by 
noticing that $\bm{\Sigma}^{\xi}$ is positive definite by Assumption \ref{ass:idio}(a), and fnally by
replacing in the proof of Theorem \ref{th:QML}, $\wh{\bm\Gamma}^{\xi,\text{\tiny QML}}$ and $\underline{\bm\Gamma}^{\xi}$ with $\wh{\bm\Sigma}^{\xi,\text{\tiny QML}}$
and $\underline{\bm\Sigma}^{\xi}$ respectively. The proof of part (b) is the same but when using $\sigma^2\mbf I_n$, with $\sigma^2>0$, in place of $\bm\Sigma^\xi$. \hfill $\Box$

\subsection{Proof of Theorem \ref{prop:hessiani}}\label{app:A4}
First of all, denote the log-likelihoods for one observation:
\begin{align}
&\ell_{t}(\mbf x_t;\underline{\bm\varphi})=-\frac 12\log \det(\underline{\bm\Lambda}\,\underline{\bm\Lambda}^\prime+\underline{\bm\Sigma}^\xi)-\frac 12 \mbf x_t^\prime (\underline{\bm\Lambda}\,\underline{\bm\Lambda}^\prime+\underline{\bm\Sigma}^\xi)^{-1}\mbf x_t,\label{eq:LL0t}\\
&\ell_{t}(\mbf x_t|\mbf F_t;\underline{\bm\varphi})=-\frac 12\log \det(\underline{\bm\Sigma}^\xi)-\frac 12 (\mbf x_t-\underline{\bm\Lambda}\mbf F_t)^\prime (\underline{\bm\Sigma}^\xi)^{-1}(\mbf x_t-\underline{\bm\Lambda}\mbf F_t),\label{eq:LL0XFt}
\end{align}

Let us consider part (a). For any fixed value of the parameters, say $\wt{\bm\varphi}$, let
\begin{align}
\bm S(\bm{\mathcal X};\wt{\bm\varphi}) &= \sum_{t=1}^T\l.\frac{\partial \ell_t(\mbf x_t;\underline{\bm\varphi})}{\partial \underline{\bm\Lambda}^\prime}\r\vert_{\underline{\bm\varphi}={\wt{\bm\varphi}}} = \l(\ba{c}
\bm s_1^\prime(\bm{\mathcal X};\wt{\bm\varphi})\\
\vdots\\
\bm s_n^\prime(\bm{\mathcal X};\wt{\bm\varphi})\\
\ea
\r),\nn\\
\bm S(\bm{\mathcal X}|\bm{\mathcal F};\wt{\bm\varphi}) &= \sum_{t=1}^T\l.\frac{\partial \ell_t(\mbf x_t|\mbf F_t;\underline{\bm\varphi})}{\partial \underline{\bm\Lambda}^\prime}\r\vert_{\underline{\bm\varphi}={\wt{\bm\varphi}}} = \l(\ba{c}
\bm s_1^\prime(\bm{\mathcal X}|\bm{\mathcal F};\wt{\bm\varphi})\\
\vdots\\
\bm s_n^\prime(\bm{\mathcal X}|\bm{\mathcal F};\wt{\bm\varphi})\\
\ea
\r),\nn
\end{align}
which are $n\times r$ matrices of first derivatives, and where, for any given $i=1,\ldots, n$,
\beq\label{eq:defscore}
\bm s_i(\bm{\mathcal X};\wt{\bm\varphi})=\sum_{t=1}^T\l.\frac{\partial \ell_t(\mbf x_t;\underline{\bm\varphi})}{\partial \underline{\bm\lambda}_i^\prime}\r\vert_{\underline{\bm\varphi}={\wt{\bm\varphi}}}, \quad \bm s_i(\bm{\mathcal X}|\bm{\mathcal F};\wt{\bm\varphi})=\sum_{t=1}^T\l.\frac{\partial \ell_t(\mbf x_t|\mbf F_t;\underline{\bm\varphi})}{\partial \underline{\bm\lambda}_i^\prime}\r\vert_{\underline{\bm\varphi}={\wt{\bm\varphi}}}, 
\eeq
which are $r$-dimensional column vectors. 

Then, recalling that $\wh{\bm\Gamma}^x = \frac 1T\sum_{t=1}^T\mbf x_t\mbf x_t^\prime$, by computing the first derivatives of the log-likelihood \eqref{eq:LL0t}, we have
\begin{align}
\bm S(\bm{\mathcal X};\underline{\bm\varphi})=&\, -T (\underline{\bm\Lambda}\,\underline{\bm\Lambda}^\prime + \underline{\bm\Sigma}^\xi)^{-1} \underline{\bm\Lambda}+
(\underline{\bm\Lambda}\,\underline{\bm\Lambda}^\prime + \underline{\bm\Sigma}^\xi)^{-1} 
\wh{\bm\Gamma}^x
(\underline{\bm\Lambda}\,\underline{\bm\Lambda}^\prime + \underline{\bm\Sigma}^\xi)^{-1} \underline{\bm\Lambda}\nn\\
 =&\, -T (\underline{\bm\Sigma}^\xi)^{-1}\underline{\bm\Lambda}\{\mbf I_r+\underline{\bm\Lambda}^\prime(\underline{\bm\Sigma}^\xi)^{-1} \underline{\bm\Lambda}\}^{-1}
+ T (\underline{\bm\Lambda}\,\underline{\bm\Lambda}^\prime + \underline{\bm\Sigma}^\xi)^{-1} \wh{\bm\Gamma}^x(\underline{\bm\Sigma}^\xi)^{-1}\underline{\bm\Lambda}\{\mbf I_r+\underline{\bm\Lambda}^\prime(\underline{\bm\Sigma}^\xi)^{-1} \underline{\bm\Lambda}\}^{-1}\nn\\
 =&\, -T (\underline{\bm\Sigma}^\xi)^{-1}\underline{\bm\Lambda}\{\mbf I_r+\underline{\bm\Lambda}^\prime(\underline{\bm\Sigma}^\xi)^{-1} \underline{\bm\Lambda}\}^{-1}
+ T (\underline{\bm\Sigma}^\xi)^{-1} \wh{\bm\Gamma}^x(\underline{\bm\Sigma}^\xi)^{-1}\underline{\bm\Lambda}\{\mbf I_r+\underline{\bm\Lambda}^\prime(\underline{\bm\Sigma}^\xi)^{-1} \underline{\bm\Lambda}\}^{-1}\nn\\
&-T(\underline{\bm\Sigma}^\xi)^{-1}\underline{\bm\Lambda}
\{\mbf I_r+\underline{\bm\Lambda}^\prime(\underline{\bm\Sigma}^\xi)^{-1} \underline{\bm\Lambda}\}^{-1}
\underline{\bm\Lambda}^\prime
(\underline{\bm\Sigma}^\xi)^{-1}
\wh{\bm\Gamma}^x(\underline{\bm\Sigma}^\xi)^{-1}\underline{\bm\Lambda}\{\mbf I_r+\underline{\bm\Lambda}^\prime(\underline{\bm\Sigma}^\xi)^{-1} \underline{\bm\Lambda}\}^{-1},\label{eq:SX}
\end{align}
where we used the Woodbury identities
\begin{align}
&(\underline{\bm\Lambda}\,\underline{\bm\Lambda}^\prime + \underline{\bm\Sigma}^\xi)^{-1} \underline{\bm\Lambda} = (\underline{\bm\Sigma}^\xi)^{-1}\underline{\bm\Lambda}\{\mbf I_r+\underline{\bm\Lambda}^\prime(\underline{\bm\Sigma}^\xi)^{-1} \underline{\bm\Lambda}\}^{-1},\label{eq:wood1}\\
&(\underline{\bm\Lambda}\,\underline{\bm\Lambda}^\prime + \underline{\bm\Sigma}^\xi)^{-1}= (\underline{\bm\Sigma}^\xi)^{-1}-(\underline{\bm\Sigma}^\xi)^{-1}
\underline{\bm\Lambda}
\{\mbf I_r+\underline{\bm\Lambda}^\prime(\underline{\bm\Sigma}^\xi)^{-1} \underline{\bm\Lambda}\}^{-1}
\underline{\bm\Lambda}^\prime
(\underline{\bm\Sigma}^\xi)^{-1}.\label{eq:wood2}
\end{align}

In what follows, we make use of the following results. Let $\wh{\bm\Gamma}^\xi=\frac 1T\sum_{t=1}^T \bm\xi_t\bm\xi_t^\prime$ and let $\wh{\bm\Sigma}^\xi=\text{diag}(\wh{\bm\Gamma}^\xi)$, the diagonal matrix having as diagonal entries the diagonal entries of $\wh{\bm\Gamma}^\xi$. Then, by using twice Lemma \ref{lem:covarianzeF}(ii) and by Lemma \ref{lem:Gxi}(v), we have
\begin{align}
\frac 1n\l\Vert \bm\Sigma^\xi -\wh{\bm\Gamma}^\xi\r\Vert &\le \frac 1n\l\Vert \bm\Sigma^\xi -  \wh{\bm\Sigma}^\xi\r\Vert +\frac 1n\l\Vert\wh{\bm\Gamma}^\xi-\wh{\bm\Sigma}^\xi \r\Vert \le O_{\mathrm P}\l(\frac 1{\sqrt T}\r) + \frac 1n\l\Vert\wh{\bm\Gamma}^\xi \r\Vert\nn\\
&\le O_{\mathrm P}\l(\frac 1{\sqrt T}\r) + \frac 1n \l\Vert\bm\Gamma^\xi\r\Vert + O_{\mathrm P}\l(\frac 1{\sqrt T}\r) = O_{\mathrm P}\l(\frac 1{\sqrt T}\r) + O\l(\frac 1n\r).
\end{align}
which implies:
\beq\label{eq:anc1}
\l\Vert (\bm\Sigma^\xi)^{-1}\wh{\bm\Gamma}^\xi - \mbf I_n\r\Vert= O_{\mathrm P}\l(\max\l(\frac 1{n},\frac 1{\sqrt T}\r)\r).
\eeq
Moreover,
\beq\label{eq:anc2}
\l\Vert\{\mbf I_r+{\bm\Lambda}^\prime({\bm\Sigma}^\xi)^{-1} {\bm\Lambda}\}^{-1}\l\{{\bm\Lambda}^\prime({\bm\Sigma}^\xi)^{-1}{\bm\Lambda}\r\}-\mbf I_r\r\Vert= O\l(\frac 1n\r),
\eeq
because of  Lemma \ref{lem:woodbury},
\beq\label{eq:anc3}
\l\Vert \l\{{\bm\Lambda}^\prime({\bm\Sigma}^\xi)^{-1}{\bm\Lambda}\r\}^{-1}\r\Vert  
\le \frac 1{n\underline C_r \Vert({\bm\Sigma}^\xi)^{-1}\Vert}=\frac 1{n\underline C_r \min_{i=1,\ldots,n} \sigma_i^2}=\frac 1{n\underline C_r C_\xi}= O\l(\frac 1n\r), 
\eeq
because of Lemma \ref{lem:Gxi}(iv)-\ref{lem:Gxi}(vi), Assumption \ref{ass:idio}(a), and \citet[Theorem 7]{MK04}, 
\begin{align}
\label{eq:anc3bis}
\l\Vert\{\mbf I_r+{\bm\Lambda}^\prime({\bm\Sigma}^\xi)^{-1} {\bm\Lambda}\}^{-1}-\l\{{\bm\Lambda}^\prime({\bm\Sigma}^\xi)^{-1}{\bm\Lambda}\r\}^{-1}\r\Vert= O\l(\frac 1{n^2}\r),\end{align}
because of \eqref{eq:anc2} and \eqref{eq:anc3},  
\begin{align}
&\l\Vert \l\{{\bm\Lambda}^\prime({\bm\Sigma}^\xi)^{-1}{\bm\Lambda}\r\}^{-1}\bm\Lambda^\prime\r\Vert =  O\l(\frac 1{\sqrt n}\r),\label{eq:anc4bis}\\
&\l\Vert\bm\Lambda\{\mbf I_r+{\bm\Lambda}^\prime({\bm\Sigma}^\xi)^{-1} {\bm\Lambda}\}^{-1}-\bm\Lambda\l\{{\bm\Lambda}^\prime({\bm\Sigma}^\xi)^{-1}{\bm\Lambda}\r\}^{-1}\r\Vert= O\l(\frac 1{n^{3/2}}\r),\label{eq:anc4}
\end{align}
because of \eqref{eq:anc3}, \eqref{eq:anc3bis} and since $\Vert {\bm\Lambda}\Vert = O(\sqrt n)$ by Lemma \ref{lem:FTLN}(i),  and, last,
\begin{align}
&\l\Vert
{\bm\Lambda}{\bm\Lambda}^\prime
(\bm\Sigma^\xi)^{-1}\bm\Lambda\l\{\mbf I_r+{\bm\Lambda}^\prime({\bm\Sigma}^\xi)^{-1}{\bm\Lambda}\r\}^{-1}-\bm\Lambda\r\Vert= O\l(\frac 1{\sqrt n}\r),
\label{eq:anc4bis}
\end{align}
because of \eqref{eq:anc2} and since $\Vert {\bm\Lambda}\Vert = O(\sqrt n)$ by Lemma \ref{lem:FTLN}(i).

Now, denote $\wh{\bm\Gamma}^{\xi F}=\frac 1T\sum_{t=1}^T \bm\xi_t\mbf F_t^\prime$ and $\wh{\bm\Gamma}^{F\xi}=\wh{\bm\Gamma}^{\xi F\prime}$, so that we can write
\beq\label{eq:G4}
\wh{\bm\Gamma}^x 
= \bm\Lambda\l(\frac 1T\sum_{t=1}^T\mbf F_t\mbf F_t^\prime\r)\bm\Lambda^\prime + \wh{\bm\Gamma}^\xi+ \bm\Lambda\wh{\bm\Gamma}^{F\xi}+\wh{\bm\Gamma}^{\xi F}\bm\Lambda^\prime
= \bm\Lambda\bm\Lambda^\prime + \wh{\bm\Gamma}^\xi+ \bm\Lambda\wh{\bm\Gamma}^{F\xi}+\wh{\bm\Gamma}^{\xi F}\bm\Lambda^\prime,
\eeq
because of Assumption \ref{ass:ident}(b). Let us consider \eqref{eq:SX} when computed in the true value of the parameters. By means of \eqref{eq:anc1}-\eqref{eq:anc4} and \eqref{eq:G4} we have
\begin{align}
\bm S(\bm{\mathcal X};{\bm\varphi})=&\, T(\bm\Sigma^\xi)^{-1}
\l\{ 
\bm\Lambda\bm\Lambda^\prime 
+ \wh{\bm\Gamma}^\xi
+ \bm\Lambda\wh{\bm\Gamma}^{F\xi}
+\wh{\bm\Gamma}^{\xi F}\bm\Lambda^\prime
\r\} 
(\bm\Sigma^\xi)^{-1}\bm\Lambda\l\{{\bm\Lambda}^\prime({\bm\Sigma}^\xi)^{-1}{\bm\Lambda}\r\}^{-1}\nn\\
&- T(\bm\Sigma^\xi)^{-1}
\bm\Lambda\l\{{\bm\Lambda}^\prime({\bm\Sigma}^\xi)^{-1}{\bm\Lambda}\r\}^{-1}\bm\Lambda^\prime(\bm\Sigma^\xi)^{-1}
\l\{ 
\bm\Lambda\bm\Lambda^\prime + \wh{\bm\Gamma}^\xi+ \bm\Lambda\wh{\bm\Gamma}^{F\xi}+\wh{\bm\Gamma}^{\xi F}\bm\Lambda^\prime
\r\} 
(\bm\Sigma^\xi)^{-1}\bm\Lambda\l\{{\bm\Lambda}^\prime({\bm\Sigma}^\xi)^{-1}{\bm\Lambda}\r\}^{-1}\nn\\
&-T(\bm\Sigma^\xi)^{-1}\bm\Lambda\l\{{\bm\Lambda}^\prime({\bm\Sigma}^\xi)^{-1}{\bm\Lambda}\r\}^{-1}+ O\l(\frac T{\sqrt n}\r)\nn\\
=&\, T(\bm\Sigma^\xi)^{-1} \bm\Lambda\bm\Lambda^\prime (\bm\Sigma^\xi)^{-1}\bm\Lambda\l\{{\bm\Lambda}^\prime({\bm\Sigma}^\xi)^{-1}{\bm\Lambda}\r\}^{-1}+T(\bm\Sigma^\xi)^{-1} \wh{\bm\Gamma}^\xi(\bm\Sigma^\xi)^{-1}\bm\Lambda\l\{{\bm\Lambda}^\prime({\bm\Sigma}^\xi)^{-1}{\bm\Lambda}\r\}^{-1}\nn\\
&+T(\bm\Sigma^\xi)^{-1} \bm\Lambda\wh{\bm\Gamma}^{F\xi}(\bm\Sigma^\xi)^{-1}\bm\Lambda\l\{{\bm\Lambda}^\prime({\bm\Sigma}^\xi)^{-1}{\bm\Lambda}\r\}^{-1}+T(\bm\Sigma^\xi)^{-1}\wh{\bm\Gamma}^{\xi F}\bm\Lambda^\prime
(\bm\Sigma^\xi)^{-1}\bm\Lambda\l\{{\bm\Lambda}^\prime({\bm\Sigma}^\xi)^{-1}{\bm\Lambda}\r\}^{-1}\nn\\
&-T(\bm\Sigma^\xi)^{-1}
\bm\Lambda\l\{{\bm\Lambda}^\prime({\bm\Sigma}^\xi)^{-1}{\bm\Lambda}\r\}^{-1}\bm\Lambda^\prime(\bm\Sigma^\xi)^{-1}\bm\Lambda\bm\Lambda^\prime(\bm\Sigma^\xi)^{-1}\bm\Lambda\l\{{\bm\Lambda}^\prime({\bm\Sigma}^\xi)^{-1}{\bm\Lambda}\r\}^{-1}\nn\\
&-T(\bm\Sigma^\xi)^{-1}
\bm\Lambda\l\{{\bm\Lambda}^\prime({\bm\Sigma}^\xi)^{-1}{\bm\Lambda}\r\}^{-1}\bm\Lambda^\prime(\bm\Sigma^\xi)^{-1}\wh{\bm\Gamma}^\xi(\bm\Sigma^\xi)^{-1}\bm\Lambda\l\{{\bm\Lambda}^\prime({\bm\Sigma}^\xi)^{-1}{\bm\Lambda}\r\}^{-1}\nn\\
&-T(\bm\Sigma^\xi)^{-1}
\bm\Lambda\l\{{\bm\Lambda}^\prime({\bm\Sigma}^\xi)^{-1}{\bm\Lambda}\r\}^{-1}\bm\Lambda^\prime(\bm\Sigma^\xi)^{-1}\bm\Lambda\wh{\bm\Gamma}^{F\xi}(\bm\Sigma^\xi)^{-1}\bm\Lambda\l\{{\bm\Lambda}^\prime({\bm\Sigma}^\xi)^{-1}{\bm\Lambda}\r\}^{-1}\nn\\
&-T(\bm\Sigma^\xi)^{-1}
\bm\Lambda\l\{{\bm\Lambda}^\prime({\bm\Sigma}^\xi)^{-1}{\bm\Lambda}\r\}^{-1}\bm\Lambda^\prime(\bm\Sigma^\xi)^{-1}\wh{\bm\Gamma}^{\xi F}\bm\Lambda^\prime(\bm\Sigma^\xi)^{-1}\bm\Lambda\l\{{\bm\Lambda}^\prime({\bm\Sigma}^\xi)^{-1}{\bm\Lambda}\r\}^{-1}\nn\\
&-T(\bm\Sigma^\xi)^{-1}\bm\Lambda\l\{{\bm\Lambda}^\prime({\bm\Sigma}^\xi)^{-1}{\bm\Lambda}\r\}^{-1}+ O\l(\frac T{\sqrt n}\r)\nn\\
=&\,T(\bm\Sigma^\xi)^{-1} \bm\Lambda+T(\bm\Sigma^\xi)^{-1} \bm\Lambda\l\{{\bm\Lambda}^\prime({\bm\Sigma}^\xi)^{-1}{\bm\Lambda}\r\}^{-1}+T(\bm\Sigma^\xi)^{-1} \bm\Lambda\wh{\bm\Gamma}^{F\xi}(\bm\Sigma^\xi)^{-1}\bm\Lambda\l\{{\bm\Lambda}^\prime({\bm\Sigma}^\xi)^{-1}{\bm\Lambda}\r\}^{-1}\nn\\
&+T(\bm\Sigma^\xi)^{-1}\wh{\bm\Gamma}^{\xi F}-T(\bm\Sigma^\xi)^{-1}\bm\Lambda-T(\bm\Sigma^\xi)^{-1}\bm\Lambda\l\{{\bm\Lambda}^\prime({\bm\Sigma}^\xi)^{-1}{\bm\Lambda}\r\}^{-1}-T(\bm\Sigma^\xi)^{-1}\wh{\bm\Gamma}^{F\xi}(\bm\Sigma^\xi)^{-1}\bm\Lambda\l\{{\bm\Lambda}^\prime({\bm\Sigma}^\xi)^{-1}{\bm\Lambda}\r\}^{-1}\nn\\
&-T(\bm\Sigma^\xi)^{-1}
\bm\Lambda\l\{{\bm\Lambda}^\prime({\bm\Sigma}^\xi)^{-1}{\bm\Lambda}\r\}^{-1}\bm\Lambda^\prime(\bm\Sigma^\xi)^{-1}\wh{\bm\Gamma}^{\xi F}-T(\bm\Sigma^\xi)^{-1}\bm\Lambda\l\{{\bm\Lambda}^\prime({\bm\Sigma}^\xi)^{-1}{\bm\Lambda}\r\}^{-1}\nn\\
&+ O\l(\frac T{\sqrt n}\r)+ O_{\mathrm P}\l(\frac {\sqrt T}{\sqrt n}\r)+ O\l(\frac T{n^{3/2}}\r)\nn\\
=&\,T(\bm\Sigma^\xi)^{-1}\wh{\bm\Gamma}^{\xi F}-T(\bm\Sigma^\xi)^{-1}
\bm\Lambda\l\{{\bm\Lambda}^\prime({\bm\Sigma}^\xi)^{-1}{\bm\Lambda}\r\}^{-1}\bm\Lambda^\prime(\bm\Sigma^\xi)^{-1}\wh{\bm\Gamma}^{\xi F}\nn\\
&-T(\bm\Sigma^\xi)^{-1}\bm\Lambda\l\{{\bm\Lambda}^\prime({\bm\Sigma}^\xi)^{-1}{\bm\Lambda}\r\}^{-1}+ O\l(\frac T{\sqrt n}\r)+ O_{\mathrm P}\l(\frac {\sqrt T}{\sqrt n}\r)+ O\l(\frac T{n^{3/2}}\r).\label{eq:finquituttobene}
\end{align}

Then, notice that $[(\bm\Sigma^\xi)^{-1}]_{i\cdot}\bm\xi_t= [(\bm\Sigma^\xi)^{-1}]_{ii} \xi_{it}= \frac{ \xi_{it}}{\sigma_i^{2}}$, and
$[(\bm\Sigma^\xi)^{-1}]_{i\cdot}\bm\Lambda = \frac{\bm\lambda_i^\prime}{\sigma_i^{2} }$, $i=1,\ldots, n$.
The following holds
\begin{align}
& \frac 1{\sigma_i^2}{\bm\lambda}_i^\prime \l\{{\bm\Lambda}^\prime({\bm\Sigma}^\xi)^{-1}{\bm\Lambda}\r\}^{-1}= O\l(\frac 1n\r),\label{eq:anc5}\\
&T\l[(\bm\Sigma^\xi)^{-1}\wh{\bm\Gamma}^\xi(\bm\Sigma^\xi)^{-1}\bm\Lambda\l\{{\bm\Lambda}^\prime({\bm\Sigma}^\xi)^{-1}{\bm\Lambda}\r\}^{-1}
\r]_{i\cdot} =  \frac T{\sigma_i^2}{\bm\lambda}_i^\prime \l\{{\bm\Lambda}^\prime({\bm\Sigma}^\xi)^{-1}{\bm\Lambda}\r\}^{-1}+ O_{\mathrm P}\l(\frac{\sqrt T}{n} \r)+ O\l(\frac{ T}{n^2} \r),\label{eq:anc6}\\
&T\l[
(\bm\Sigma^\xi)^{-1}\bm\Lambda \l\{{\bm\Lambda}^\prime({\bm\Sigma}^\xi)^{-1}{\bm\Lambda}\r\}^{-1}\bm\Lambda^\prime(\bm\Sigma^\xi)^{-1}\wh{\bm\Gamma}^\xi(\bm\Sigma^\xi)^{-1}\bm\Lambda\l\{{\bm\Lambda}^\prime({\bm\Sigma}^\xi)^{-1}{\bm\Lambda}\r\}^{-1}
\r]_{i\cdot}=\frac T{\sigma_i^2}{\bm\lambda}_i^\prime \l\{{\bm\Lambda}^\prime({\bm\Sigma}^\xi)^{-1}{\bm\Lambda}\r\}^{-1}\nn\\
&+ O_{\mathrm P}\l(\frac{\sqrt T}{n} \r)+ O\l(\frac{ T}{n^2} \r),\label{eq:anc7}\\
&T\bm\lambda_i^\prime\{\mbf I_r+{\bm\Lambda}^\prime({\bm\Sigma}^\xi)^{-1} {\bm\Lambda}\}^{-1}=T\bm\lambda_i^\prime\l\{{\bm\Lambda}^\prime({\bm\Sigma}^\xi)^{-1}{\bm\Lambda}\r\}^{-1}+ O\l(\frac T{n^{2}}\r),\label{eq:anc8}\\
&T{\bm\lambda}_i^\prime{\bm\Lambda}^\prime
(\bm\Sigma^\xi)^{-1}\bm\Lambda\l\{\mbf I_r+{\bm\Lambda}^\prime({\bm\Sigma}^\xi)^{-1}{\bm\Lambda}\r\}^{-1}=T\bm\lambda_i^\prime + O\l(\frac T{n}\r),\label{eq:anc9}
\end{align}
where we used Assumption \ref{ass:common}(a) and then \eqref{eq:anc5} follows from \eqref{eq:anc3}, \eqref{eq:anc6}  and \eqref{eq:anc7} follow from \eqref{eq:anc1} and \eqref{eq:anc5}, \eqref{eq:anc8} follows from \eqref{eq:anc3bis}, and \eqref{eq:anc9} follows from and \eqref{eq:anc4bis}.

Therefore, by using \eqref{eq:anc6}-\eqref{eq:anc9} in \eqref{eq:finquituttobene} we have that the $i$th row of $\bm S(\bm{\mathcal X};\bm{\varphi})$ is such that
\begin{align}
\bm s_i^\prime(\bm{\mathcal X};\bm{\varphi})=&\,\frac 1{\sigma_i^2}\sum_{t=1}^T {\xi}_{it}\mbf F_t^\prime-\frac 1{\sigma_i^2}\bm\lambda_i^\prime
\l\{{\bm\Lambda}^\prime({\bm\Sigma}^\xi)^{-1}{\bm\Lambda}\r\}^{-1}{\bm\Lambda}^\prime(\bm\Sigma^\xi)^{-1}\sum_{t=1}^T {\bm\xi}_{t}\mbf F_t^\prime\nn\\
&-\frac T{\sigma_i^2}{\bm\lambda}_i^\prime\l\{{\bm\Lambda}^\prime({\bm\Sigma}^\xi)^{-1}{\bm\Lambda}\r\}^{-1}+ O\l(\frac T{ n}\r)+ O_{\mathrm P}\l(\frac {\sqrt T}{n}\r)+ O\l(\frac T{ n^2}\r).\label{eq:SXi}
\end{align}
Hence, $\Vert\bm s_i^\prime(\bm{\mathcal X};\bm{\varphi})\Vert=O_{\mathrm P}(\sqrt T)$, because of Assumption \ref{ass:idio}(a), \eqref{eq:anc5} and since 
\begin{align}
\E\l[\l\Vert \sum_{t=1}^T \frac{\mbf F_t\bm\xi_t^\prime \bm\Lambda}{ nT}\r\Vert^2\r] = O\l(\frac 1{nT}\r) \;\text{ and }\; \E\l[\l\Vert \sum_{t=1}^T \frac{\mbf F_t\bm\xi_t^\prime }{\sqrt nT}\r\Vert^2\r] = O\l(\frac 1{\sqrt T}\r)
,\label{eq:2a3bis}
\end{align}
because of  \eqref{eq:1a3} in the proof of Proposition \ref{prop:load}(a) and by Lemma \ref{lem:LLN}.
Moreover, from \eqref{eq:SXi} and by noticing that $(\bm\Sigma^\xi)^{-1}{\bm\Lambda}$ has the same properties as ${\bm\Lambda}$ since $\Vert (\bm\Sigma^\xi)^{-1}\Vert = O(1)$ by Assumption \ref{ass:idio}(a), we get
\begin{align}
\frac 1{\sqrt T}\l\Vert
\bm s_i^\prime(\bm{\mathcal X};\bm{\varphi})-\frac 1{\sigma_i^2}\sum_{t=1}^T {\xi}_{it}\mbf F_t^\prime
\r\Vert\le &\,
\frac 1{C_\xi}\l\Vert 
{\bm\lambda}_i^\prime
\l\{{\bm\Lambda}^\prime({\bm\Sigma}^\xi)^{-1}{\bm\Lambda}\r\}^{-1}
\r\Vert\,
\l\Vert\frac 1{\sqrt T} {\bm\Lambda}^\prime(\bm\Sigma^\xi)^{-1}\sum_{t=1}^T {\bm\xi}_{t}\mbf F_t^\prime
\r\Vert\nn\\
&+ 
\frac{\sqrt T}{C_\xi} \l\Vert
{\bm\lambda}_i^\prime\l\{{\bm\Lambda}^\prime({\bm\Sigma}^\xi)^{-1}{\bm\Lambda}\r\}^{-1}
\r\Vert+ O\l(\frac {\sqrt T}{n}\r)+ O_{\mathrm P}\l(\frac {1}{n}\r)+ O\l(\frac {\sqrt T}{n^2}\r)\nn\\
=&\,  O\l(\frac 1n\r) O_{\mathrm P}\l(\sqrt {n}\r)+  O\l(\frac {\sqrt T}{n}\r)+O_{\mathrm P}\l(\frac 1{n}\r)+ O\l(\frac {\sqrt T}{n^2}\r)\nn\\
=&\, O_{\mathrm P}\l(\max\l( \frac 1{\sqrt n},\frac{\sqrt T} {n}\r)\r),\label{eq:GOL}
\end{align}
.

Finally, by computing the first derivatives of the log-likelihood \eqref{eq:LL0XFt}, we have
\begin{align}
\bm S(\bm{\mathcal X}|\bm{\mathcal F};\underline{\bm\varphi})&= (\underline{\bm\Sigma}^\xi)^{-1}\sum_{t=1}^T (\mbf x_t-\underline{\bm\Lambda}\mbf F_t)\mbf F_t^\prime.\label{eq:SXF}
\end{align}
And \eqref{eq:SXF} when computed in the true value of the parameters is
\begin{align}
\bm S(\bm{\mathcal X}|\bm{\mathcal F};{\bm\varphi})&= ({\bm\Sigma}^\xi)^{-1}\sum_{t=1}^T \bm\xi_t\mbf F_t^\prime.\label{eq:ancorabene}
\end{align}
From \eqref{eq:ancorabene}, the $i$th row of $\bm S(\bm{\mathcal X}|\bm{\mathcal F};{\bm\varphi})$ is then:
\beq
\bm s_i^\prime(\bm{\mathcal X}|\bm{\mathcal F};{\bm\varphi})=\frac 1{\sigma_i^2}\sum_{t=1}^T {\xi}_{it}\mbf F_t^\prime.\label{eq:SFXi}
\eeq
Hence, $\Vert \bm s_i^\prime(\bm{\mathcal X}|\bm{\mathcal F};{\bm\varphi})\Vert = O_{\mathrm P}(\sqrt T)$ because of \eqref{eq:2a3bis},  and, by using \eqref{eq:SFXi} in \eqref{eq:GOL}, for any given $i=1,\ldots, n$ we have:
\[
\frac 1{\sqrt T}\l\Vert
\bm s_i^\prime(\bm{\mathcal X};\bm{\varphi})-\bm s_i^\prime(\bm{\mathcal X}|\bm{\mathcal F};{\bm\varphi})
\r\Vert = O_{\mathrm P}\l(\max\l( \frac 1{\sqrt n},\frac{\sqrt T} {n}\r)\r),
\]
which proves part (a).\smallskip

Turning to part (b).  For any specific value of the parameters, say $\wt{\bm\varphi}$, let
\begin{align}
\bm H(\bm{\mathcal X};\wt{\bm\varphi}) &= \sum_{t=1}^T\l.\frac{\partial^2 \ell_t(\mbf x_t;\underline{\bm\varphi})}{\partial \text{vec}(\underline{\bm\Lambda})^\prime\partial \text{vec}(\underline{\bm\Lambda})}\r\vert_{\underline{\bm\varphi}={\wt{\bm\varphi}}} = 
\l.
\frac{\partial \text{vec}(\bm S(\bm{\mathcal X};\underline{\bm\varphi}))^\prime}{\partial \text{vec}(\underline{\bm\Lambda})^\prime}
\r\vert_{\underline{\bm\varphi}={\wt{\bm\varphi}}}
= \l(\ba{cccc}
\bm h_{11}(\bm{\mathcal X};\wt{\bm\varphi})&\ldots& \bm h_{1n}(\bm{\mathcal X};\wt{\bm\varphi})\\
\vdots&\ddots&\vdots\\
\bm h_{n1}(\bm{\mathcal X};\wt{\bm\varphi})&\ldots& \bm h_{nn}(\bm{\mathcal X};\wt{\bm\varphi})\\
\ea
\r),\nn\\
\bm H(\bm{\mathcal X}|\bm{\mathcal F};\wt{\bm\varphi}) &= \sum_{t=1}^T\l.\frac{\partial^2 \ell_t(\mbf x_t|\mbf F_t;\underline{\bm\varphi})}{\partial \text{vec}(\underline{\bm\Lambda})^\prime\partial \text{vec}(\underline{\bm\Lambda})}\r\vert_{\underline{\bm\varphi}={\wt{\bm\varphi}}} = 
\l.
\frac{\partial \text{vec}(\bm S(\bm{\mathcal X}|\bm{\mathcal F};\underline{\bm\varphi}))^\prime}{\partial \text{vec}(\underline{\bm\Lambda})^\prime}
\r\vert_{\underline{\bm\varphi}={\wt{\bm\varphi}}}
= \l(\ba{cccc}
\bm h_{11}(\bm{\mathcal X}|\bm{\mathcal F};\wt{\bm\varphi})&\ldots& \bm h_{1n}(\bm{\mathcal X}|\bm{\mathcal F};\wt{\bm\varphi})\\
\vdots&\ddots&\vdots\\
\bm h_{n1}(\bm{\mathcal X}|\bm{\mathcal F};\wt{\bm\varphi})&\ldots& \bm h_{nn}(\bm{\mathcal X}|\bm{\mathcal F};\wt{\bm\varphi})\\
\ea
\r),\nn
\end{align}
which are $nr\times nr$ matrices obtained from the matricization of the 4th order tensor of second derivatives, and where, for any given $ i=1,\ldots, n$,
\begin{align}
\bm h_{ii}(\bm{\mathcal X};\wt{\bm\varphi})&= \sum_{t=1}^T\l.\frac{\partial^2 \ell_t(\mbf x_t;\underline{\bm\varphi})}{\partial \underline{\bm\lambda}_i^\prime\partial\underline{\bm\lambda}_i}\r\vert_{\underline{\bm\varphi}={\wt{\bm\varphi}}} = 
\l.
\frac{\partial \bm s_i^\prime(\bm{\mathcal X};\underline{\bm\varphi})}{\partial \underline{\bm\lambda}_i^\prime}
\r\vert_{\underline{\bm\varphi}={\wt{\bm\varphi}}},\label{eq:defhess}\\
\bm h_{ii}(\bm{\mathcal X}|\bm{\mathcal F};\wt{\bm\varphi})&= \sum_{t=1}^T\l.\frac{\partial^2 \ell_t(\mbf x_t|\mbf F_t;\underline{\bm\varphi})}{\partial \underline{\bm\lambda}_i^\prime\partial \underline{\bm\lambda}_i}\r\vert_{\underline{\bm\varphi}={\wt{\bm\varphi}}} = 
\l.
\frac{\partial \bm s_i^\prime(\bm{\mathcal X}|\bm{\mathcal F};\underline{\bm\varphi})}{\partial \underline{\bm\lambda}_i^\prime}
\r\vert_{\underline{\bm\varphi}={\wt{\bm\varphi}}},\nn
\end{align}
which are $r\times r$ matrices.

We now use the following relation:
\begin{align}
\text{vec}\l(\mathrm d \bm S^\prime(\bm{\mathcal X};\underline{\bm\varphi})\r) &=\l(
\frac{\partial \text{vec}(\bm S(\bm{\mathcal X};\underline{\bm\varphi}))^\prime}{\partial \text{vec}(\underline{\bm\Lambda})^\prime}\r)^\prime \text{vec}(\mathrm d \underline{\bm \Lambda}^\prime)= 
\bm H(\bm{\mathcal X};\underline{\bm\varphi}) 
\text{vec}(\mathrm d \underline{\bm \Lambda}^\prime),\label{eq:SHL}
\end{align}
Then, by denoting $\underline{\bm P}=\l\{\mbf I_r+\underline{\bm\Lambda}^\prime(\underline{\bm\Sigma}^\xi)^{-1}\underline{\bm\Lambda}\r\}$, from \eqref{eq:SX} we have
\begin{align}\label{eq:diffS}
\text{vec}\l(\mathrm d \bm S^\prime(\bm{\mathcal X};\underline{\bm\varphi})\r)=&\, 
-T
\l\{
(\underline{\bm\Sigma}^\xi)^{-1}\underline{\bm \Lambda}\otimes \mbf I_r
\r\}
\text{vec}\l(\mathrm d \underline{\bm P}^{-1}\r)\nn\\
&-T
\l\{
(\underline{\bm\Sigma}^\xi)^{-1}\otimes \underline{\bm P}^{-1}
\r\}
\text{vec}(\mathrm d \underline{\bm \Lambda}^\prime)\nn\\
&+T
\l\{
(\underline{\bm\Sigma}^\xi)^{-1}\wh{\bm\Gamma}^x (\underline{\bm\Sigma}^\xi)^{-1}\underline{\bm \Lambda}\otimes \mbf I_r
\r\}
\text{vec}\l(\mathrm d \underline{\bm P}^{-1}\r)\nn\\
&+T
\l\{
(\underline{\bm\Sigma}^\xi)^{-1}\wh{\bm\Gamma}^x (\underline{\bm\Sigma}^\xi)^{-1}\otimes \underline{\bm P}^{-1}
\r\}
\text{vec}(\mathrm d \underline{\bm \Lambda}^\prime)\nn\\
&-T
\l\{
(\underline{\bm\Sigma}^\xi)^{-1}\underline{\bm \Lambda}\,\underline{\bm P}^{-1}\underline{\bm \Lambda}^\prime(\underline{\bm\Sigma}^\xi)^{-1}\wh{\bm\Gamma}^x (\underline{\bm\Sigma}^\xi)^{-1}\underline{\bm \Lambda}
\otimes 
\mbf I_r
\r\}
\text{vec}\l(\mathrm d \underline{\bm P}^{-1}\r)\nn\\
&-T
\l\{
(\underline{\bm\Sigma}^\xi)^{-1}\underline{\bm \Lambda}\,\underline{\bm P}^{-1}\underline{\bm \Lambda}^\prime(\underline{\bm\Sigma}^\xi)^{-1}\wh{\bm\Gamma}^x (\underline{\bm\Sigma}^\xi)^{-1}
\otimes 
\underline{\bm P}^{-1}
\r\}
\text{vec}(\mathrm d \underline{\bm \Lambda}^\prime)\nn\\
&-T
\l\{
(\underline{\bm\Sigma}^\xi)^{-1}\underline{\bm \Lambda}\,\underline{\bm P}^{-1}
\otimes 
\underline{\bm P}^{-1}\underline{\bm \Lambda}^\prime
(\underline{\bm\Sigma}^\xi)^{-1}\wh{\bm\Gamma}^x (\underline{\bm\Sigma}^\xi)^{-1}
\r\}
\bm C_{n,r}
\text{vec}(\mathrm d \underline{\bm \Lambda}^\prime)\nn\\
&-T
\l\{
(\underline{\bm\Sigma}^\xi)^{-1}\underline{\bm \Lambda}
\otimes
\underline{\bm P}^{-1}\underline{\bm \Lambda}^\prime
(\underline{\bm\Sigma}^\xi)^{-1}\wh{\bm\Gamma}^x (\underline{\bm\Sigma}^\xi)^{-1}\underline{\bm \Lambda}
\r\}
\text{vec}\l(\mathrm d \underline{\bm P}^{-1}\r)\nn\\
&-T
\l\{(\underline{\bm\Sigma}^\xi)^{-1}
\otimes
\underline{\bm P}^{-1}\underline{\bm \Lambda}^\prime
(\underline{\bm\Sigma}^\xi)^{-1}\wh{\bm\Gamma}^x (\underline{\bm\Sigma}^\xi)^{-1}\underline{\bm \Lambda}\underline{\bm P}^{-1}
\r\}
\text{vec}(\mathrm d \underline{\bm \Lambda}^\prime).
\end{align}
Moreover,
\begin{align}
\text{vec}\l(\mathrm d \underline{\bm P}^{-1}\r) &= -\l( \underline{\bm P}^{-1}\otimes  \underline{\bm P}^{-1} \r)
\l\{
\l[
\underline{\bm \Lambda}^\prime(\underline{\bm\Sigma}^\xi)^{-1}\otimes \mbf I_r
\r]
+
\l[
\mbf I_r \otimes \underline{\bm \Lambda}^\prime(\underline{\bm\Sigma}^\xi)^{-1}
\r]\bm C_{n,r}
\r\}
\text{vec}(\mathrm d \underline{\bm \Lambda}^\prime)\nn\\
&= - \l\{
\l[
\underline{\bm P}^{-1}\underline{\bm \Lambda}^\prime(\underline{\bm\Sigma}^\xi)^{-1}\otimes \underline{\bm P}^{-1}
\r]
+\l[
\underline{\bm P}^{-1}\otimes\underline{\bm P}^{-1}\underline{\bm \Lambda}^\prime(\underline{\bm\Sigma}^\xi)^{-1}
\r]\bm C_{n,r}
\r\}
\text{vec}(\mathrm d \underline{\bm \Lambda}^\prime),\label{eq:diffP}
\end{align}
where $\bm C_{n,r}$ is the $nr\times nr$ commutation matrix such that $\text{vec}(\mathrm d \underline{\bm \Lambda})=\bm C_{n,r}\text{vec}(\mathrm d \underline{\bm \Lambda}^\prime)$.
Therefore, from \eqref{eq:SHL}, \eqref{eq:diffP}, and \eqref{eq:diffS} we get
\begin{align}\label{eq:HX}
\bm H(\bm{\mathcal X};\underline{\bm\varphi}) =&\,
T\l[
(\underline{\bm\Sigma}^\xi)^{-1}\underline{\bm \Lambda}\,\underline{\bm P}^{-1} \underline{\bm \Lambda}^\prime (\underline{\bm\Sigma}^\xi)^{-1}
\otimes 
\underline{\bm P}^{-1}
\r]\quad \text{\it A.1}\nn\\
&-T\l[
(\underline{\bm\Sigma}^\xi)^{-1}\wh{\bm\Gamma}^x (\underline{\bm\Sigma}^\xi)^{-1}\underline{\bm \Lambda}\,\underline{\bm P}^{-1} \underline{\bm \Lambda}^\prime (\underline{\bm\Sigma}^\xi)^{-1}
\otimes 
\underline{\bm P}^{-1}
\r]\quad \text{\it A.2}\nn\\
&+T\l[
(\underline{\bm\Sigma}^\xi)^{-1}\underline{\bm \Lambda}\,\underline{\bm P}^{-1} \underline{\bm \Lambda}^\prime (\underline{\bm\Sigma}^\xi)^{-1}\wh{\bm\Gamma}^x (\underline{\bm\Sigma}^\xi)^{-1}\underline{\bm \Lambda}\,\underline{\bm P}^{-1} \underline{\bm \Lambda}^\prime (\underline{\bm\Sigma}^\xi)^{-1}
\otimes 
\underline{\bm P}^{-1}
\r]\quad \text{\it A.3}\nn\\
&+T\l[
(\underline{\bm\Sigma}^\xi)^{-1}\underline{\bm \Lambda}\,\underline{\bm P}^{-1} \underline{\bm \Lambda}^\prime (\underline{\bm\Sigma}^\xi)^{-1}
\otimes 
\underline{\bm P}^{-1}\underline{\bm \Lambda}^\prime(\underline{\bm\Sigma}^\xi)^{-1}\wh{\bm\Gamma}^x (\underline{\bm\Sigma}^\xi)^{-1}\underline{\bm \Lambda}\,\underline{\bm P}^{-1}
\r]\quad \text{\it A.4}\nn\\
&+T\l[
(\underline{\bm\Sigma}^\xi)^{-1}\underline{\bm \Lambda}\,\underline{\bm P}^{-1}
\otimes 
\underline{\bm P}^{-1}\underline{\bm \Lambda}^\prime(\underline{\bm\Sigma}^\xi)^{-1}
\r]
\bm C_{n,r}\quad \text{\it A.5}\nn\\
&-T\l[
(\underline{\bm\Sigma}^\xi)^{-1}\wh{\bm\Gamma}^x (\underline{\bm\Sigma}^\xi)^{-1}\underline{\bm \Lambda}\,\underline{\bm P}^{-1}
\otimes 
\underline{\bm P}^{-1}\underline{\bm \Lambda}^\prime(\underline{\bm\Sigma}^\xi)^{-1}
\r]
\bm C_{n,r}\quad \text{\it A.6}\nn\\
&+T\l[
(\underline{\bm\Sigma}^\xi)^{-1}\underline{\bm \Lambda}\,\underline{\bm P}^{-1}\underline{\bm \Lambda}^\prime(\underline{\bm\Sigma}^\xi)^{-1}\wh{\bm\Gamma}^x (\underline{\bm\Sigma}^\xi)^{-1}\underline{\bm \Lambda}\,\underline{\bm P}^{-1}
\otimes 
\underline{\bm P}^{-1}\underline{\bm \Lambda}^\prime(\underline{\bm\Sigma}^\xi)^{-1}
\r]
\bm C_{n,r}\quad \text{\it A.7}\nn\\
&+T\l[
(\underline{\bm\Sigma}^\xi)^{-1}\underline{\bm \Lambda}\,\underline{\bm P}^{-1}
\otimes 
\underline{\bm P}^{-1}\underline{\bm \Lambda}^\prime(\underline{\bm\Sigma}^\xi)^{-1}\wh{\bm\Gamma}^x (\underline{\bm\Sigma}^\xi)^{-1} \underline{\bm \Lambda}\,\underline{\bm P}^{-1} \underline{\bm \Lambda}^\prime(\underline{\bm\Sigma}^\xi)^{-1}
\r]
\bm C_{n,r}\quad \text{\it A.8}\nn\\
&-T\l[
(\underline{\bm\Sigma}^\xi)^{-1}
\otimes 
\underline{\bm P}^{-1} 
\r]\quad \text{\it B.1}\nn\\
&+T\l[
(\underline{\bm\Sigma}^\xi)^{-1}\wh{\bm\Gamma}^x (\underline{\bm\Sigma}^\xi)^{-1}
\otimes 
\underline{\bm P}^{-1} 
\r]\quad \text{\it B.2}\nn\\
&-T\l[
(\underline{\bm\Sigma}^\xi)^{-1} \underline{\bm \Lambda}\,\underline{\bm P}^{-1} \underline{\bm \Lambda}^\prime(\underline{\bm\Sigma}^\xi)^{-1}\wh{\bm\Gamma}^x (\underline{\bm\Sigma}^\xi)^{-1}
\otimes 
\underline{\bm P}^{-1} 
\r]\quad \text{\it B.3}\nn\\
&-T\l[
(\underline{\bm\Sigma}^\xi)^{-1} \underline{\bm \Lambda}\,\underline{\bm P}^{-1} 
\otimes 
\underline{\bm P}^{-1} \underline{\bm \Lambda}^\prime(\underline{\bm\Sigma}^\xi)^{-1}\wh{\bm\Gamma}^x (\underline{\bm\Sigma}^\xi)^{-1}
\r]\bm C_{n,r}\quad \text{\it B.4}\nn\\
&-T\l[
(\underline{\bm\Sigma}^\xi)^{-1}
\otimes 
\underline{\bm P}^{-1} \underline{\bm \Lambda}^\prime(\underline{\bm\Sigma}^\xi)^{-1}\wh{\bm\Gamma}^x (\underline{\bm\Sigma}^\xi)^{-1}\underline{\bm \Lambda}\,\underline{\bm P}^{-1} 
\r].\quad \text{\it B.5}
\end{align}
Let us consider \eqref{eq:HX} when computed in the true value of the parameters. By means of \eqref{eq:G4} we have: 
{
\begin{align}
\bm H(\bm{\mathcal X};\bm\varphi)=&\, 
T\l[(\bm\Sigma^\xi)^{-1}\bm\Lambda \bm P^{-1}\bm\Lambda^\prime(\bm\Sigma^\xi)^{-1}
\otimes
 \bm P^{-1} 
\r]\quad \text{\it A.1}
\nn\\
&-T\l[
(\bm\Sigma^\xi)^{-1}\bm\Lambda\bm\Lambda^\prime(\bm\Sigma^\xi)^{-1}\bm\Lambda\bm P^{-1}\bm\Lambda^\prime(\bm\Sigma^\xi)^{-1}
\otimes
 \bm P^{-1} 
\r]\quad \text{\it A.2.1}
\nn\\
&-T\l[
(\bm\Sigma^\xi)^{-1}\wh{\bm\Gamma}^\xi(\bm\Sigma^\xi)^{-1}\bm\Lambda\bm P^{-1}\bm\Lambda^\prime(\bm\Sigma^\xi)^{-1}
\otimes
 \bm P^{-1} 
\r]\quad \text{\it A.2.2}
\nn\\
&-T\l[
(\bm\Sigma^\xi)^{-1}\bm\Lambda\wh{\bm\Gamma}^{F\xi}(\bm\Sigma^\xi)^{-1}\bm\Lambda\bm P^{-1}\bm\Lambda^\prime(\bm\Sigma^\xi)^{-1}
\otimes
 \bm P^{-1} 
\r]\quad \text{\it A.2.3}
\nn\\
&-T\l[
(\bm\Sigma^\xi)^{-1}\wh{\bm\Gamma}^{\xi F}\bm\Lambda^\prime(\bm\Sigma^\xi)^{-1}\bm\Lambda\bm P^{-1}\bm\Lambda^\prime(\bm\Sigma^\xi)^{-1}
\otimes
 \bm P^{-1} 
\r]\quad \text{\it A.2.4}
\nn\\
&+T\l[
(\bm\Sigma^\xi)^{-1} \bm\Lambda\bm P^{-1}\bm\Lambda^\prime(\bm\Sigma^\xi)^{-1} \bm\Lambda\bm\Lambda^\prime(\bm\Sigma^\xi)^{-1}\bm\Lambda\bm P^{-1}\bm\Lambda^\prime(\bm\Sigma^\xi)^{-1}
\otimes 
\bm P^{-1} 
\r]\quad \text{\it A.3.1}
\nn\\
&+T\l[
(\bm\Sigma^\xi)^{-1} \bm\Lambda\bm P^{-1}\bm\Lambda^\prime(\bm\Sigma^\xi)^{-1} \wh{\bm\Gamma}^\xi(\bm\Sigma^\xi)^{-1}\bm\Lambda\bm P^{-1}\bm\Lambda^\prime(\bm\Sigma^\xi)^{-1}
\otimes 
\bm P^{-1} 
\r]\quad \text{\it A.3.2}
\nn\\
&+T\l[
(\bm\Sigma^\xi)^{-1} \bm\Lambda\bm P^{-1}\bm\Lambda^\prime(\bm\Sigma^\xi)^{-1} \bm\Lambda\wh{\bm\Gamma}^{F\xi}(\bm\Sigma^\xi)^{-1}\bm\Lambda\bm P^{-1}\bm\Lambda^\prime(\bm\Sigma^\xi)^{-1}
\otimes 
\bm P^{-1} 
\r]\quad \text{\it A.3.3}
\nn\\
&+T\l[
(\bm\Sigma^\xi)^{-1} \bm\Lambda\bm P^{-1}\bm\Lambda^\prime(\bm\Sigma^\xi)^{-1} \wh{\bm\Gamma}^{\xi F}\bm\Lambda^\prime(\bm\Sigma^\xi)^{-1}\bm\Lambda\bm P^{-1}\bm\Lambda^\prime(\bm\Sigma^\xi)^{-1}
\otimes 
\bm P^{-1} 
\r]\quad \text{\it A.3.4}
\nn\\
&+T\l[
(\bm\Sigma^\xi)^{-1} \bm\Lambda\bm P^{-1}\bm\Lambda^\prime(\bm\Sigma^\xi)^{-1}
\otimes 
\bm P^{-1} \bm\Lambda^\prime(\bm\Sigma^\xi)^{-1}\bm\Lambda\bm\Lambda^\prime(\bm\Sigma^\xi)^{-1}\bm\Lambda\bm P^{-1}
\r]\quad \text{\it A.4.1}
\nn\\
&+T\l[
(\bm\Sigma^\xi)^{-1} \bm\Lambda\bm P^{-1}\bm\Lambda^\prime(\bm\Sigma^\xi)^{-1}
\otimes 
\bm P^{-1} \bm\Lambda^\prime(\bm\Sigma^\xi)^{-1}\wh{\bm\Gamma}^\xi(\bm\Sigma^\xi)^{-1}\bm\Lambda\bm P^{-1}
\r]\quad \text{\it A.4.2}
\nn\\
&+T\l[
(\bm\Sigma^\xi)^{-1} \bm\Lambda\bm P^{-1}\bm\Lambda^\prime(\bm\Sigma^\xi)^{-1}
\otimes 
\bm P^{-1} \bm\Lambda^\prime(\bm\Sigma^\xi)^{-1}\bm\Lambda\wh{\bm\Gamma}^{F\xi}(\bm\Sigma^\xi)^{-1}\bm\Lambda\bm P^{-1}
\r]\quad \text{\it A.4.3}
\nn\\
&+T\l[
(\bm\Sigma^\xi)^{-1} \bm\Lambda\bm P^{-1}\bm\Lambda^\prime(\bm\Sigma^\xi)^{-1}
\otimes 
\bm P^{-1} \bm\Lambda^\prime(\bm\Sigma^\xi)^{-1}\wh{\bm\Gamma}^{\xi F}\bm\Lambda^\prime(\bm\Sigma^\xi)^{-1}\bm\Lambda\bm P^{-1}
\r]\quad \text{\it A.4.4}
\nn\\
&+T\l[
(\bm\Sigma^\xi)^{-1}\bm\Lambda\bm P^{-1}
\otimes 
\bm P^{-1} \bm\Lambda^\prime(\bm\Sigma^\xi)^{-1}
\r]\bm C_{n,r}\quad \text{\it A.5}
\nn\\
&-T\l[
(\bm\Sigma^\xi)^{-1}\bm\Lambda\bm\Lambda^\prime(\bm\Sigma^\xi)^{-1}\bm\Lambda\bm P^{-1}
\otimes 
\bm P^{-1} \bm\Lambda^\prime (\bm\Sigma^\xi)^{-1}
\r]\bm C_{n,r}\quad \text{\it A.6.1}
\nn\\
&-T\l[
(\bm\Sigma^\xi)^{-1}\wh{\bm\Gamma}^\xi(\bm\Sigma^\xi)^{-1}\bm\Lambda\bm P^{-1}
\otimes 
\bm P^{-1} \bm\Lambda^\prime (\bm\Sigma^\xi)^{-1}
\r]\bm C_{n,r}\quad \text{\it A.6.2}
\nn\\
&-T\l[
(\bm\Sigma^\xi)^{-1}\bm\Lambda\wh{\bm\Gamma}^{F\xi}(\bm\Sigma^\xi)^{-1}\bm\Lambda\bm P^{-1}
\otimes 
\bm P^{-1} \bm\Lambda^\prime (\bm\Sigma^\xi)^{-1}
\r]\bm C_{n,r}\quad \text{\it A.6.3}
\nn\\
&-T\l[
(\bm\Sigma^\xi)^{-1}\wh{\bm\Gamma}^{\xi F}\bm\Lambda^\prime(\bm\Sigma^\xi)^{-1}\bm\Lambda\bm P^{-1}
\otimes 
\bm P^{-1} \bm\Lambda^\prime (\bm\Sigma^\xi)^{-1}
\r]\bm C_{n,r}\quad \text{\it A.6.4}
\nn\\
&+T\l[
(\bm\Sigma^\xi)^{-1}\bm\Lambda\bm P^{-1}\bm\Lambda^\prime(\bm\Sigma^\xi)^{-1}\bm\Lambda\bm\Lambda^\prime(\bm\Sigma^\xi)^{-1}\bm\Lambda\bm P^{-1}
\otimes 
\bm P^{-1} \bm\Lambda^\prime (\bm\Sigma^\xi)^{-1}
\r]\bm C_{n,r}\quad \text{\it A.7.1}
\nn\\
&+T\l[
(\bm\Sigma^\xi)^{-1}\bm\Lambda\bm P^{-1}\bm\Lambda^\prime(\bm\Sigma^\xi)^{-1}\wh{\bm\Gamma}^\xi(\bm\Sigma^\xi)^{-1}\bm\Lambda\bm P^{-1}
\otimes 
\bm P^{-1} \bm\Lambda^\prime (\bm\Sigma^\xi)^{-1}
\r]\bm C_{n,r}\quad \text{\it A.7.2}
\nn\\
&+T\l[
(\bm\Sigma^\xi)^{-1}\bm\Lambda\bm P^{-1}\bm\Lambda^\prime(\bm\Sigma^\xi)^{-1}\bm\Lambda\wh{\bm\Gamma}^{F\xi}(\bm\Sigma^\xi)^{-1}\bm\Lambda\bm P^{-1}
\otimes 
\bm P^{-1} \bm\Lambda^\prime (\bm\Sigma^\xi)^{-1}
\r]\bm C_{n,r}\quad \text{\it A.7.3}
\nn\\
&+T\l[
(\bm\Sigma^\xi)^{-1}\bm\Lambda\bm P^{-1}\bm\Lambda^\prime(\bm\Sigma^\xi)^{-1}\wh{\bm\Gamma}^{\xi F}\bm\Lambda^\prime(\bm\Sigma^\xi)^{-1}\bm\Lambda\bm P^{-1}
\otimes
\bm P^{-1} \bm\Lambda^\prime (\bm\Sigma^\xi)^{-1}
\r]\bm C_{n,r}\quad \text{\it A.7.4}
\nn\\
&+T\l[
(\bm\Sigma^\xi)^{-1}\bm\Lambda\bm P^{-1}
\otimes 
\bm P^{-1} \bm\Lambda^\prime(\bm\Sigma^\xi)^{-1}\bm\Lambda\bm\Lambda^\prime(\bm\Sigma^\xi)^{-1}\bm\Lambda\bm P^{-1}\bm\Lambda^\prime(\bm\Sigma^\xi)^{-1}
\r]\bm C_{n,r}\quad \text{\it A.8.1}
\nn\\
&+T\l[
(\bm\Sigma^\xi)^{-1}\bm\Lambda\bm P^{-1}
\otimes 
\bm P^{-1} \bm\Lambda^\prime(\bm\Sigma^\xi)^{-1}\wh{\bm\Gamma}^\xi(\bm\Sigma^\xi)^{-1}\bm\Lambda\bm P^{-1}\bm\Lambda^\prime(\bm\Sigma^\xi)^{-1}
\r]\bm C_{n,r}\quad \text{\it A.8.2}
\nn\\
&+T\l[
(\bm\Sigma^\xi)^{-1}\bm\Lambda\bm P^{-1}
\otimes 
\bm P^{-1} \bm\Lambda^\prime(\bm\Sigma^\xi)^{-1}\bm\Lambda\wh{\bm\Gamma}^{F \xi}(\bm\Sigma^\xi)^{-1}\bm\Lambda\bm P^{-1}\bm\Lambda^\prime(\bm\Sigma^\xi)^{-1}
\r]\bm C_{n,r}\quad \text{\it A.8.3}
\nn\\
&+T\l[
(\bm\Sigma^\xi)^{-1}\bm\Lambda\bm P^{-1}
\otimes 
\bm P^{-1} \bm\Lambda^\prime(\bm\Sigma^\xi)^{-1}\wh{\bm\Gamma}^{\xi F}\bm\Lambda^\prime(\bm\Sigma^\xi)^{-1}\bm\Lambda\bm P^{-1}\bm\Lambda^\prime(\bm\Sigma^\xi)^{-1}
\r]\bm C_{n,r}\quad \text{\it A.8.4}
\nn\\
&-T\l[(\bm\Sigma^\xi)^{-1}
\otimes \bm P^{-1} 
\r]\quad \text{\it B.1}
\nn\\
&+T\l[
(\bm\Sigma^\xi)^{-1}\bm\Lambda\bm\Lambda^\prime(\bm\Sigma^\xi)^{-1}
\otimes 
\bm P^{-1} 
\r]\quad \text{\it B.2.1}
\nn\\
&+T\l[
(\bm\Sigma^\xi)^{-1}\wh{\bm\Gamma}^\xi(\bm\Sigma^\xi)^{-1}
\otimes 
\bm P^{-1} 
\r]\quad \text{\it B.2.2}
\nn\\
&+T\l[
(\bm\Sigma^\xi)^{-1}\bm\Lambda\wh{\bm\Gamma}^{F\xi}(\bm\Sigma^\xi)^{-1}
\otimes 
\bm P^{-1} 
\r]\quad \text{\it B.2.3}
\nn\\
&+T\l[
(\bm\Sigma^\xi)^{-1}\wh{\bm\Gamma}^{\xi F}\bm\Lambda^\prime(\bm\Sigma^\xi)^{-1}
\otimes 
\bm P^{-1} 
\r]\quad \text{\it B.2.4}
\nn\\
&[\ldots]\nn
\end{align}
\begin{align}
&[\ldots]\nn\\
&-T\l[
(\bm\Sigma^\xi)^{-1}\bm\Lambda\bm P^{-1}\bm\Lambda^\prime(\bm\Sigma^\xi)^{-1}\bm\Lambda\bm\Lambda^\prime(\bm\Sigma^\xi)^{-1}
\otimes 
\bm P^{-1} 
\r]\quad \text{\it B.3.1}
\nn\\
&-T\l[
(\bm\Sigma^\xi)^{-1}\bm\Lambda\bm P^{-1}\bm\Lambda^\prime(\bm\Sigma^\xi)^{-1}\wh{\bm\Gamma}^\xi(\bm\Sigma^\xi)^{-1}
\otimes 
\bm P^{-1} 
\r]\quad \text{\it B.3.2}
\nn\\
&-T\l[
(\bm\Sigma^\xi)^{-1}\bm\Lambda\bm P^{-1}\bm\Lambda^\prime(\bm\Sigma^\xi)^{-1}\bm\Lambda\wh{\bm\Gamma}^{F\xi}(\bm\Sigma^\xi)^{-1}
\otimes 
\bm P^{-1} 
\r]\quad \text{\it B.3.3}
\nn\\
&-T\l[
(\bm\Sigma^\xi)^{-1}\bm\Lambda\bm P^{-1}\bm\Lambda^\prime(\bm\Sigma^\xi)^{-1}\wh{\bm\Gamma}^{\xi F}\bm\Lambda^\prime(\bm\Sigma^\xi)^{-1}
\otimes 
\bm P^{-1} 
\r]\quad \text{\it B.3.4}
\nn\\
&-T\l[
(\bm\Sigma^\xi)^{-1}\bm\Lambda\bm P^{-1}
\otimes 
\bm P^{-1} \bm\Lambda^\prime(\bm\Sigma^\xi)^{-1}\bm\Lambda\bm\Lambda^\prime(\bm\Sigma^\xi)^{-1}
\r]\bm C_{n,r}\quad \text{\it B.4.1}
\nn\\
&-T\l[
(\bm\Sigma^\xi)^{-1}\bm\Lambda\bm P^{-1}
\otimes 
\bm P^{-1} \bm\Lambda^\prime(\bm\Sigma^\xi)^{-1}\wh{\bm\Gamma}^\xi(\bm\Sigma^\xi)^{-1}
\r]\bm C_{n,r}\quad \text{\it B.4.2}
\nn\\
&-T\l[
(\bm\Sigma^\xi)^{-1}\bm\Lambda\bm P^{-1}
\otimes 
\bm P^{-1} \bm\Lambda^\prime(\bm\Sigma^\xi)^{-1}\bm\Lambda\wh{\bm\Gamma}^{F\xi}(\bm\Sigma^\xi)^{-1}
\r]\bm C_{n,r}\quad \text{\it B.4.3}
\nn\\
&-T\l[
(\bm\Sigma^\xi)^{-1}\bm\Lambda\bm P^{-1}
\otimes 
\bm P^{-1} \bm\Lambda^\prime(\bm\Sigma^\xi)^{-1}\wh{\bm\Gamma}^{\xi F}\bm\Lambda^\prime(\bm\Sigma^\xi)^{-1}
\r]\bm C_{n,r}\quad \text{\it B.4.4}
\nn\\
&-T\l[
(\bm\Sigma^\xi)^{-1}
\otimes 
\bm P^{-1}  \bm\Lambda^\prime(\bm\Sigma^\xi)^{-1}\bm\Lambda  \bm\Lambda^\prime(\bm\Sigma^\xi)^{-1}\bm\Lambda\bm P^{-1}
\r]\quad \text{\it B.5.1}
\nn\\
&-T\l[
(\bm\Sigma^\xi)^{-1}
\otimes 
\bm P^{-1}  \bm\Lambda^\prime(\bm\Sigma^\xi)^{-1}\wh{\bm\Gamma}^\xi(\bm\Sigma^\xi)^{-1}\bm\Lambda\bm P^{-1}
\r]\quad \text{\it B.5.2}
\nn\\
&-T\l[
(\bm\Sigma^\xi)^{-1}
\otimes 
\bm P^{-1}  \bm\Lambda^\prime(\bm\Sigma^\xi)^{-1}\bm\Lambda\wh{\bm\Gamma}^{F\xi}(\bm\Sigma^\xi)^{-1}\bm\Lambda\bm P^{-1}
\r]\quad \text{\it B.5.3}
\nn\\
&-T\l[
(\bm\Sigma^\xi)^{-1}
\otimes 
\bm P^{-1}  \bm\Lambda^\prime(\bm\Sigma^\xi)^{-1}\wh{\bm\Gamma}^{\xi F}\bm\Lambda^\prime(\bm\Sigma^\xi)^{-1}\bm\Lambda\bm P^{-1}
\r].\quad \text{\it B.5.4}
\label{eq:43gatti}
\end{align}
}

By using \eqref{eq:anc1}-\eqref{eq:anc4} 30 terms of \eqref{eq:43gatti} cancel out  asymptotically, namely:
\ben
\item $\Vert \text{\it A.1}-\text{\it A.2.2}\Vert = O_{\mathrm P}(\sqrt T n^{-1})+ O_{\mathrm P}(T n^{-2})$; 
\item $\Vert \text{\it A.2.1}-\text{\it A.3.1}\Vert = O_{\mathrm P}(T n^{-1})$;
\item $\Vert \text{\it A.2.3}-\text{\it A.3.3}\Vert = O_{\mathrm P}(T n^{-1})$;
\item $\Vert \text{\it A.5}-\text{\it A.6.2}\Vert = O_{\mathrm P}(\sqrt T n^{-1})+ O_{\mathrm P}(T n^{-2})$;
\item $\Vert \text{\it A.6.1}-\text{\it A.7.1}\Vert = O_{\mathrm P}(T n^{-1})$;
\item $\Vert \text{\it A.6.3}-\text{\it A.7.3}\Vert = O_{\mathrm P}(T n^{-1})$;
\item $\Vert \text{\it A.2.4}-\text{\it B.2.4}\Vert = O_{\mathrm P}(T n^{-1})$;
\item $\Vert \text{\it B.2.1}-\text{\it B.3.1}\Vert = O_{\mathrm P}(T n^{-1})$;
\item $\Vert \text{\it A.3.2}-\text{\it B.3.2}\Vert = O_{\mathrm P}(\sqrt T n^{-1})+ O_{\mathrm P}(T n^{-2})$; 
\item $\Vert \text{\it B.2.3}-\text{\it B.3.3}\Vert = O_{\mathrm P}(T n^{-1})$;
\item $\Vert \text{\it A.3.4}-\text{\it B.3.4}\Vert = O_{\mathrm P}(T n^{-1})$;
\item $\Vert \text{\it A.8.1}-\text{\it B.4.1}\Vert = O_{\mathrm P}(T n^{-1})$;
\item $\Vert \text{\it A.8.2}-\text{\it B.4.2}\Vert = O_{\mathrm P}(T n^{-1})+O_{\mathrm P}(\sqrt T n^{-1})+ O_{\mathrm P}(T n^{-2})$;
\item $\Vert \text{\it A.8.4}-\text{\it B.4.4}\Vert = O_{\mathrm P}(T n^{-1})$;
\item $\Vert \text{\it B.2.2}-\text{\it B.5.2}\Vert = O_{\mathrm P}(T n^{-1})+O_{\mathrm P}(\sqrt T n^{-1})+ O_{\mathrm P}(T n^{-2})$.
\een
and by using again \eqref{eq:anc2} we are left with the following 13 terms (ordered differently than in the previous expression):
\begin{align}
\bm H(\bm{\mathcal X};\bm\varphi)=&\, 
-T \l[
(\bm\Sigma^{\xi})^{-1}
\otimes 
\mbf I_r
\r]\quad \text{\it B.5.1}\nn\\
&+T\l[
(\bm\Sigma^{\xi})^{-1}\bm\Lambda
\l\{\bm\Lambda^\prime(\bm\Sigma^{\xi})^{-1}\bm\Lambda \r\}^{-1} 
\bm\Lambda^\prime(\bm\Sigma^{\xi})^{-1} 
\otimes 
\mbf I_r
 \r]\quad \text{\it A.4.1}\nn\\
&+T\l[
(\bm\Sigma^{\xi})^{-1}\bm\Lambda
\l\{\bm\Lambda^\prime(\bm\Sigma^{\xi})^{-1}\bm\Lambda \r\}^{-1} 
\bm\Lambda^\prime(\bm\Sigma^{\xi})^{-1} 
\otimes 
\l\{\bm\Lambda^\prime(\bm\Sigma^{\xi})^{-1}\bm\Lambda \r\}^{-1} 
\r]\quad \text{\it A.4.2}\nn\\
&+T\l[
(\bm\Sigma^{\xi})^{-1}\bm\Lambda
\l\{\bm\Lambda^\prime(\bm\Sigma^{\xi})^{-1}\bm\Lambda \r\}^{-1} 
\bm\Lambda^\prime(\bm\Sigma^{\xi})^{-1} 
\otimes 
\wh{\bm\Gamma}^{F\xi} (\bm\Sigma^{\xi})^{-1}\bm\Lambda\l\{\bm\Lambda^\prime(\bm\Sigma^{\xi})^{-1}\bm\Lambda \r\}^{-1} 
\r]\quad \text{\it A.4.3}\nn\\
&+T\l[
(\bm\Sigma^{\xi})^{-1}\bm\Lambda
\l\{\bm\Lambda^\prime(\bm\Sigma^{\xi})^{-1}\bm\Lambda \r\}^{-1} 
\bm\Lambda^\prime(\bm\Sigma^{\xi})^{-1} 
\otimes 
\l\{\bm\Lambda^\prime(\bm\Sigma^{\xi})^{-1}\bm\Lambda \r\}^{-1} 
\bm\Lambda^\prime(\bm\Sigma^{\xi})^{-1} \wh{\bm\Gamma}^{\xi F}
\r]\quad \text{\it A.4.4}\nn\\
&-T\l[
(\bm\Sigma^{\xi})^{-1}
\otimes
\l\{\bm\Lambda^\prime(\bm\Sigma^{\xi})^{-1}\bm\Lambda \r\}^{-1} 
\r]\quad \text{\it B.1}\nn\\
&-T\l[
(\bm\Sigma^{\xi})^{-1}
\otimes
\wh{\bm\Gamma}^{F\xi}(\bm\Sigma^{\xi})^{-1}\bm\Lambda\l\{\bm\Lambda^\prime(\bm\Sigma^{\xi})^{-1}\bm\Lambda \r\}^{-1} 
\r]\quad \text{\it B.5.2}\nn\\
&-T\l[
(\bm\Sigma^{\xi})^{-1}
\otimes
\l\{\bm\Lambda^\prime(\bm\Sigma^{\xi})^{-1}\bm\Lambda \r\}^{-1} \bm\Lambda^\prime(\bm\Sigma^{\xi})^{-1}\wh{\bm\Gamma}^{\xi F}
\r]\quad \text{\it B.5.3}\nn\\
&-T\l[
(\bm\Sigma^{\xi})^{-1}\bm\Lambda\l\{\bm\Lambda^\prime(\bm\Sigma^{\xi})^{-1}\bm\Lambda \r\}^{-1}
\otimes
\wh{\bm\Gamma}^{F\xi}(\bm\Sigma^{\xi})^{-1}
\r]\bm C_{n,r}\quad \text{\it B.4.3}\nn\\
&-T\l[
(\bm\Sigma^{\xi})^{-1}\wh{\bm\Gamma}^{\xi F}
\otimes
\l\{\bm\Lambda^\prime(\bm\Sigma^{\xi})^{-1}\bm\Lambda \r\}^{-1}\bm\Lambda^\prime (\bm\Sigma^{\xi})^{-1}
\r]\bm C_{n,r}\quad \text{\it A.6.4}\nn\\
&+T\l[
(\bm\Sigma^{\xi})^{-1}\bm\Lambda\l\{\bm\Lambda^\prime(\bm\Sigma^{\xi})^{-1}\bm\Lambda \r\}^{-1}
\otimes
\l\{\bm\Lambda^\prime(\bm\Sigma^{\xi})^{-1}\bm\Lambda \r\}^{-1} \bm\Lambda^\prime(\bm\Sigma^{\xi})^{-1}
\r]\bm C_{n,r}\quad \text{\it A.7.1}\nn\\
&+T\l[
(\bm\Sigma^{\xi})^{-1}\bm\Lambda\l\{\bm\Lambda^\prime(\bm\Sigma^{\xi})^{-1}\bm\Lambda \r\}^{-1} \bm\Lambda^\prime(\bm\Sigma^{\xi})^{-1}\wh{\bm\Gamma}^{\xi F}
\otimes
\l\{\bm\Lambda^\prime(\bm\Sigma^{\xi})^{-1}\bm\Lambda \r\}^{-1} \bm\Lambda^\prime(\bm\Sigma^{\xi})^{-1}
\r]\bm C_{n,r}\quad \text{\it A.7.2}\nn\\
&+T\l[
(\bm\Sigma^{\xi})^{-1}\bm\Lambda\l\{\bm\Lambda^\prime(\bm\Sigma^{\xi})^{-1}\bm\Lambda \r\}^{-1}
\otimes 
\wh{\bm\Gamma}^{F\xi}(\bm\Sigma^{\xi})^{-1}\bm\Lambda\l\{\bm\Lambda^\prime(\bm\Sigma^{\xi})^{-1}\bm\Lambda \r\}^{-1} \bm\Lambda^\prime(\bm\Sigma^{\xi})^{-1}
\r]\bm C_{n,r}\quad \text{\it A.8.3}\nn\\
&+O\l(\frac {T}n\r)+O_{\mathrm P}\l(\frac {\sqrt T}n\r)+O\l(\frac {T}{n^2}\r).\label{eq:miticuzzo}
\end{align}
Therefore, 
by using again arguments as those in \eqref{eq:anc6}-\eqref{eq:anc9} in \eqref{eq:finquituttobene} we have that the $i$th $r\times r$ sub-matrix of $\bm H(\bm{\mathcal X};\bm{\varphi})$ is such that
\begin{align}
\bm h_{ii}&(\bm{\mathcal X};\bm\varphi)= -\frac T{\sigma_i^2} \mbf I_r\nn\\
& + \frac T{\sigma_i^4}{\bm\lambda}_i^\prime\l\{\bm\Lambda^\prime(\bm\Sigma^{\xi})^{-1}\bm\Lambda \r\}^{-1}{\bm\lambda}_i+\nn\\ 
& + \frac T{\sigma_i^4}{\bm\lambda}_i^\prime\l\{\bm\Lambda^\prime(\bm\Sigma^{\xi})^{-1}\bm\Lambda \r\}^{-1}{\bm\lambda}_i
\l\{
\l\{\bm\Lambda^\prime(\bm\Sigma^{\xi})^{-1}\bm\Lambda \r\}^{-1}+
\frac 1T\sum_{t=1}^T\mbf F_t{\bm\xi}_t^\prime(\bm\Sigma^{\xi})^{-1}\bm\Lambda\l\{\bm\Lambda^\prime(\bm\Sigma^{\xi})^{-1}\bm\Lambda \r\}^{-1}\r.\nn\\
&\;\;+\l.\l\{\bm\Lambda^\prime(\bm\Sigma^{\xi})^{-1}\bm\Lambda \r\}^{-1}\bm\Lambda^\prime(\bm\Sigma^{\xi})^{-1}\frac 1T\sum_{t=1}^T{\bm\xi}_t\mbf F_t^\prime
\r\}\nn\\
&- \frac T{\sigma_i^2}
\l\{
\l\{\bm\Lambda^\prime(\bm\Sigma^{\xi})^{-1}\bm\Lambda \r\}^{-1}+
\frac 1T\sum_{t=1}^T\mbf F_t{\bm\xi}_t^\prime(\bm\Sigma^{\xi})^{-1}\bm\Lambda\l\{\bm\Lambda^\prime(\bm\Sigma^{\xi})^{-1}\bm\Lambda \r\}^{-1}+
\l\{\bm\Lambda^\prime(\bm\Sigma^{\xi})^{-1}\bm\Lambda \r\}^{-1}\bm\Lambda^\prime(\bm\Sigma^{\xi})^{-1}\frac 1T\sum_{t=1}^T{\bm\xi}_t\mbf F_t^\prime
\r\}\nn\\
&- \frac T{\sigma_i^4}\l\{
\frac 1T \sum_{t=1}^T\mbf F_t{\xi}_{it}
\otimes 
{\bm\lambda}_i^\prime\l\{\bm\Lambda^\prime(\bm\Sigma^{\xi})^{-1}\bm\Lambda \r\}^{-1}
\r\}
- \frac T{\sigma_i^4}\l\{
\l\{\bm\Lambda^\prime(\bm\Sigma^{\xi})^{-1}\bm\Lambda \r\}^{-1}{\bm\lambda}_i
\otimes 
\frac 1T \sum_{t=1}^T{\xi}_{it}\mbf F_t^\prime
\r\}\nn\\
&+ \frac T{\sigma_i^4}
\Bigg\{
\l\{\bm\Lambda^\prime(\bm\Sigma^{\xi})^{-1}\bm\Lambda \r\}^{-1} {\bm\lambda}_i
\otimes 
{\bm\lambda}_i^\prime\l\{\bm\Lambda^\prime(\bm\Sigma^{\xi})^{-1}\bm\Lambda \r\}^{-1}
\Bigg\}\nn\\
&+ \frac T{\sigma_i^4}\Bigg\{
\l\{\bm\Lambda^\prime(\bm\Sigma^{\xi})^{-1}\bm\Lambda \r\}^{-1}{\bm\lambda}_i
\otimes
{\bm\lambda}_i^\prime\l\{\bm\Lambda^\prime(\bm\Sigma^{\xi})^{-1}\bm\Lambda \r\}^{-1}\bm\Lambda^\prime(\bm\Sigma^\xi)^{-1}\frac 1T\sum_{t=1}^T{\bm\xi}_t\mbf F_t^\prime
\Bigg\}\nn\\
&+ \frac T{\sigma_i^4}\Bigg\{
\frac 1T\sum_{t=1}^T\mbf F_t{\bm\xi}_t^\prime(\bm\Sigma^\xi)^{-1}\bm\Lambda\l\{\bm\Lambda^\prime(\bm\Sigma^{\xi})^{-1}\bm\Lambda \r\}^{-1}{\bm\lambda}_i
\otimes 
{\bm\lambda}_i^\prime\l\{\bm\Lambda^\prime(\bm\Sigma^{\xi})^{-1}\bm\Lambda \r\}^{-1}
\Bigg\}\nn\\
&+O\l(\frac{T}{n^2}\r)+O_{\mathrm P}\l(\frac{\sqrt T}{n^2}\r)+O\l(\frac{T}{n^4}\r).\label{eq:HXii}
\end{align}
Hence, $\Vert \bm h_{ii}(\bm{\mathcal X};\bm\varphi)\Vert=O_{\mathrm P}(T)$, because of Assumption \ref{ass:common}(a), \eqref{eq:anc5}, and \eqref{eq:2a3bis}. Moreover, from \eqref{eq:HXii} and by noticing that $(\bm\Sigma^\xi)^{-1}{\bm\Lambda}$ has the same properties as ${\bm\Lambda}$ since $\Vert (\bm\Sigma^\xi)^{-1}\Vert = O(1)$ by Assumption \ref{ass:idio}(a), we get
\begin{align}
\frac 1T&\l\Vert\bm h_{ii}(\bm{\mathcal X};\bm\varphi)- \l(-\frac T{\sigma_i^2} \mbf I_r\r)\r\Vert\le \frac{M_\Lambda^2}{C_\xi^2}
\l\Vert \l\{\bm\Lambda^\prime(\bm\Sigma^{\xi})^{-1}\bm\Lambda \r\}^{-1}\r\Vert\nn\\
& +
\frac{M_\Lambda^2}{C_\xi^2}
\l\Vert \l\{\bm\Lambda^\prime(\bm\Sigma^{\xi})^{-1}\bm\Lambda \r\}^{-1}\r\Vert^2
\Bigg\{
1 + 2\l\Vert\frac 1T\bm\Lambda^\prime(\bm\Sigma^\xi)^{-1}\sum_{t=1}^T{\bm\xi}_t\mbf F_t^\prime  \r\Vert
\Bigg\}\nn\\
&
\frac 1{C_\xi}
\l\Vert \l\{\bm\Lambda^\prime(\bm\Sigma^{\xi})^{-1}\bm\Lambda \r\}^{-1}\r\Vert
\Bigg\{
1 + 2\l\Vert\frac 1T\bm\Lambda^\prime(\bm\Sigma^\xi)^{-1}\sum_{t=1}^T{\bm\xi}_t\mbf F_t^\prime  \r\Vert
\Bigg\}\nn\\
&+ \frac{2M_\Lambda}{C_\xi^2} \l\Vert \l\{\bm\Lambda^\prime(\bm\Sigma^{\xi})^{-1}\bm\Lambda \r\}^{-1}\r\Vert\,\l\Vert\frac 1T\sum_{t=1}^T\mbf F_t{\xi}_{it}\r\Vert+
\frac{M_\Lambda^2}{C_\xi^2}
\l\Vert \l\{\bm\Lambda^\prime(\bm\Sigma^{\xi})^{-1}\bm\Lambda \r\}^{-1}\r\Vert^2\Bigg\{
1 + 2\l\Vert\frac 1T\bm\Lambda^\prime(\bm\Sigma^\xi)^{-1}\sum_{t=1}^T{\bm\xi}_t\mbf F_t^\prime  \r\Vert
\Bigg\}\nn\\
&+O\l(\frac{1}{n^2}\r)+O_{\mathrm P}\l(\frac{1}{n^2\sqrt T}\r)+O\l(\frac{1}{n^4}\r)\nn\\
=&\,O\l(\frac{1}{n}\r) +O\l(\frac{1}{n^2}\r) \Bigg\{1+O_{\mathrm P}\l(\frac{\sqrt n}{\sqrt T}\r)\Bigg\}+O\l(\frac{1}{n}\r) \Bigg\{1+O_{\mathrm P}\l(\frac{\sqrt n}{\sqrt T}\r)\Bigg\}+O\l(\frac{1}{n}\r)O_{\mathrm P}\l(\frac{1}{\sqrt T}\r)\nn\\
&+O\l(\frac{1}{n^2}\r) \Bigg\{1+O_{\mathrm P}\l(\frac{\sqrt n}{\sqrt T}\r)\Bigg\}+O\l(\frac{1}{n^2}\r)+O_{\mathrm P}\l(\frac{1}{n^2\sqrt T}\r)+O\l(\frac{1}{n^4}\r)= O_{\mathrm P}\l(\max\l(\frac{1}{n},\frac 1{\sqrt{nT}}\r)\r).\label{eq:GOL2} 
\end{align}

Finally, from \eqref{eq:SXF} we have
\begin{align}
\mathrm d \bm S^\prime(\bm{\mathcal X}|\bm{\mathcal F};\underline{\bm\varphi})&=\mathrm d\l(\sum_{t=1}^T\mbf F_t(\mbf x_t-\underline{\bm\Lambda}\mbf F_t)^\prime(\underline{\bm\Sigma}^\xi)^{-1} \r) = -\sum_{t=1}^T\mbf F_t\mbf F_t^\prime\l(\mathrm d\underline{\bm\Lambda}^\prime\r)(\underline{\bm\Sigma}^\xi)^{-1},\nn
\end{align}
so 
\beq
\text{vec}(\mathrm d \bm S^\prime(\bm{\mathcal X}|\bm{\mathcal F};\underline{\bm\varphi})) = - (\underline{\bm\Sigma}^\xi)^{-1}\otimes \l(\sum_{t=1}^T\mbf F_t\mbf F_t^\prime\r)\text{vec}(\mathrm d \underline{\bm \Lambda}^\prime),\nn
\eeq
and
\beq
\bm H(\bm{\mathcal X}|\bm{\mathcal F};\underline{\bm\varphi})  = - (\underline{\bm\Sigma}^\xi)^{-1}\otimes \l(\sum_{t=1}^T\mbf F_t\mbf F_t^\prime\r)=-T (\underline{\bm\Sigma}^\xi)^{-1}\otimes \mbf I_r,\label{eq:HXF}
\eeq
because of Assumption \ref{ass:ident}(b). And \eqref{eq:HXF} computed in the true value of the parameters is
\beq
\bm H(\bm{\mathcal X}|\bm{\mathcal F};{\bm\varphi})  = -T ({\bm\Sigma}^\xi)^{-1}\otimes \mbf I_r.\label{eq:ottimo}
\eeq
From \eqref{eq:ottimo} the $i$th $r\times r$ sub-matrix of $\bm H(\bm{\mathcal X}|\bm{\mathcal F};\bm{\varphi})$ is then:
\beq
\bm h_{ii}(\bm{\mathcal X}|\bm{\mathcal F};{\bm\varphi})  = -\frac T{\sigma_i^2}\mbf I_r.\label{eq:HXFi}
\eeq
Hence, $\Vert\bm h_{ii}(\bm{\mathcal X}|\bm{\mathcal F};{\bm\varphi}) \Vert=O_{\mathrm P}(T)$, and, by using \eqref{eq:HXFi} into \eqref{eq:GOL2}, for any given $i=1,\ldots, n$, we have:
\beq\nn
\frac 1T\l\Vert\bm h_{ii}(\bm{\mathcal X};\bm\varphi)- \bm h_{ii}(\bm{\mathcal X}|\bm{\mathcal F};{\bm\varphi})\r\Vert=O_{\mathrm P}\l(\max\l(\frac{1}{n},\frac 1{\sqrt{nT}}\r)\r),
\eeq
which proves part (b).\smallskip

For part (c) we provide two different equivalent proofs.  First, 
\begin{align}
\l\Vert
\bm s_{i}(\bm{\mathcal X};\wh{\bm{\varphi}}^{\text{\tiny QML,E}})-\bm s_{i}(\bm{\mathcal X}|\bm{\mathcal F};\wh{\bm{\varphi}}^{\text{\tiny QML,E}})
\r\Vert
\le&\, \l\Vert
\bm s_{i}(\bm{\mathcal X};{\bm{\varphi}})-\bm s_{i}(\bm{\mathcal X}|\bm{\mathcal F};{\bm{\varphi}})
\r\Vert\nn\\
&+\l\Vert
\bm h_{ii}(\bm{\mathcal X};{\bm{\varphi}})-\bm h_{ii}(\bm{\mathcal X}|\bm{\mathcal F};{\bm{\varphi}})
\r\Vert\, \l\Vert \wh{\bm\lambda}_i^{\text{\tiny QML,E}}-{\bm\lambda}_i \r\Vert+O_{\mathrm P}\l(\l\Vert \wh{\bm\lambda}_i^{\text{\tiny QML,E}}-{\bm\lambda}_i \r\Vert^2\r)\nn\\
=&\,O_{\mathrm P}\l(\max\l(\frac {T}n,\frac {\sqrt T}{\sqrt{n}}\r)\r)+O_{\mathrm P}\l(\max\l(\frac {T}n,\frac {\sqrt T}{\sqrt{n}}\r)\r)O_{\mathrm P}\l(\max\l(\frac {1}n,\frac {1}{\sqrt{T}}\r)\r).\nn
\end{align}
because of parts (a) and (b), and \eqref{eq:QMLPCOLS}, which follows from Theorem \ref{th:QML}, which, in turn, holds even if we use the simpler log-likelihood \eqref{eq:LL00} in place of the log-likelihood \eqref{eq:LL0}. 
As a consequence,
\[
\frac 1{\sqrt T}\l\Vert
\bm s_{i}(\bm{\mathcal X};\wh{\bm{\varphi}}^{\text{\tiny QML,E}})-\bm s_{i}(\bm{\mathcal X}|\bm{\mathcal F};\wh{\bm{\varphi}}^{\text{\tiny QML,E}})
\r\Vert= O_{\mathrm P}\l(\max\l(\frac {\sqrt T}n,\frac 1{\sqrt{n}}\r)\r).
\]
Alternatively, let $\wh{\bm\varphi}^{\text{\tiny OLS}}=(\mathrm{vec}(\bm\Lambda^{\text{\tiny OLS}})^\prime, \mathrm{vech}(\underline{\bm\Gamma}^\xi)^\prime)^\prime$ which is the vector of parameters made of the OLS estimator of the loadings and any generic value of the idiosyncratic covariance matrix satisfying Assumption \ref{ass:idio}. Then, notice that $\bm s_{i}(\bm{\mathcal X};\wh{\bm{\varphi}}^{\text{\tiny QML,E}})=\mbf 0_r$ by definition of QML estimator, and $\bm s_{i}(\bm{\mathcal X}|\bm{\mathcal F};\wh{\bm{\varphi}}^{\text{\tiny OLS}})=\mbf 0_r$ because the OLS estimator is the QML estimator when maximizing the conditional log-likelihood, and the OLS  estimator of the loadings does not depend on the estimator of the idiosyncratic covariance. Recall also that
\[
\l\Vert \wh{\bm\lambda}_i^{\text{\tiny QML,E}}-{\bm\lambda}_i^{\text{\tiny OLS}} \r\Vert\le \l\Vert \wh{\bm\lambda}_i^{\text{\tiny QML,E}}-\wh{\bm\lambda}_i \r\Vert+\l\Vert \wh{\bm\lambda}_i-{\bm\lambda}_i^{\text{\tiny OLS}} \r\Vert = O_{\mathrm P}\l(\max\l(\frac 1n,\frac 1{\sqrt{nT}}\r)\r),
\]
because of Theorem \ref{th:QML} and \eqref{eq:QMLPCOLS} which hold even if we use the log-likelihood \eqref{eq:LL00} in place of the log-likelihood \eqref{eq:LL0}. Moreover, from \eqref{eq:HX} it is straightforward to see that 
$\Vert\bm h_{ii}(\bm{\mathcal X}|\bm{\mathcal F};\underline{\bm{\varphi}})\Vert=O_{\mathrm P}(T)$, for any $\underline{\bm{\varphi}}\in\mathcal O_n$. Therefore,
\begin{align}
\l\Vert
\bm s_{i}(\bm{\mathcal X};\wh{\bm{\varphi}}^{\text{\tiny QML,E}})-\bm s_{i}(\bm{\mathcal X}|\bm{\mathcal F};\wh{\bm{\varphi}}^{\text{\tiny QML,E}})
\r\Vert
=&\, \l\Vert
\bm s_{i}(\bm{\mathcal X}|\bm{\mathcal F};\wh{\bm{\varphi}}^{\text{\tiny QML,E}})
\r\Vert\nn\\
\le&\, \l\Vert
\bm s_{i}(\bm{\mathcal X}|\bm{\mathcal F};\wh{\bm{\varphi}}^{\text{\tiny OLS}})
\r\Vert+ \l\Vert
\bm h_{ii}(\bm{\mathcal X}|\bm{\mathcal F};\wh{\bm{\varphi}}^{\text{\tiny OLS}})
\r\Vert\,
\l\Vert \wh{\bm\lambda}_i^{\text{\tiny QML,E}}-{\bm\lambda}_i^{\text{\tiny OLS}} \r\Vert\nn\\
&+ O_{\mathrm P}\l(\l\Vert \wh{\bm\lambda}_i^{\text{\tiny QML,E}}-{\bm\lambda}_i^{\text{\tiny OLS}} \r\Vert^2\r)\nn\\
\le &\, O_{\mathrm P}\l(T\r) O_{\mathrm P}\l(\max\l(\frac 1n,\frac 1{\sqrt{nT}}\r)\r)=O_{\mathrm P}\l(\max\l(\frac {T}n,\frac {\sqrt T}{\sqrt{n}}\r)\r).\nn
\end{align}
This proves part (c).\smallskip

For part (d), for any specific value of the parameters, say $\wt{\bm\varphi}$, define the $n^2r^2\times nr$ matrices of third derivatives
\[
\bm T(\bm{\mathcal X};\wt{\bm\varphi}) =\l.\frac{\partial \text{vec}(\bm H(\bm{\mathcal X};\underline{\bm\varphi}))^\prime}{\partial \text{vec}(\underline{\bm\Lambda})^\prime}
\r|_{\underline{\bm\varphi}=\wt{\bm\varphi}},\quad \bm T(\bm{\mathcal X}|\bm{\mathcal F};\wt{\bm\varphi}) =\l.\frac{\partial \text{vec}(\bm H(\bm{\mathcal X}|\bm{\mathcal F};\underline{\bm\varphi}))^\prime}{\partial \text{vec}(\underline{\bm\Lambda})^\prime}
\r|_{\underline{\bm\varphi}=\wt{\bm\varphi}},
\]
with components
\[
\bm t_{iii}(\bm{\mathcal X};\wt{\bm\varphi}) = \l.\frac{\partial \text{vec}(\bm h_{ii}(\bm{\mathcal X};\underline{\bm\varphi}))^\prime}{\partial \underline{\bm\lambda}_i^\prime}
\r|_{\underline{\bm\varphi}=\wt{\bm\varphi}},\quad \bm t_{iii}(\bm{\mathcal X}|\bm{\mathcal F};\wt{\bm\varphi}) = \l.\frac{\partial \text{vec}(\bm h_{ii}(\bm{\mathcal X}|\bm{\mathcal F};\underline{\bm\varphi}))^\prime}{\partial \underline{\bm\lambda}_i^\prime}
\r|_{\underline{\bm\varphi}=\wt{\bm\varphi}},\quad i=1,\ldots, n,
\]
which are $r^2\times r$ matrices. Then, from \eqref{eq:HX} the $i$th $r\times r$ sub-matrix of $\bm H(\bm{\mathcal X};\underline{\bm\varphi})$ is given by:
\begin{align}\label{eq:HXigen}
\bm h_{ii}(\bm{\mathcal X};\underline{\bm\varphi})=&\,
\frac T{\underline{\sigma}_i^4} 
\l[
\l\{\underline{\bm\lambda}_i^\prime \underline{\bm P}^{-1}\underline{\bm\lambda}_i\r\}
\underline{\bm P}^{-1}
\r]\nn\\
&-\frac T{\underline{\sigma}_i^4} 
\l[
\l\{[\wh{\bm\Gamma}^x]_{i\cdot}(\underline{\bm\Sigma}^\xi)^{-1}\underline{\bm\Lambda}\underline{\bm P}^{-1}\underline{\bm\lambda}_i
\r\}
\underline{\bm P}^{-1}
\r]\nn\\
&+\frac T{\underline{\sigma}_i^4} 
\l[
\l\{\underline{\bm \lambda}_i\underline{\bm P}^{-1} \underline{\bm \Lambda}^\prime (\underline{\bm\Sigma}^\xi)^{-1}\wh{\bm\Gamma}^x (\underline{\bm\Sigma}^\xi)^{-1}\underline{\bm \Lambda}\,\underline{\bm P}^{-1} \underline{\bm \lambda}_i
\r\} 
\underline{\bm P}^{-1}
\r]\nn\\
&+\frac T{\underline{\sigma}_i^4} 
\l[
\l\{\underline{\bm\lambda}_i^\prime \underline{\bm P}^{-1}\underline{\bm\lambda}_i
\r\}
\underline{\bm P}^{-1}\underline{\bm\Lambda}^\prime(\underline{\bm\Sigma}^\xi)^{-1}\wh{\bm\Gamma}^\xi(\underline{\bm\Sigma}^\xi)^{-1}
\underline{\bm\Lambda}\,\underline{\bm P}^{-1}
\r]\nn\\
&+\frac T{\underline{\sigma}_i^4} 
\l[
\underline{\bm P}^{-1}\underline{\bm\lambda}_i
\otimes
\underline{\bm\lambda}_i^\prime \underline{\bm P}^{-1} 
\r]\nn\\
&-\frac T{\underline{\sigma}_i^4} 
\l[
\underline{\bm P}^{-1}\underline{\bm\lambda}_i
\otimes
[\wh{\bm\Gamma}^x]_{i\cdot}(\underline{\bm\Sigma}^\xi)^{-1}\underline{\bm\Lambda}\,\underline{\bm P}^{-1}
\r]\nn\\
&+\frac T{\underline{\sigma}_i^4} 
\l[
\underline{\bm P}^{-1}\underline{\bm\lambda}_i
\otimes
\underline{\bm\lambda}_i^\prime \underline{\bm P}^{-1}\underline{\bm\Lambda}^\prime(\underline{\bm\Sigma}^\xi)^{-1}\wh{\bm\Gamma}^\xi(\underline{\bm\Sigma}^\xi)^{-1}\underline{\bm\Lambda}\,\underline{\bm P}^{-1}
\r]\nn\\
&+\frac T{\underline{\sigma}_i^4} 
\l[
\underline{\bm P}^{-1}\underline{\bm\Lambda}^\prime(\underline{\bm\Sigma}^\xi)^{-1}\wh{\bm\Gamma}^\xi(\underline{\bm\Sigma}^\xi)^{-1}\underline{\bm\Lambda}\underline{\bm P}^{-1}\underline{\bm\lambda}_i
\otimes
\underline{\bm\lambda}_i^\prime \underline{\bm P}^{-1}
\r]\nn\\
&-\frac T{\underline{\sigma}_i^2}
\l[
\underline{\bm P}^{-1}
\r] \nn\\
&+\frac T{\underline{\sigma}_i^4} 
\l[
[\wh{\bm\Gamma}^x]_{ii}
\underline{\bm P}^{-1}
\r]\nn\\
&-\frac T{\underline{\sigma}_i^4} 
\l[
\l\{\underline{\bm\lambda}_i^\prime\underline{\bm P}^{-1}\underline{\bm\Lambda}^\prime(\underline{\bm\Sigma}^\xi)^{-1}[\wh{\bm\Gamma}^x]_{\cdot i}
\r\}
\underline{\bm P}^{-1} 
\r]\nn\\
&-\frac T{\underline{\sigma}_i^4} 
\l[
\underline{\bm P}^{-1}\underline{\bm\Lambda}^\prime(\underline{\bm\Sigma}^\xi)^{-1}[\wh{\bm\Gamma}^x]_{\cdot i}
\otimes
\underline{\bm\lambda}_i^\prime\underline{\bm P}^{-1}
\r]\nn\\
&-\frac T{\underline{\sigma}_i^2}
\l[
\underline{\bm P}^{-1}\underline{\bm\Lambda}^\prime(\underline{\bm\Sigma}^\xi)^{-1}\wh{\bm\Gamma}^x(\underline{\bm\Sigma}^\xi)^{-1}\underline{\bm\Lambda}\underline{\bm P}^{-1}
\r],
\end{align}
where $[\wh{\bm\Gamma}^x]_{i\cdot}$ and $[\wh{\bm\Gamma}^x]_{\cdot i}$ are the $i$th row and column of $\wh{\bm\Gamma}^x$, respectively, and $[\wh{\bm\Gamma}^x]_{ii}$ is the $i$th term on its diagonal. 
Since, 
\begin{align}
&\l\Vert \underline{\bm\lambda}_i\r\Vert \le M_\Lambda,\quad \l\Vert \underline{\bm\Lambda}\r\Vert = O(\sqrt n),\nn\\
&\l\Vert (\underline{\bm\Sigma}^\xi)^{-1}\r\Vert \le \frac 1{C_\xi},\quad \l\vert \frac 1{\sigma_i^2}\r\vert\le \frac 1{C_\xi},\nn\\
&\l\Vert \underline{\bm P}^{-1}\r\Vert =\l\Vert \l\{\mbf I_r+\underline{\bm \Lambda}^\prime (\underline{\bm\Sigma}^\xi)^{-1}\underline{\bm \Lambda} \r\}^{-1}\r\Vert = O\l(\frac 1n\r),\nn\\
&\l\Vert \wh{\bm\Gamma}^x\r\Vert \le \l\Vert \wh{\bm\Gamma}^x-\bm\Gamma^x\r\Vert+\l\Vert {\bm\Gamma}^x\r\Vert = O_{\mathrm P}\l(\frac n{\sqrt T}\r) +O(n),\nn\\
& \l\Vert [\wh{\bm\Gamma}^x]_{i\cdot }\r\Vert = \l\Vert [\wh{\bm\Gamma}^x]_{i\cdot }-[{\bm\Gamma}^x]_{i\cdot }\r\Vert+\l\Vert [{\bm\Gamma}^x]_{i\cdot }\r\Vert= O_{\mathrm P}\l(\frac{\sqrt n}{\sqrt T}\r) +O(\sqrt n),\label{eq:varieOp}
\end{align}
because of Assumption \ref{ass:common}(a), Lemma \ref{lem:FTLN}(i), Assumption \ref{ass:idio}(a), Lemma \ref{lem:covarianze}(i), and Lemma \ref{lem:Gxi}(vi), respectively

From \eqref{eq:HXigen} and \eqref{eq:varieOp} it is clear that the leading term in $\bm h_{ii}(\bm{\mathcal X};\underline{\bm\varphi})$ is the last one, which is $O_{\mathrm P}(T)$ while all others are $o_{\mathrm P}(T)$, specifically,
\begin{align}\label{eq:approxHXigen}
\frac 1T \bm h_{ii}(\bm{\mathcal X};\underline{\bm\varphi})&= -\frac 1{\underline{\sigma}_i^2}
\l[
\underline{\bm P}^{-1}\underline{\bm\Lambda}^\prime(\underline{\bm\Sigma}^\xi)^{-1}\wh{\bm\Gamma}^x(\underline{\bm\Sigma}^\xi)^{-1}\underline{\bm\Lambda}\underline{\bm P}^{-1}
\r]+O_{\mathrm P}\l(\frac 1n\r)+O\l(\frac 1{n^2}\r)\nn\\
&=\frac 1T \bm h^*_{ii}(\bm{\mathcal X};\underline{\bm\varphi})+O_{\mathrm P}\l(\frac 1n\r)+O\l(\frac 1{n^2}\r), \;\text{say.}
\end{align}
From \eqref{eq:approxHXigen}  
\begin{align}
\text{vec}\l(\mathrm d \bm h_{ii}^{*\prime} (\bm{\mathcal X};\underline{\bm\varphi})\r) =&\,
-\frac T{\underline{\sigma}_i^2}\l\{
\underline{\bm P}^{-1}\underline{\bm\Lambda}^\prime(\underline{\bm\Sigma}^\xi)^{-1}\wh{\bm\Gamma}^x(\underline{\bm\Sigma}^\xi)^{-1}\underline{\bm\Lambda}
\otimes 
\mbf I_r
\r\}\text{vec}\l(\mathrm d\underline{\bm P}^{-1}\r)\nn\\
&-\frac T{\underline{\sigma}_i^2}\l\{
\underline{\bm P}^{-1}\underline{\bm\Lambda}^\prime(\underline{\bm\Sigma}^\xi)^{-1}\wh{\bm\Gamma}^x(\underline{\bm\Sigma}^\xi)^{-1}
\otimes
\underline{\bm P}^{-1}
\r\} \text{vec}\l(\mathrm d\underline{\bm\Lambda}^\prime\r)\nn\\
&-\frac T{\underline{\sigma}_i^2}\l\{
\underline{\bm P}^{-1}
\otimes
\underline{\bm P}^{-1}\underline{\bm\Lambda}^\prime(\underline{\bm\Sigma}^\xi)^{-1}\wh{\bm\Gamma}^x(\underline{\bm\Sigma}^\xi)^{-1}
\r\}\bm C_{n,r}\text{vec}\l(\mathrm d\underline{\bm\Lambda}^\prime\r)\nn\\
&-\frac T{\underline{\sigma}_i^2}\l\{\mbf I_r
\otimes 
\underline{\bm P}^{-1}\underline{\bm\Lambda}^\prime(\underline{\bm\Sigma}^\xi)^{-1}\wh{\bm\Gamma}^x(\underline{\bm\Sigma}^\xi)^{-1}\underline{\bm\Lambda}
\r\}\text{vec}\l(\mathrm d\underline{\bm P}^{-1}\r).\label{eq:pearljam}
\end{align}
And, since
\beq
\text{vec}\l(\mathrm d \bm h_{ii}^{*\prime} (\bm{\mathcal X};\underline{\bm\varphi})\r) =\l(
\frac{\partial \text{vec}(\bm h_{ii}^* (\bm{\mathcal X};\underline{\bm\varphi}))^\prime}
{\partial \text{vec}(\underline{\bm\Lambda})^\prime}
\r)^\prime \text{vec}\l(\mathrm d\underline{\bm\Lambda}^\prime\r),\nn
\eeq
from \eqref{eq:pearljam} and \eqref{eq:diffP},
\begin{align}
\l(
\frac{\partial \text{vec}(\bm h_{ii}^* (\bm{\mathcal X};\underline{\bm\varphi}))^\prime}
{\partial \text{vec}(\underline{\bm\Lambda})^\prime}
\r)^\prime=&\,
\frac T{\underline{\sigma}_i^2}\l[
\underline{\bm P}^{-1}\underline{\bm\Lambda}^\prime(\underline{\bm\Sigma}^\xi)^{-1}\wh{\bm\Gamma}^x(\underline{\bm\Sigma}^\xi)^{-1}\underline{\bm\Lambda}\,\underline{\bm P}^{-1}\underline{\bm\Lambda}^\prime(\underline{\bm\Sigma}^\xi)^{-1}
\otimes
\underline{\bm P}^{-1}
\r]\nn\\
&+\frac T{\underline{\sigma}_i^2}\l[
\underline{\bm P}^{-1}\underline{\bm\Lambda}^\prime(\underline{\bm\Sigma}^\xi)^{-1}\wh{\bm\Gamma}^x(\underline{\bm\Sigma}^\xi)^{-1}\underline{\bm\Lambda}\,\underline{\bm P}^{-1}
\otimes \underline{\bm P}^{-1}\underline{\bm\Lambda}^\prime(\underline{\bm\Sigma}^\xi)^{-1}
\r]\bm C_{n,r}\nn\\
&-\frac T{\underline{\sigma}_i^2}\l[
\underline{\bm P}^{-1}\underline{\bm\Lambda}^\prime(\underline{\bm\Sigma}^\xi)^{-1}\wh{\bm\Gamma}^x(\underline{\bm\Sigma}^\xi)^{-1}
\otimes
\underline{\bm P}^{-1}
\r]\nn\\
&-\frac T{\underline{\sigma}_i^2}\l[
\underline{\bm P}^{-1}
\otimes
\underline{\bm P}^{-1}\underline{\bm\Lambda}^\prime(\underline{\bm\Sigma}^\xi)^{-1}\wh{\bm\Gamma}^x(\underline{\bm\Sigma}^\xi)^{-1}
\r]\bm C_{n,r}\nn\\
&+\frac T{\underline{\sigma}_i^2}\l[ \underline{\bm P}^{-1}\underline{\bm\Lambda}^\prime(\underline{\bm\Sigma}^\xi)^{-1}
\otimes
\underline{\bm P}^{-1}\underline{\bm\Lambda}^\prime(\underline{\bm\Sigma}^\xi)^{-1}\wh{\bm\Gamma}^x(\underline{\bm\Sigma}^\xi)^{-1}\underline{\bm\Lambda}\,\underline{\bm P}^{-1}
\r]\nn\\
&+\frac T{\underline{\sigma}_i^2}\l[
\underline{\bm P}^{-1}
\otimes
\underline{\bm P}^{-1}\underline{\bm\Lambda}^\prime(\underline{\bm\Sigma}^\xi)^{-1}\wh{\bm\Gamma}^x(\underline{\bm\Sigma}^\xi)^{-1}\underline{\bm\Lambda}\,\underline{\bm P}^{-1}\underline{\bm\Lambda}^\prime(\underline{\bm\Sigma}^\xi)^{-1}
\r]\bm C_{n,r},\label{eq:duke}
\end{align}
which is an $r^2\times nr$ matrix. From \eqref{eq:duke} we have
\begin{align}
\frac{\partial \text{vec}(\bm h^*_{ii}(\bm{\mathcal X};\underline{\bm\varphi}))^\prime}{\partial \underline{\bm\lambda}_i^\prime}=&\,
\frac {2T}{\underline{\sigma}_i^4}\l[
\underline{\bm P}^{-1}\underline{\bm\Lambda}^\prime(\underline{\bm\Sigma}^\xi)^{-1}\wh{\bm\Gamma}^x(\underline{\bm\Sigma}^\xi)^{-1}\underline{\bm\Lambda}\,\underline{\bm P}^{-1}\underline{\bm\lambda}_i
\otimes
\underline{\bm P}^{-1}
\r]\nn\\
&+\frac {2T}{\underline{\sigma}_i^4}\l[
\underline{\bm P}^{-1}\underline{\bm\lambda}_i
\otimes 
\underline{\bm P}^{-1}\underline{\bm\Lambda}^\prime(\underline{\bm\Sigma}^\xi)^{-1}\wh{\bm\Gamma}^x(\underline{\bm\Sigma}^\xi)^{-1}\underline{\bm\Lambda}\,\underline{\bm P}^{-1}
\r]\nn\\
&-\frac {2T}{\underline{\sigma}_i^4}\l[
\underline{\bm P}^{-1}\underline{\bm\Lambda}^\prime(\underline{\bm\Sigma}^\xi)^{-1}[\wh{\bm\Gamma}^x]_{\cdot i}
\otimes
\underline{\bm P}^{-1}
\r].
\label{eq:dukei}
\end{align}
By using \eqref{eq:varieOp} into \eqref{eq:dukei} and because of \eqref{eq:approxHXigen}, we have 
\beq\label{eq:TXiii}
\l\Vert \text{vec}\l(\bm t_{iii}(\bm{\mathcal X};\bm\varphi)\r)\r\Vert = \l\Vert \bm t_{iii}(\bm{\mathcal X};\bm\varphi)\r\Vert = O_{\mathrm P}\l(\frac T{ n}\r).
\eeq
Moreover, from \eqref{eq:HXF} it immediately follows that $\text{vec}\l(\bm T(\bm{\mathcal X}|\bm{\mathcal F};{\bm\varphi})\r) = \mbf 0_{n^3r^3}$.
and, thus,
$\text{vec}\l(\bm t_{iii}(\bm{\mathcal X}|\bm{\mathcal F};{\bm\varphi})\r) = \mbf 0_{r^3} $, or, equivalently, 
\beq\label{eq:TXFiii}
\bm t_{iii}(\bm{\mathcal X}|\bm{\mathcal F};{\bm\varphi}) = \mbf 0_{r\times r\times r},
\eeq
i.e., a 3rd order tensor of zeros.

Finally, 
\begin{align}
\l\Vert
\bm h_{ii}(\bm{\mathcal X};\wh{\bm{\varphi}}^{\text{\tiny QML,E}})-\bm h_{ii}(\bm{\mathcal X}|\bm{\mathcal F};\wh{\bm{\varphi}}^{\text{\tiny QML,E}})
\r\Vert
\le&\, \l\Vert
\bm h_{ii}(\bm{\mathcal X};{\bm{\varphi}})-\bm h_{ii}(\bm{\mathcal X}|\bm{\mathcal F};{\bm{\varphi}})
\r\Vert\nn\\
&+\l\Vert
\bm t_{iii}(\bm{\mathcal X};{\bm{\varphi}})-\bm t_{iii}(\bm{\mathcal X}|\bm{\mathcal F};{\bm{\varphi}})
\r\Vert\, \l\Vert \wh{\bm\lambda}_i^{\text{\tiny QML,E}}-{\bm\lambda}_i \r\Vert+O_{\mathrm P}\l(\l\Vert \wh{\bm\lambda}_i^{\text{\tiny QML,E}}-{\bm\lambda}_i \r\Vert^2\r)\nn\\
=&\,O_{\mathrm P}\l(\max\l(\frac {T}n,\frac {\sqrt T}{\sqrt{n}}\r)\r)+O_{\mathrm P}\l(\frac {T}{n}\r)O_{\mathrm P}\l(\max\l(\frac {1}n,\frac {1}{\sqrt{T}}\r)\r),\nn
\end{align}
which follows from part (b), \eqref{eq:TXiii}, \eqref{eq:TXFiii}, and \eqref{eq:QMLPCOLS} in the proof of Corollary \ref{cor:QMLcons}, which, in turn, holds even if we use the  log-likelihood \eqref{eq:LL00} in place of thel log-likelihood \eqref{eq:LL0}. This completes the proof. \hfill $\Box$

\setcounter{equation}{0}

\section{Auxiliary propositions}\label{app:B}

\begin{prop}\label{prop:K00}
Under Assumptions \ref{ass:common} and \ref{ass:ident},
\begin{compactenum}[(a)]
\item $\frac{\bm\Lambda^\prime\bm\Lambda}{n} = \frac{\mbf M^\chi}{n}$, for all $n\in\mathbb N$;
\item $\bm\Lambda=\mbf V^\chi(\mbf M^\chi)^{1/2}$, for all $n\in\mathbb N$.
\end{compactenum}
\end{prop}

\noindent
\textbf{Proof of Proposition \ref{prop:K00}.} For part (a), first notice that Assumptions \ref{ass:common}(c) and \ref{ass:ident}(b), imply $\bm\Gamma^F=\mbf I_r$ which, in turn, implies $\bm\Gamma^\chi=\bm\Lambda\bm\Lambda^\prime$. Therefore, since the non-zero eigenvalues of $\frac{\bm\Gamma^\chi}n$ are the same as the $r$ eigenvalues of $\frac{\bm\Lambda^\prime\bm\Lambda}{n}$, which is diagonal by Assumption \ref{ass:ident}(a). Then, we must have, for all $n\in\mathbb N$, $\frac{\bm\Lambda^\prime\bm\Lambda}{n} = \frac{\mbf M^\chi}{n}$. Equivalently, for a given $T$, from Assumption \ref{ass:ident}(b), we have $\wh{\bm\Gamma}^\chi=\frac 1T \sum_{t=1}^T\bm\chi_t\bm\chi_t^\prime=\bm\Lambda\bm\Lambda^\prime$ which implies that $\frac{\bm\Lambda^\prime\bm\Lambda}{n} = \frac{\wh{\mbf M}^\chi}{n}$. However, since $\wh{\bm\Gamma}^\chi = {\bm\Gamma}^\chi$ it must be that $\frac{\wh{\mbf M}^\chi}{n}=\frac{{\mbf M}^\chi}{n}$. This proves part (a). \smallskip

For part (b), since $\bm\Gamma^\chi=\mbf V^\chi\mbf M^\chi\mbf V^{\chi\prime}$, it must be that 
\beq\label{eq:LKstar}
\bm\Lambda\mbf K_*= \mbf V^\chi(\mbf M^\chi)^{1/2},
\eeq
for some $r\times r$ invertible $\mbf K_*$. Now, since $\text{rk}(\frac{\bm\Lambda} {\sqrt n})=r$ for all $n>N$ (see the proof of Proposition \ref{prop:L}):
\beq\label{eq:kappaID}
\mbf K_*= (\bm\Lambda^\prime\bm\Lambda)^{-1}\bm\Lambda^\prime\mbf V^\chi(\mbf M^\chi)^{1/2}=(\mbf M^\chi)^{-1}\bm\Lambda^\prime\mbf V^\chi(\mbf M^\chi)^{1/2}, 
\eeq
which is also obtained by linear projection, and 
\beq\label{eq:kappainvID}
\mbf K^{-1}_*= (\mbf M^\chi)^{-1/2}{\mbf V^{\chi\prime}\bm\Lambda}.
\eeq
Notice that since $\mbf K_*$ is a special case of the matrix $\mbf K$ defined in \eqref{eq:kappa} in the proof of Proposition \ref{prop:L}, then $\mbf K_*$  
is finite and positive definite because of Lemma \ref{lem:KO1}(i) and \ref{lem:KO1}(ii), respectively, hence, $\mbf K^{-1}_*$ in \eqref{eq:kappainvID} is well defined. 

Moreover, from \eqref{eq:kappaID}, because of Assumption \ref{ass:ident}(b) and part (a):
\begin{align}
\mbf K_*\mbf K_*^\prime&=(\mbf M^\chi)^{-1}\bm\Lambda^\prime\mbf V^\chi\mbf M^\chi\mbf V^{\chi\prime}\bm\Lambda(\mbf M^\chi)^{-1}\nn\\
&=(\mbf M^\chi)^{-1}\bm\Lambda^\prime\bm\Gamma^\chi\bm\Lambda(\mbf M^\chi)^{-1}\nn\\
&=(\mbf M^\chi)^{-1}\bm\Lambda^\prime\bm\Lambda\bm\Lambda^\prime\bm\Lambda(\mbf M^\chi)^{-1}\nn\\
&=\mbf I_r.\label{eq:KKstar}
\end{align}
So because of \eqref{eq:KKstar}, we have that $\mbf K_*$ is an orthogonal matrix, i.e., $\mbf K_*=\mbf K_*^{-1}$.
Finally, by \eqref{eq:kappaID} we also have
\beq\label{eq:VKstar}
\mbf V^\chi = \bm\Lambda\mbf K_*(\mbf M^\chi)^{-1/2}
\eeq
and by substituting \eqref{eq:VKstar} into \eqref{eq:kappainvID}, because of \eqref{eq:KKstar},
\[
\mbf K_*^{-1} = (\mbf M^\chi)^{-1}\mbf K_*^\prime \bm\Lambda^\prime \bm\Lambda=(\mbf M^\chi)^{-1}\mbf K_*^{-1}\bm\Lambda^\prime \bm\Lambda,
\]
which is equivalent to:
\beq\label{eq:LKstarT}
\mbf K_*^{-1}\bm\Lambda^\prime \bm\Lambda \mbf K_*= \mbf M^\chi,
\eeq
and by part (a), we must have $\mbf K_*=\mbf I_r$. Alternatively, by multiplying on the right  both sides of \eqref{eq:LKstar} by their transposed:
\beq\label{eq:LKstarT2}
\mbf K_*^\prime\bm\Lambda^\prime \bm\Lambda \mbf K_*= \mbf M^\chi,
\eeq
since eigenvectors are normalized, and again by part (a), we must have $\mbf K_*=\mbf I_r$. So from \eqref{eq:LKstar} or \eqref{eq:LKstarT2} we prove part (b). This completes the proof. \hfill $\Box$\medskip

\begin{prop}\label{prop:L}
Under Assumptions \ref{ass:common} through \ref{ass:ind}, as $n,T\to\infty$
 $\min(n,\sqrt T)\l\Vert\frac{\wh{\bm\Lambda}-\bm\Lambda\bm{\mathcal H}}{\sqrt n}\r\Vert= O_{\mathrm P}(1)$,
where $\bm{\mathcal H}=(\bm\Lambda^\prime\bm\Lambda)^{-1}\bm\Lambda^\prime\mbf V^\chi (\mbf M^\chi)^{1/2}\mbf J$
and $\mbf J$ is an $r\times r$ diagonal matrix with entries $\pm 1$.
\end{prop}

\begin{proof}
Notice that $\text{rk}\l(\frac{\bm\Lambda}{\sqrt n}\r)=r$ for all $n$, since $\text{rk}(\bm\Gamma^F)=r$ by Assumption \ref{ass:common}(b) and $\text{rk}\l(\frac{\bm\Gamma^\chi}n \r)=r$ by Lemma \ref{lem:Gxi}(iv).
Indeed, $\text{rk}\l(\frac{\bm\Gamma^\chi}n \r)\le \min (\text{rk}(\bm\Gamma^F),\text{rk}\l(\frac{\bm\Lambda}{\sqrt n}\r))$.
This holds for all $n>N$ and since eigenvalues are an increasing sequence in $n$. Therefore, $(\frac{\bm\Lambda^\prime\bm\Lambda}n)^{-1}$ is well defined for all $n$ and $\bm\Lambda$ admits a left inverse.

Second, since
\beq
\frac{\bm\Gamma^\chi}n = \mbf V^\chi\frac{\mbf M^\chi}n\mbf V^{\chi\prime}= \frac{\bm\Lambda}{\sqrt n}\bm\Gamma^F\frac{\bm\Lambda^\prime}{\sqrt n}. 
\eeq
 the columns of $\frac{\mbf V^\chi(\mbf M^\chi)^{1/2}}{\sqrt n}$ and the columns of $\frac{\bm\Lambda(\bm\Gamma^F)^{1/2}}{\sqrt n}$ must span the same space. 

So there exists an $r\times r$ invertible matrix $\mbf K$ such that
\beq\label{eq:kappaproj}
{\bm\Lambda(\bm\Gamma^F)^{1/2}}\mbf K= {\mbf V^\chi(\mbf M^\chi)^{1/2}}
\eeq
Therefore, from \eqref{eq:kappaproj}
\beq\label{eq:kappa}
\mbf K=(\bm\Gamma^F)^{-1/2} (\bm\Lambda^\prime\bm\Lambda)^{-1}\bm\Lambda^\prime\mbf V^\chi(\mbf M^\chi)^{1/2}
\eeq
which is also obtained by linear projection, and also
\beq\label{eq:kappainv}
\mbf K^{-1}= (\mbf M^\chi)^{-1/2}{\mbf V^{\chi\prime}\bm\Lambda}(\bm\Gamma^F)^{1/2}
\eeq
which are both finite and positive definite because of Lemma \ref{lem:KO1}.

Now, from \eqref{eq:kappaproj} and \eqref{eq:kappa}
\begin{align}\label{eq:VKH}
{\mbf V}^\chi ( {\mbf M}^\chi)^{1/2}= \bm\Lambda(\bm\Gamma^F )^{1/2}\mbf K
= \bm\Lambda(\bm\Lambda^\prime\bm\Lambda)^{-1}\bm\Lambda^\prime\mbf V^\chi(\mbf M^\chi)^{1/2}
\end{align}
Let
\beq\label{eq:mcH}
\bm{\mathcal H}=(\bm\Lambda^\prime\bm\Lambda)^{-1}\bm\Lambda^\prime\mbf V^\chi(\mbf M^\chi)^{1/2}\mbf J,
\eeq 
which is finite and positive definite because of Lemma \ref{lem:HO1bis}.

Now, 
because of Lemmas \ref{lem:covarianze}(iii), \ref{lem:covarianze}(iv), \ref{lem:MO1}(i), using \eqref{eq:estL}, \eqref{eq:VKH}, and \eqref{eq:mcH},
\begin{align}
\l\Vert\frac { \wh{\bm\Lambda} - \bm\Lambda\bm{\mathcal H}}{\sqrt n}\r\Vert=&\,
\l\Vert\wh{\mbf V}^x\l(\frac{\wh{\mbf M}^x}n\r)^{1/2} - {\mbf V}^\chi \l(\frac{ {\mbf M}^\chi}n\r)^{1/2} \mbf J\r\Vert=\ \l\Vert\wh{\mbf V}^x\l(\frac{\wh{\mbf M}^x}n\r)^{1/2} - {\mbf V}^\chi\mbf J \l(\frac{ {\mbf M}^\chi}n\r)^{1/2} \r\Vert\nn\\
\le&\, \l\Vert \wh{\mbf V}^x-{\mbf V}^\chi\mbf J\r\Vert\,\l\Vert\frac{\mbf M^\chi}{n} \r\Vert+ 
\l\Vert\frac 1 {\sqrt n}\l\{\l(\wh{\mbf M}^x\r)^{1/2}-\l(\mbf M^\chi\r)^{1/2}\r\}\r\Vert \,\l\Vert \mbf V^\chi\r\Vert\nn\\
&+ \l\Vert \wh{\mbf V}^x-{\mbf V}^\chi\mbf J\r\Vert\,
\l\Vert\frac 1 {\sqrt n}\l\{\l(\wh{\mbf M}^x\r)^{1/2}-\l(\mbf M^\chi\r)^{1/2}\r\}\r\Vert\nn\\
&= O_{\mathrm P}\l(\max\l(\frac 1 n,\frac 1{\sqrt T}\r)\r)+ O_{\mathrm P}\l(\max\l(\frac 1 {n^2},\frac 1{T}\r)\r), \nn
\end{align}
since $\Vert \mbf V^\chi\Vert=1$ because eigenvectors are normalized.
This completes the proof.
\end{proof}

\begin{prop}\label{prop:load}
Under Assumptions \ref{ass:common} through \ref{ass:ind} the terms in \eqref{eq:sviluppoLambda} are such that, as $n,T\to\infty$, 
\begin{compactenum}[(a)]
\item $\sqrt {nT}\l\Vert \text{\upshape (1.a)}\r\Vert = O_{\mathrm P}(1)$;
\item $\sqrt T \l\Vert \text{\upshape (1.b)}\r\Vert = O_{\mathrm P}(1)$;
\item $\min(n,\sqrt{nT})\l\Vert \text{\upshape (1.c)}\r\Vert = O_{\mathrm P}(1)$;
\item $\min(\sqrt{nT},T)\l\Vert \text{\upshape (1.d)}\r\Vert = O_{\mathrm P}(1)$;
\item $\min(\sqrt{nT},T)\l\Vert \text{\upshape (1.e)}\r\Vert = O_{\mathrm P}(1)$;
\item $\min(n,\sqrt{nT},T)\l\Vert \text{\upshape (1.f)}\r\Vert = O_{\mathrm P}(1)$;
\end{compactenum}
uniformly in $i$.
\end{prop}

\begin{proof}
For part (a), for any $i=1,\ldots,n$, by Assumption \ref{ass:common}(a),
\beq
\l\Vert\frac 1{nT}{\bm\lambda}_i^\prime\sum_{t=1}^T\sum_{j=1}^n\mbf F_t\xi_{jt}{\bm\lambda}_j^\prime\r\Vert \le \l\Vert \bm\lambda_i\r\Vert\,\l\Vert\frac 1{nT}\sum_{t=1}^T\sum_{j=1}^n\mbf F_t\xi_{jt}{\bm\lambda}_j^\prime \r\Vert\le M_\Lambda\l\Vert\frac 1{nT}\sum_{t=1}^T\sum_{j=1}^n\mbf F_t\xi_{jt}{\bm\lambda}_j^\prime \r\Vert.\label{eq:1a}
\eeq
Then, by Assumptions \ref{ass:common}(a) and \ref{ass:ind}
\begin{align}
\E\l[\l\Vert\frac 1{nT}\sum_{t=1}^T\sum_{j=1}^n\mbf F_t\xi_{jt}{\bm\lambda}_j^\prime \r\Vert^2\r] 
&\le \E\l[\l\Vert\frac 1{nT}\sum_{t=1}^T\sum_{j=1}^n\mbf F_t\xi_{jt}{\bm\lambda}_j^\prime \r\Vert_F^2\r] 
=\frac 1{n^2T^2} \sum_{k=1}^r\sum_{h=1}^r \E\l[\l(\sum_{t=1}^T F_{kt}\sum_{j=1}^n \xi_{jt}[\bm\Lambda_{jh}]\r)^2\r] \nn\\
&\le \frac{r^2}{n^2T^2} \max_{h,k=1,\ldots, r}\sum_{t=1}^T\sum_{s=1}^T \E\l[F_{kt} \l(\sum_{j=1}^n \xi_{jt}[\bm\Lambda_{jh}]\r) F_{ks}\l(\sum_{\ell=1}^n \xi_{\ell s}[\bm\Lambda_{\ell h}]\r) \r]\nn\\
&= \frac{r^2}{n^2T^2} \max_{h,k=1,\ldots, r}\sum_{t=1}^T\sum_{s=1}^T \E[F_{kt}  F_{ks}]\,\sum_{j=1}^n\sum_{\ell=1}^n \E\l[\xi_{jt} \xi_{\ell s} \r] [\bm\Lambda_{jh}][\bm\Lambda_{\ell h}]\nn\\
&\le \l\{\frac{r^2}{T}\max_{k=1,\ldots, r}\sum_{t=1}^T\sum_{s=1}^T \E\l[F_{kt}F_{ks}\r]\r\}
\,\l\{\max_{t,s=1\ldots,T} \frac{M_{\Lambda}^2}{n^2T}\sum_{j=1}^n\sum_{\ell=1}^n\l\vert \E[\xi_{jt}\xi_{\ell s}]\r\vert\r\}.\label{eq:1a2}
\end{align}
Now, by Cauchy-Schwarz inequality, for any $k=1,\ldots, r$, 
\beq\label{eq:1aCS}
\l\vert \frac{1}{T}\sum_{t=1}^T\sum_{s=1}^TF_{kt}F_{ks}\r\vert\le \l(\frac 1 T\sum_{t=1}^T F_{kt}^2\r)^{1/2}\l(\frac 1 T\sum_{s=1}^T F_{ks}^2\r)^{1/2},
\eeq
and, by Assumption \ref{ass:common}(b), using \eqref{eq:1aCS} and again Cauchy-Schwarz inequality, 
\begin{align}
\max_{k=1,\ldots, r}\frac{1}{T}\sum_{t=1}^T\sum_{s=1}^T\E[F_{kt}F_{ks}]&
=\max_{k=1,\ldots, r}\E\l[\frac{1}{T}\sum_{t=1}^T\sum_{s=1}^TF_{kt}F_{ks}\r]\le\max_{k=1,\ldots, r}\E\l[\l\vert \frac{1}{T}\sum_{t=1}^T\sum_{s=1}^TF_{kt}F_{ks}\r\vert\r]\nn\\
&\le \max_{k=1,\ldots, r} \E\l[\l(\frac 1 T\sum_{t=1}^T F_{kt}^2\r)^{1/2}\l(\frac 1 T\sum_{s=1}^T F_{ks}^2\r)^{1/2}\r]\nn\\
&\le \max_{k=1,\ldots, r} 
\l(\E\l[\l(\frac 1 T\sum_{t=1}^T F_{kt}^2\r)\r] \r)^{1/2}
\l(\E\l[\l(\frac 1 T\sum_{s=1}^T F_{ks}^2\r)\r] \r)^{1/2}\nn\\
&\le \max_{k=1,\ldots, r} 
\l(\frac 1 T\sum_{t=1}^T \E[F_{kt}^2] \r)^{1/2}
\l(\frac 1 T\sum_{s=1}^T \E[F_{ks}^2] \r)^{1/2}\nn\\
&=\max_{k=1,\ldots, r}  \frac{1}{T}\sum_{t=1}^T \E[F_{kt}^2]\le \max_{t=1,\ldots,T}\max_{k=1,\ldots, r}\E[F_{kt}^2]\nn\\
&= \max_{k=1,\ldots, r} \bm\eta_k^\prime \bm\Gamma_F\bm\eta_k \le \Vert\bm\Gamma^F\Vert\le M_F,\label{eq:1a21}
\end{align}
since $M_F$ is independent of $t$ and where $\bm\eta_k$ is an $r$-dimensional vector with one in the $k$th entry and zero elsewhere. And, because of Lemma \ref{lem:Gxi}(ii)
\beq
\max_{t,s=1,\ldots,T}\frac 1{n^2T}\sum_{j=1}^n\sum_{\ell=1}^n\l\vert \E[\xi_{jt}\xi_{\ell s}]\r\vert\le \frac{M_\xi}{nT},\label{eq:1a22}
\eeq
since $M_\xi$ is independent of $n$ and $t$.
By substituting \eqref{eq:1a21} and \eqref{eq:1a22} into \eqref{eq:1a2}, 
\beq
\E\l[\l\Vert\sum_{t=1}^T\sum_{j=1}^n\mbf F_t\xi_{jt}{\bm\lambda}_j^\prime \r\Vert^2\r] \le \frac{r^2M_FM_\Lambda^2M_\xi}{nT}.\label{eq:1a3}
\eeq 
By substituting \eqref{eq:1a3} into \eqref{eq:1a}, we prove part (a).\smallskip

For part (b), for any $i=1,\ldots,n$, because of Lemma \ref{lem:FTLN}(i),
\beq
\l\Vert\frac 1{nT} \sum_{t=1}^T \xi_{it}\mbf F_t^\prime\sum_{j=1}^n\bm\lambda_j\bm\lambda_j^\prime\r\Vert=
\l\Vert\frac 1{nT} \sum_{t=1}^T \xi_{it}\mbf F_t^\prime(\bm\Lambda^\prime\bm\Lambda)\r\Vert\le 
\l\Vert\frac 1{T} \sum_{t=1}^T \xi_{it}\mbf F_t\r\Vert\,\l\Vert
\frac{\bm\Lambda}{\sqrt n}
\r\Vert^2\le \l\Vert\frac 1{T} \sum_{t=1}^T \xi_{it}\mbf F_t\r\Vert M_\Lambda^2.\label{eq:1b}
\eeq
Then, by Assumptions \ref{ass:ind} and \ref{ass:idio}(b) and using \eqref{eq:1a21} 
\begin{align}
\E\l[\l\Vert\frac 1{T} \sum_{t=1}^T \xi_{it}\mbf F_t\r\Vert^2\r]&=\frac 1{T^2} \sum_{j=1}^r \E\l[\l(\sum_{t=1}^T \xi_{it}F_{jt}\r)^2\r]\nn\\
&\le \frac{r}{T^2}\max_{j=1,\ldots, r}\sum_{t=1}^T\sum_{s=1}^T \E[\xi_{it}F_{jt}\xi_{is}F_{js}]=
 \frac{r}{T^2}\max_{j=1,\ldots, r}\sum_{t=1}^T\sum_{s=1}^T \E[F_{jt}F_{js}]\,\E[\xi_{it}\xi_{is}]\nn\\
&\le \l\{\frac{r}{T}\max_{j=1,\ldots, r}\sum_{t=1}^T\sum_{s=1}^T \E[F_{jt}F_{js}]\r\}
\l\{\frac 1T \max_{t,s=1,\ldots, n} \l\vert\E[\xi_{it}\xi_{is}]\r\vert\r\}\nn\\
&\le rM_F\frac{ M_{ii}}T \max_{t,s=1,\ldots, n}\rho^{|t-s|} \le  \frac{rM_F M_{\xi}}T,\label{eq:1b2}
\end{align}
since $M_\xi$ is independent of $i$. Or, equivalently, by Lemma \ref{lem:Gxi}(iii) and Cauchy-Schwarz inequality 
\begin{align}
\E\l[\l\Vert\frac 1{T} \sum_{t=1}^T \xi_{it}\mbf F_t\r\Vert^2\r]
&\le  \frac{r}{T^2}\max_{j=1,\ldots, r}\sum_{t=1}^T\sum_{s=1}^T \E[F_{jt}F_{js}]\,\E[\xi_{it}\xi_{is}]\nn\\
&\le \l\{\frac{r}{T}\max_{j=1,\ldots, r}\max_{t,s=1,\ldots, n} \vert\E[F_{jt}F_{js}]\vert\r\}
\l\{\frac 1T  \sum_{t=1}^T\sum_{s=1}^T\l\vert\E[\xi_{it}\xi_{is}]\r\vert\r\}\nn\\
&\le \frac r T \max_{j=1,\ldots, r}\max_{t,s=1,\ldots, n} \E[F_{jt}^2] \frac{M_\xi(1+\rho)}{1-\rho}
\le \frac{r M_F M_{3\xi}}{T},\label{eq:1b3}
\end{align}
since $M_F$ is independent of $t$ and $M_{3\xi}$ is independent of $i$. Notice that $M_\xi\le M_{3\xi}$. Notice that \eqref{eq:1b3} is a special case of Lemma \ref{lem:LLN}(i). By substituting \eqref{eq:1b2}, or \eqref{eq:1b3}, into \eqref{eq:1b}, we prove part (b).
\smallskip

For part (c), for any $i=1,\ldots,n$,  because of Assumption \ref{ass:common}(a),
\begin{align}
\l\Vert\frac 1{nT} \sum_{t=1}^T\sum_{j=1}^n\xi_{it}\xi_{jt} \bm\lambda_j^\prime\r\Vert &=\l\{\sum_{k=1}^r \l(\frac 1{nT} \sum_{t=1}^T\sum_{j=1}^n\xi_{it}\xi_{jt} \lambda_{jk}\r)^2\r\}^{1/2}\le \sqrt rM_\Lambda\l\vert\frac 1{nT} \sum_{t=1}^T\sum_{j=1}^n\xi_{it}\xi_{jt} \r\vert\nn\\
&\le\sqrt r M_\Lambda\l\{ \l\vert\frac 1{nT} \sum_{t=1}^T\sum_{j=1}^n\l\{\xi_{it}\xi_{jt}-\E[\xi_{it}\xi_{jt}]\r\} \r\vert+\l\vert\frac 1{nT} \sum_{t=1}^T\sum_{j=1}^n\E[\xi_{it}\xi_{jt}] \r\vert\r\}.\label{eq:1c}
\end{align}
Then, by Assumption \ref{ass:idio}(b),
\begin{align}
\l\vert\frac 1{nT} \sum_{t=1}^T\sum_{j=1}^n\E[\xi_{it}\xi_{jt}] \r\vert\le 
\frac 1{nT}\sum_{t=1}^T\sum_{j=1}^n\l\vert\E[\xi_{it}\xi_{jt}]\r\vert\le 
\max_{t=1,\ldots,T}\frac 1n\sum_{j=1}^n\l\vert\E[\xi_{it}\xi_{jt}]\r\vert\le\frac 1n \sum_{j=1}^n M_{ij} \le \frac{M_\xi}n,\label{eq:1c2}
\end{align}
since $M_\xi$ is independent of $i$ and $t$. Moreover, by Assumption \ref{ass:idio}(c), 
\beq\label{eq:1c3}
\E\l[\l\vert\frac 1{nT} \sum_{t=1}^T\sum_{j=1}^n\l\{\xi_{it}\xi_{jt}-\E[\xi_{it}\xi_{jt}]\r\} \r\vert^2\r]\le \frac{K_\xi}{nT}.
\eeq
By substituting \eqref{eq:1c2} and \eqref{eq:1c3} into \eqref{eq:1c},  we prove part (c).
\smallskip

For part (d),  for any $i=1,\ldots,n$, because of Assumption \ref{ass:common}(a)
\begin{align}
\l\Vert\frac 1{nT}{\bm\lambda}_i^\prime\sum_{t=1}^T\sum_{j=1}^n\mbf F_t\xi_{jt}(\wh{\bm\lambda}_j^\prime-{\bm\lambda}_j^\prime\bm{\mathcal H})\r\Vert &\le M_\Lambda\l\Vert\frac 1{nT}\sum_{t=1}^T\sum_{j=1}^n\mbf F_t\xi_{jt}(\wh{\bm\lambda}_j^\prime-{\bm\lambda}_j^\prime\bm{\mathcal H}) \r\Vert\nn\\
&= M_\Lambda\l\Vert\frac {\bm F^\prime\bm \Xi (\wh{\bm\Lambda}-\bm\Lambda\bm{\mathcal H} )}{nT} \r\Vert\le M_\Lambda \l\Vert\frac {\bm F^\prime\bm \Xi}{\sqrt n T} \r\Vert\, \l\Vert\frac{\wh{\bm\Lambda}-\bm\Lambda\bm{\mathcal H}}{\sqrt n}\r\Vert.\label{eq:1d}
\end{align}
Then, by using Lemma \ref{lem:LLN} and Proposition \ref{prop:L} in \eqref{eq:1d},   we prove part (d).\smallskip

For part (e),  for any $i=1,\ldots,n$,
\begin{align}
\l\Vert\frac 1{nT} \sum_{t=1}^T \xi_{it}\mbf F_t^\prime\sum_{j=1}^n\bm\lambda_j(\wh{\bm\lambda}_j^\prime-{\bm\lambda}_j^\prime\bm{\mathcal H})\r\Vert&=
\l\Vert\frac 1{nT} \sum_{t=1}^T \xi_{it}\mbf F_t^\prime\bm\Lambda^\prime(\wh{\bm\Lambda}-\bm\Lambda\bm{\mathcal H})\r\Vert\nn\\ 
&\le \l\Vert\frac 1{T} \sum_{t=1}^T \xi_{it}\mbf F_t\r\Vert\,\l\Vert
\frac{\bm\Lambda}{\sqrt n}
\r\Vert \, \l\Vert\frac{\wh{\bm\Lambda}-\bm\Lambda\bm{\mathcal H}}{\sqrt n}\r\Vert .\label{eq:1e}
\end{align}
By substituting part (ii), Lemma \ref{lem:FTLN}(i), and part (a) of Proposition \ref{prop:L} into \eqref{eq:1e}, we prove part (e).\smallskip

Finally, for part (f),  for any $i=1,\ldots,n$, let $\bm\zeta_i=(\xi_{i1}\cdots \xi_{iT})^\prime$, then
\begin{align}
\l\Vert\frac 1{nT} \sum_{t=1}^T\sum_{j=1}^n\xi_{it}\xi_{jt} (\wh{\bm\lambda}_j^\prime-{\bm\lambda}_j^\prime\bm{\mathcal H})\r\Vert &= \l\Vert\frac{\bm\zeta_i^\prime\bm\Xi\l(\wh{\bm\Lambda}-\bm\Lambda\bm{\mathcal H}\r)}{nT} \r\Vert \le \l\Vert\frac{\bm\zeta_i^\prime\bm\Xi}{\sqrt n T}\r\Vert\, \l\Vert\frac{\wh{\bm\Lambda}-\bm\Lambda\bm{\mathcal H}}{\sqrt n}\r\Vert.\label{eq:1f}
\end{align}
Then, by the $C_r$-inequality with $r=2$,
\begin{align}
\l\Vert\frac{\bm\zeta_i^\prime\bm\Xi}{\sqrt n T}\r\Vert^2&=
\l\Vert\frac 1 {\sqrt n T} \sum_{t=1}^T \xi_{it} \bm\xi_t^\prime\r\Vert^2
 =\frac 1{n} \sum_{j=1}^n \l(\frac 1T\sum_{t=1}^T \xi_{it}\xi_{jt}\r)^2\nn\\
& =\frac 1{n} \sum_{j=1}^n \l(\frac 1T\sum_{t=1}^T \l\{\xi_{it}\xi_{jt}-\E[\xi_{it}\xi_{jt}]\r\}+\frac 1T\sum_{t=1}^T \E[\xi_{it}\xi_{jt}]\r)^2\nn\\
&\le\frac 2{n} \sum_{j=1}^n\l\{  \l(\frac 1T\sum_{t=1}^T \l\{\xi_{it}\xi_{jt}-\E[\xi_{it}\xi_{jt}]\r\}\r)^2+\l(\frac 1T\sum_{t=1}^T \E[\xi_{it}\xi_{jt}]\r)^2\r\}\label{eq:1f2}
\end{align}
By taking the expectation of \eqref{eq:1f2}, and because of Assumption \ref{ass:idio}(c) and Lemma \ref{lem:Gxi}(v),
\begin{align}
\E\l[\l\Vert\frac{\bm\zeta_i^\prime\bm\Xi}{\sqrt n T}\r\Vert^2\r]&=\frac 2{n} \sum_{j=1}^n \E\l[
\l(\frac 1T\sum_{t=1}^T \l\{\xi_{it}\xi_{jt}-\E[\xi_{it}\xi_{jt}]\r\}\r)^2
\r]+\frac 2{n} \sum_{j=1}^n \l(\frac 1T\sum_{t=1}^T \E[\xi_{it}\xi_{jt}]\r)^2\nn\\
&\le 2\max_{j=1,\ldots, n}\E\l[
\l(\frac 1T\sum_{t=1}^T \l\{\xi_{it}\xi_{jt}-\E[\xi_{it}\xi_{jt}]\r\}\r)^2
\r]+\frac{2}{n}\sum_{i=1}^n\l(\frac 1T\sum_{t=1}^T \E[\xi_{it}\xi_{jt}]\r)^2\nn\\
&\le \frac{2K_\xi}{T} + \frac{2}{nT^2}\sum_{i=1}^n\sum_{t=1}^T \E[\xi_{it}\xi_{jt}]\sum_{s=1}^T\E[\xi_{is}\xi_{js}]\nn\\
&\le \frac{2K_\xi}{T} + \max_{t=1,\ldots ,T}\frac{2}{n}\sum_{i=1}^n \l\vert\E[\xi_{it}\xi_{jt}]\r\vert \max_{s=1,\ldots,T}\max_{i,j=1,\ldots, n} \E[\xi_{is}\xi_{js}]\nn\\
&\le \frac{2K_\xi}{T} + \frac{2}n \sum_{i=1}^n M_{ij} \max_{i,j=1,\ldots, n} \bm\varepsilon_i^\prime\bm\Gamma^\xi\bm\varepsilon_j\le  \frac{2K_\xi}{T}+ \frac{2M_\xi}{n}\Vert \bm\Gamma^\xi\Vert\le \frac{2K_\xi}{T}+ \frac{2M_\xi M_{2\xi}}{n}, \label{eq:1f3}
\end{align}
since $K_\xi$ is independent of $j$ and $M_\xi$ is independent of $i$, $j$, $t$, and $s$ and where $\bm\varepsilon_i$ is an $n$-dimensional vector with one in the $i$th entry and zero elsewhere.  
By substituting \eqref{eq:1f3} and part (a) of Proposition \ref{prop:L} into \eqref{eq:1f}, we prove part (f). This completes the proof. 
\end{proof}

\begin{prop}\label{prop:KKK}
Under Assumptions \ref{ass:common} through \ref{ass:ind},
\begin{compactenum}[(a)]
\item $
\l\Vert \bm{\mathcal H}^\prime\l(\frac{\bm\Lambda^\prime\bm\Lambda}{n}\r)-\bm Q_0\r\Vert = o(1)$, as $n\to\infty$;
\item $\l\Vert \bm{\mathcal H}^\prime\l(\frac{\bm\Lambda^\prime\bm\Lambda}{n}\r)-\l(\frac{\wh{\bm\Lambda}^\prime\bm\Lambda}{n}\r)\r\Vert = o_{\mathrm P}(1)$, as $n,T\to\infty$.
%
\end{compactenum}
where
$\bm Q_0=\bm V_0\bm{\mathcal J}_0\bm\Upsilon_0^\prime  (\bm\Gamma^F)^{-1/2}$, and where 
$\bm{\mathcal J}_0$ is an $r\times r$ diagonal matrix with entries $\pm 1$, $\bm\Upsilon_0$ is the $r\times r$ matrix having as columns the normalized eigenvectors of $(\bm\Gamma^F)^{1/2}\bm\Sigma_\Lambda(\bm\Gamma^F)^{1/2}$, and $\bm V_0$ is the $r\times r$ matrix of corresponding eigenvalues sorted in descending order.
\end{prop}

\begin{proof} Start with part (a). From \eqref{eq:mcH} and \eqref{eq:kappainv}
\beq
\bm{\mathcal H}^\prime\l(\frac{\bm\Lambda^\prime\bm\Lambda}{n}\r) = \mbf J(\mbf M^\chi)^{1/2} \mbf V^{\chi\prime}\bm\Lambda(\bm\Lambda^\prime\bm\Lambda)^{-1} \l(\frac{\bm\Lambda^\prime\bm\Lambda}{n}\r) =\mbf J\l(\frac{\mbf M^\chi}n\r)^{1/2} \frac{\mbf V^{\chi\prime}\bm\Lambda}{\sqrt n}=\mbf J\l(\frac{\mbf M^\chi}n\r) \mbf K^{-1}(\bm\Gamma^F)^{-1/2}.\label{eq:propKKK}
\eeq
Then, from \eqref{eq:kappa}
\beq\label{eq:vchi}
\mbf V^{\chi} = \bm\Lambda (\bm\Gamma^F)^{1/2}\mbf K ({\mbf M^\chi})^{-1/2}
\eeq
thus, from \eqref{eq:vchi} and \eqref{eq:kappainv} 
\begin{align}
\mbf J\l(\frac{\mbf M^\chi}n\r) \mbf K^{-1} &=\mbf J \l(\frac{\mbf M^\chi}n\r) (\mbf M^\chi)^{-1/2}{\mbf V^{\chi\prime}\bm\Lambda}(\bm\Gamma^F)^{1/2}\nn\\
&=
\mbf J\l(\frac{\mbf M^\chi}n\r) ({\mbf M^\chi})^{-1}\mbf K^\prime(\bm\Gamma^F)^{1/2}\bm\Lambda^\prime \bm\Lambda(\bm\Gamma^F)^{1/2}\nn\\
&= \mbf J\mbf K^\prime(\bm\Gamma^F)^{1/2}\frac{\bm\Lambda^\prime \bm\Lambda}n(\bm\Gamma^F)^{1/2}.
\end{align}
And, from \eqref{eq:kappa} 
\begin{align}
\mbf J\mbf K\mbf K^\prime\mbf J &= \mbf J (\bm\Gamma^F)^{-1/2} (\bm\Lambda^\prime\bm\Lambda)^{-1}\bm\Lambda^\prime\mbf V^\chi(\mbf M^\chi)^{1/2}
(\mbf M^\chi)^{1/2}\mbf V^{\chi\prime}\bm\Lambda(\bm\Lambda^\prime\bm\Lambda)^{-1}(\bm\Gamma^F)^{-1/2} 
\mbf J\nn\\
&=\mbf J (\bm\Gamma^F)^{-1/2} (\bm\Lambda^\prime\bm\Lambda)^{-1}\bm\Lambda^\prime\bm\Gamma^\chi\bm\Lambda(\bm\Lambda^\prime\bm\Lambda)^{-1}(\bm\Gamma^F)^{-1/2} 
\mbf J\nn\\
&=\mbf J (\bm\Gamma^F)^{-1/2} (\bm\Lambda^\prime\bm\Lambda)^{-1}\bm\Lambda^\prime\bm\Lambda\bm\Gamma^F\bm\Lambda^\prime\bm\Lambda(\bm\Lambda^\prime\bm\Lambda)^{-1}(\bm\Gamma^F)^{-1/2} 
\mbf J\nn\\
&= \mbf J (\bm\Gamma^F)^{-1/2} \bm\Gamma^F(\bm\Gamma^F)^{-1/2} 
\mbf J=\mbf I_r.
\end{align}
Therefore, the columns of $\mbf J\mbf K$ are the normalized eigenvectors of $(\bm\Gamma^F)^{1/2}\frac{\bm\Lambda^\prime \bm\Lambda}n(\bm\Gamma^F)^{1/2}$ with eigenvalues $\frac{\mbf M^\chi}n$ (notice that $\mbf J\l(\frac{\mbf M^\chi}n\r)=\l(\frac{\mbf M^\chi}n\r)\mbf J$). Moreover, by Assumption \ref{ass:common}(a)
\beq
\lim_{n\to\infty} \l\Vert(\bm\Gamma^F)^{1/2}\frac{\bm\Lambda^\prime \bm\Lambda}n(\bm\Gamma^F)^{1/2}-(\bm\Gamma^F)^{1/2}\bm\Sigma_\Lambda(\bm\Gamma^F)^{1/2}\r\Vert=0.\label{eq:limB}
\eeq
Letting, $\bm V_0$ be the matrix of eigenvalues of $(\bm\Gamma^F)^{1/2}\bm\Sigma_\Lambda(\bm\Gamma^F)^{1/2}$ sorted in descending order, from \eqref{eq:limB} we also have (this is proved also in Lemma \ref{lem:Vzero}(i))
\beq\label{eq:MchiV0}
\lim_{n\to\infty} \l\Vert\frac{\mbf M^\chi}n - \bm V_0\r\Vert=0.
\eeq
Let $\bm\Upsilon_0$ be the normalized eigenvectors of $(\bm\Gamma^F)^{1/2}\bm\Sigma_\Lambda(\bm\Gamma^F)^{1/2}$
Hence, by continuity of eigenvectors, and since the eigenvalues $\frac{\mbf M^\chi}n$ are distinct because of Assumption \ref{ass:eval}, from \citet[Theorem 2]{yu15}, by Lemma \ref{lem:Vzero}(iii) and using \eqref{eq:limB} it follows that
\beq\label{eq:ups}
\lim_{n\to\infty} \l\Vert\mbf J\mbf K-\bm\Upsilon_0\bm{\mathcal J}_0\r\Vert \le \lim_{n\to\infty}\frac{2^{3/2}\sqrt r  \l\Vert(\bm\Gamma^F)^{1/2}\frac{\bm\Lambda^\prime \bm\Lambda}n(\bm\Gamma^F)^{1/2}-(\bm\Gamma^F)^{1/2}\bm\Sigma_\Lambda(\bm\Gamma^F)^{1/2}\r\Vert}{\mu_r(\bm V_0)}=0,
\eeq
where $\bm{\mathcal J}_0$ is an $r\times r$ diagonal matrix with entries $\pm 1$, which is in general different from $\mbf J$. Finally, since $\bm\Upsilon_0$ is an orthogonal matrix, we have $\bm\Upsilon_0^{-1}=\bm\Upsilon_0^\prime$, so from \eqref{eq:ups}
\beq\label{eq:upsinv}
\lim_{n\to\infty} \l\Vert\mbf K^{-1}\mbf J-\bm{\mathcal J}_0\bm\Upsilon_0^\prime\r\Vert =0.
\eeq
By using \eqref{eq:upsinv} and \eqref{eq:MchiV0} into \eqref{eq:propKKK} and since $\mbf K^{-1}\mbf J=\mbf J\mbf K^{-1}$, we have
\beq\nn
\lim_{n\to\infty}\l\Vert\bm{\mathcal H}^\prime\l(\frac{\bm\Lambda^\prime\bm\Lambda}{n}\r) -\bm V_0\bm{\mathcal J}_0\bm\Upsilon_0^\prime(\bm\Gamma^F)^{-1/2}\r\Vert=0.
\eeq
By defining $\bm Q_0=\bm V_0\bm{\mathcal J}_0\bm\Upsilon_0^\prime(\bm\Gamma^F)^{-1/2}$, we prove part (a).\smallskip

Part (b) follows from Proposition \ref{prop:L} and Lemma \ref{lem:FTLN}(i), and since
\beq
\l\Vert \frac{\wh{\bm\Lambda}^\prime\bm\Lambda}{n}-\bm {\mathcal H}^\prime\frac{\bm\Lambda^\prime\bm\Lambda}{n}\r\Vert \le 
\l\Vert \frac{\wh{\bm\Lambda}^\prime-\bm {\mathcal H}^\prime\bm\Lambda^\prime}{\sqrt n}\r\Vert\,\l\Vert 
\frac{\bm\Lambda}{\sqrt n}\r\Vert = o_{\mathrm P}(1).
\eeq
This completes the proof.
\end{proof}

\begin{prop}\label{prop:factor}
Under Assumptions \ref{ass:common} through \ref{ass:ind}, as $n,T\to\infty$,
 $\min(n,\sqrt{nT},T)\l\Vert \frac{(\bm\Lambda\wh{\mbf H}-\wh{\bm\Lambda})^\prime\wh{\bm\Lambda}}{n}\r\Vert=O_{\mathrm {P}}(1)$, where $\wh{\mbf H}$ is defined in \eqref{eq:acca}.
\end{prop}

\begin{proof}We have,
\begin{align}
\l\Vert \frac{(\bm\Lambda\wh{\mbf H}-\wh{\bm\Lambda})^\prime\wh{\bm\Lambda}}{n}\r\Vert&\le 
\l\{\l\Vert
\frac{(\wh{\bm\Lambda}-\bm\Lambda\wh{\mbf H})^\prime\bm\Lambda\wh{\mbf H}}{n}\r\Vert+\l\Vert \frac{(\wh{\bm\Lambda}-\bm\Lambda\wh{\mbf H})^\prime (\wh{\bm\Lambda}-\bm\Lambda\wh{\mbf H})}{n}\r\Vert\r\}\nn\\
&=\l\{\bm I +\bm{II}\r\},\;\text{say.}\label{eq:2a}
\end{align}
First, consider $\bm I$ in \eqref{eq:2a}. From \eqref{eq:start4} 
\begin{align}
\bm I=&\,\frac {(\wh{\bm\Lambda}-\bm\Lambda\wh{\mbf H})^\prime \bm\Lambda \wh{\mbf H}}n\nn\\
=&\,  \l(\frac{\wh{\mbf M}^x}{n}\r)^{-1}\bm{\mathcal H}^\prime\frac{\bm\Lambda^\prime\bm\Lambda}{n}
\frac{\bm F^\prime\bm\Xi\bm\Lambda \wh{\mbf H}}{ nT}
+\l(\frac{\wh{\mbf M}^x}{n}\r)^{-1}\bm{\mathcal H}^\prime\frac{\bm\Lambda^\prime\bm\Xi^\prime\bm F}{nT}\frac{\bm\Lambda^\prime\bm\Lambda \wh{\mbf H}}{n} + \l(\frac{\wh{\mbf M}^x}{n}\r)^{-1}\bm{\mathcal H}^\prime \frac{\bm\Lambda^\prime\bm\Xi^\prime\bm\Xi\bm\Lambda \wh{\mbf H}}{n^2T}\nn\\
&+\l(\frac{\wh{\mbf M}^x}{n}\r)^{-1}\frac{(\wh{\bm\Lambda}-\bm\Lambda\bm{\mathcal H})^\prime}{\sqrt n}\frac{\bm \Xi^\prime\bm F\bm\Lambda^\prime}{nT}\frac{\bm\Lambda \wh{\mbf H}}{\sqrt n}
+\l(\frac{\wh{\mbf M}^x}{n}\r)^{-1}\frac{(\wh{\bm\Lambda}-\bm\Lambda\bm{\mathcal H})^\prime}{\sqrt n}\frac{\bm\Lambda\bm F^\prime\bm\Xi}{nT}\frac{\bm\Lambda \wh{\mbf H}}{\sqrt n}\nn\\
&+\l(\frac{\wh{\mbf M}^x}{n}\r)^{-1}\frac{(\wh{\bm\Lambda}-\bm\Lambda\bm{\mathcal H})^\prime}{\sqrt n}\frac{\bm\Xi^\prime\bm \Xi}{nT}\frac{\bm\Lambda \wh{\mbf H}}{\sqrt n}
= \bm I_a+\bm I_b+\bm I_c+\bm I_d+\bm I_e+\bm I_f, \;\text{say.} \label{eq:2a3}
\end{align}
Then, because of \eqref{eq:1a2} and \eqref{eq:1a3} in the proof of Proposition \ref{prop:load}(a),
\begin{align}
\E\l[\l\Vert \frac{\bm F^\prime\bm\Xi\bm\Lambda}{ nT}\r\Vert^2\r] =\frac 1{n^2T^2}\sum_{k=1}^r\sum_{h=1}^r \E\l[\l(\sum_{t=1}^T F_{kt}\sum_{j=1}^n \xi_{jt}\lambda_{jh}\r)^2\r] \le \frac{r^2M_FM_\Lambda^2M_\xi}{nT}.\label{eq:2a3bis}
\end{align}
Therefore, by Assumption \ref{ass:common}(a), Lemma \eqref{lem:MO1}(iv),  \ref{lem:HO1bis}(i), \ref{lem:HO1}(i), 
and using \eqref{eq:2a3bis} , we get
\begin{align}
\Vert \bm I_a\Vert \le \l\Vert \l(\frac{\wh{\mbf M}^x}{n}\r)^{-1}\r\Vert\,\l\Vert \bm{\mathcal H}\r\Vert\,\l\Vert\frac{\bm\Lambda^\prime\bm\Lambda}{n}\r\Vert\,
\l\Vert\frac{\bm F^\prime\bm\Xi\bm\Lambda }{ nT}\r\Vert\,\l\Vert \wh{\mbf H}\r\Vert=O_{\mathrm P}\l(\frac 1{\sqrt {nT}}\r),\label{eq:2aIa}\\
\Vert \bm I_b\Vert\le  \l\Vert \l(\frac{\wh{\mbf M}^x}{n}\r)^{-1}\r\Vert\,\l\Vert \bm{\mathcal H}\r\Vert\,\l\Vert \frac{\bm\Lambda^\prime\bm\Xi^\prime\bm F}{nT}\r\Vert \,\l\Vert \frac{\bm\Lambda^\prime\bm\Lambda}n \r\Vert\,\l\Vert \wh{\mbf H}\r\Vert 
=O_{\mathrm P}\l(\frac 1{\sqrt {nT}}\r).\label{eq:2aIb}
\end{align}
Moreover, because of Lemma \ref{lem:FTLN}(i), \ref{lem:LLN}(iii), \ref{lem:MO1}(iv), \ref{lem:HO1}(i), and \ref{lem:HO1bis}(i),
\begin{align}\label{eq:2aIc}
\Vert \bm I_c\Vert \le  \l\Vert\l(\frac{\wh{\mbf M}^x}{n}\r)^{-1}\r\Vert\,\l\Vert\bm{\mathcal H}\r\Vert\, 
\l\Vert\frac{\bm\Lambda}{\sqrt n}\r\Vert\,
\l\Vert\frac{\bm\Lambda^\prime\bm\Xi^\prime\bm\Xi }{n^{3/2}T}\r\Vert\,
\l\Vert\wh{\mbf H}\r\Vert=O_{\mathrm P}\l(\max\l(\frac 1n,\frac 1{\sqrt {nT}}\r)\r).
\end{align}
Similarly, because of Proposition \ref{prop:L}(a), Lemma \ref{lem:FTLN}(i),  \ref{lem:LLN}(i),  \ref{lem:LLN}(iv), \ref{lem:MO1}(iv), and \ref{lem:HO1}(i), 
\begin{align}
\Vert \bm I_d\Vert &\le \l\Vert\l(\frac{\wh{\mbf M}^x}{n}\r)^{-1}\r\Vert\,\l\Vert \frac{\wh{\bm\Lambda}-\bm\Lambda\bm{\mathcal H}}{\sqrt n}\r\Vert \,\l\Vert \frac{\bm \Xi^\prime\bm F\bm}{\sqrt nT}\r\Vert\, \l\Vert \frac{\bm\Lambda^\prime \bm\Lambda }{n}
\r\Vert\, \l\Vert \wh{\mbf H}\r\Vert=  O_{\mathrm P}\l(\max\l(\frac{1}{n\sqrt {T}},\frac 1{T}\r)\r), \label{eq:2aldd}\\ 
\Vert \bm I_e\Vert &\le \l\Vert\l(\frac{\wh{\mbf M}^x}{n}\r)^{-1}\r\Vert\,\l\Vert \frac{\wh{\bm\Lambda}-\bm\Lambda\bm{\mathcal H}}{\sqrt n}\r\Vert \,\l\Vert \frac{\bm \Xi^\prime\bm F\bm}{\sqrt nT}\r\Vert\, \l\Vert \frac{\bm\Lambda  }{\sqrt n}
\r\Vert^2\, \l\Vert \wh{\mbf H}\r\Vert=  O_{\mathrm P}\l(\max\l(\frac{1}{n\sqrt {T}},\frac 1{T}\r)\r),\label{eq:2aIde}
\end{align}
and, because of Proposition \ref{prop:L}(a), Lemma \ref{lem:FTLN}(i), \ref{lem:LLN}(iii), \ref{lem:MO1}(iv), and \ref{lem:HO1}(i),
\begin{align}\label{eq:2aIf}
\Vert \bm I_f\Vert \le \l\Vert\l(\frac{\wh{\mbf M}^x}{n}\r)^{-1}\r\Vert\,\l\Vert \frac{\wh{\bm\Lambda}-\bm\Lambda\bm{\mathcal H}}{\sqrt n}\r\Vert \,\l\Vert \frac{\bm \Xi^\prime \bm\Xi}{nT}\r\Vert\, \l\Vert \frac{\bm\Lambda  }{\sqrt n}
\r\Vert\, \l\Vert \wh{\mbf H}\r\Vert
=  O_{\mathrm P}\l(\max\l(\frac 1{n^2},\frac{1}{n\sqrt {T}},\frac 1{T}\r)\r).
\end{align}
By using \eqref{eq:2aIa}, \eqref{eq:2aIb}, \eqref{eq:2aIc}, \eqref{eq:2aldd}, \eqref{eq:2aIde}, and \eqref{eq:2aIf} into \eqref{eq:2a3}
\beq\label{eq:2a4}
\Vert\bm{I}\Vert = O_{\mathrm P}\l(\max\l(\frac 1 {n},\frac 1{\sqrt{nT}},\frac 1T\r)\r).
\eeq
Second, consider $\bm {II}$ in \eqref{eq:2a}. From Proposition \ref{prop:L2}(b),
\beq\label{eq:2a5}
\Vert\bm{II}\Vert  \le \frac 1 {n}\l\Vert \wh{\bm\Lambda}-{\bm\Lambda}\wh{\mbf H}\r\Vert^2 
= O_{\mathrm P}\l(\max\l(\frac{1}{n^2},\frac 1{T}\r)\r).
\eeq
And, by using \eqref{eq:2a4} and \eqref{eq:2a5} in \eqref{eq:2a}, because of Lemma \ref{lem:FTLN}(ii), and Lemma \ref{lem:HO1}(i), we complete the proof. 
\end{proof}
\begin{prop}\label{prop:H}
Under Assumptions \ref{ass:common} through \ref{ass:ind}, and Assumption \ref{ass:ident}, if $\sqrt T/n\to0$ and $\sqrt n/T\to 0$, as $n,T\to\infty$,
$\min(\sqrt{n},\sqrt T)\Vert \wh{\mbf H}-\bm J\Vert = o_{\mathrm P}(1)$;
where $\bm J$ is a diagonal $r\times r$ matrix with entries $\pm 1$.
\end{prop}

\begin{proof}
From Proposition \ref{prop:factor} and by using \eqref{eq:estL}
\beq
\frac{\wh{\bm\Lambda}^\prime\bm\Lambda\wh{\mbf H}}{n}=\frac{\wh{\bm\Lambda}^\prime(\bm\Lambda\wh{\mbf H}-\wh{\bm\Lambda}+\wh{\bm\Lambda})}{n}=\frac{\wh{\mbf M}^x}{n}+O_{\mathrm P}\l(\max\l(\frac 1n,\frac 1{\sqrt {nT}}, \frac 1T\r)\r).
\eeq
Or, equivalently,
\beq\label{eq:HinvLL}
\l(\frac{\wh{\mbf M}^x}{n}\r)^{-1}\frac{\wh{\bm\Lambda}^\prime\bm\Lambda}{n}=\wh{\mbf H}^{-1}+O_{\mathrm P}\l(\max\l(\frac 1n,\frac 1{\sqrt {nT}}, \frac 1T\r)\r).
\eeq
Moreover, by Assumption \ref{ass:ident} and using \eqref{eq:HinvLL} in  \eqref{eq:acca}  we have
\beq\label{eq:Horth}
\wh{\mbf H}^\prime = \l(\frac{\wh{\mbf M}^x}{n}\r)^{-1}\frac{\wh{\bm\Lambda}^\prime\bm\Lambda}{n}=\wh{\mbf H}^{-1}+O_{\mathrm P}\l(\max\l(\frac 1n,\frac 1{\sqrt {nT}}, \frac 1T\r)\r).
\eeq
Therefore, because of \eqref{eq:Horth}, as $n,T\to\infty$, $\wh{\mbf H}$ is an $r\times r$ orthogonal matrix thus it has eigenvalues $\pm 1$.

Moreover, because of  Proposition \ref{prop:factor}
\beq
\frac{\wh{\bm\Lambda}^\prime\bm\Lambda}{n}= \frac{(\wh{\bm\Lambda}-\bm\Lambda\wh{\mbf H}+\bm\Lambda\wh{\mbf H})^\prime\bm\Lambda}{n} =\frac{\wh{\mbf H}^\prime\bm\Lambda^\prime\bm\Lambda}{n}+ O_{\mathrm P}\l(\max\l(\frac 1n,\frac 1{\sqrt {nT}}, \frac 1T\r)\r).\label{eq:Horth2}
\eeq
Thus, from \eqref{eq:Horth} and \eqref{eq:Horth2} and by Lemma \ref{lem:MO1}(iv),
\beq
\wh{\mbf H}^\prime = \l(\frac{\wh{\mbf M}^x}{n}\r)^{-1}\frac{\wh{\bm\Lambda}^\prime\bm\Lambda}{n}=\l(\frac{\wh{\mbf M}^x}{n}\r)^{-1}\frac{\wh{\mbf H}^\prime\bm\Lambda^\prime\bm\Lambda}{n}+ O_{\mathrm P}\l(\max\l(\frac 1n,\frac 1{\sqrt {nT}}, \frac 1T\r)\r).\label{eq:Horth3}
\eeq
And from \eqref{eq:Horth3} it follows that
\beq
\l(\frac{\wh{\mbf M}^x}{n}\r)\wh{\mbf H}^\prime  = \wh{\mbf H}^\prime\frac{\bm\Lambda^\prime\bm\Lambda}{n}+ O_{\mathrm P}\l(\max\l(\frac 1n,\frac 1{\sqrt {nT}}, \frac 1T\r)\r).\label{eq:Horth4}
\eeq
So, because of \eqref{eq:Horth4}, as $n,T\to\infty$, the columns of  $\wh{\mbf H}$ are the eigenvectors of $\frac{\bm\Lambda^\prime\bm\Lambda}{n}$ with eigenvalues $\frac{\wh{\mbf M}^x}{n}$. The eigenvectors are normalized since $\wh{\mbf H}$ is orthogonal, as $n,T\to\infty$. Moreover, under Assumption \ref{ass:ident}, $\frac{\bm\Lambda^\prime\bm\Lambda}{n}$  is diagonal, so, as $n,T\to\infty$, $\wh{\mbf H}$ must be diagonal with eigenvalues $\pm 1$. By noticing that, if $\sqrt T/n\to0$ and $\sqrt n/T\to 0$, then $\max\l(\frac 1n,\frac 1{\sqrt {nT}}, \frac 1T\r)=o_{\mathrm P}\l(\max\l(\frac 1{\sqrt {n}}, \frac 1{\sqrt T}\r)\r)$, we complete the proof. 
\end{proof}
%

\setcounter{equation}{0}

\section{Auxiliary lemmata}\label{app:C}
\begin{lem}\label{lem:Gxi}
Under Assumptions \ref{ass:common} and \ref{ass:idio}:
\begin{compactenum} 
\item [(i)] for all $n\in\mathbb N$ and $T\in\mathbb N$, $\frac 1{nT}\sum_{i,j=1}^n\sum_{t,s=1}^T \vert\E_{}[\xi_{it}\xi_{js}]\vert \le M_{1\xi}$, for some finite positive real $M_{1\xi}$ independent of $n$ and $T$;
\item [(ii)] for all $n\in\mathbb N$ and $t\in\mathbb Z$, $\frac 1{n}\sum_{i,j=1}^n \vert\E_{}[\xi_{it}\xi_{jt}]\vert \le M_{2\xi}$, for some finite positive real $M_{2\xi}$ independent of $n$ and $t$;
\item [(iii)] for all $i\in\mathbb N$ and $T\in\mathbb N$, $\!\frac 1{T}\sum_{t,s=1}^T \vert\E_{}[\xi_{it}\xi_{is}]\vert \le M_{3\xi}$, $\!$for some finite positive real $M_{3\xi}$ independent of $i$ and $T$;
\item [(iv)] 
for all $j=1,\ldots,r$, $\underline C_j\!\le \lim\inf_{n\to\infty} \frac{ \mu_{j}^\chi}n\le\lim\sup_{n\to\infty} \frac{\mu_{j}^\chi}n\le\! \overline C_j$, $\!\!$
for some finite positive reals $\underline C_j$ and $\overline C_j$;
\item [(v)] for all $n\in\mathbb N$, $\mu_{1}^\xi  \le M_{2\xi}$, where $M_\xi$ is defined in part (ii);
\item [(vi)] for all $j=1,\ldots,r$, $\underline C_j\le \lim\inf_{n\to\infty} \frac{ \mu_{j}^x}n\le\lim\sup_{n\to\infty} \frac{\mu_{j}^x}n\le \overline C_j$, and for all $n\in\mathbb N$, $\mu_{r+1}^x \le M_\xi$, where $M_\xi$ is defined in Assumption \ref{ass:idio}(b).
\end{compactenum}
\end{lem}

\begin{proof} Using Assumptions \ref{ass:idio}(a) and \ref{ass:idio}(b), we have:
\begin{align}
\frac 1{nT}\sum_{i,j=1}^n\sum_{t,s=1}^T \vert\E_{}[\xi_{it}\xi_{js}]\vert &=\frac 1{n}\sum_{i,j=1}^n\sum_{k=-(T-1)}^{T-1} \l(1-\frac{\vert k\vert}{T}\r) \vert\E_{}[\xi_{it}\xi_{j,t-k}]\vert\nn\\
&\le \frac 1n\sum_{i=1}^n \sum_{k=-\infty}^{\infty} \rho^{|k|}\sigma_i^2 +  \max_{j=1,\ldots, n}\sum_{j=1, j\ne i}^n\sum_{k=-\infty}^{\infty} \rho^{\vert k\vert} M_{ij}\nn\\
&\le \frac{C_\xi^\prime(1+\rho)}{1-\rho}+\frac{M_\xi(1+\rho)}{1-\rho}. \nn
\end{align}
Similarly, 
\begin{align}
\frac 1{n}\sum_{i,j=1}^n \vert\E_{}[\xi_{it}\xi_{jt}]\vert &\le \frac 1n\sum_{i=1}^n \sigma_i^2 + \max_{j=1,\ldots,n}\sum_{j=1, j\ne i}^n M_{ij}\le C_\xi^\prime+ M_\xi,\nn
\end{align}
and 
\begin{align}
\frac 1{T}\sum_{t,s=1}^T \vert\E_{}[\xi_{it}\xi_{is}]\vert &= \sum_{k=-(T-1)}^{T-1} \l(1-\frac{\vert k\vert}{T}\r) \vert\E_{}[\xi_{it}\xi_{i,t-k}]\vert\le \sum_{k=-\infty}^{\infty} \rho^{\vert k\vert} \sigma_i^2 \le  \frac{C_\xi^\prime (1+\rho)}{1-\rho}.\nn
\end{align}
Defining, $M_{1\xi}=\frac{(C_\xi^\prime+M_\xi)(1+\rho)}{1-\rho}$ and $M_{2\xi}=C_\xi^\prime+M_\xi$, and $M_{3\xi}=\frac{C_\xi^\prime(1+\rho)}{1-\rho}$, we prove parts (i), (ii), and (iii).\smallskip

For part (iv), by \citet[Theorem 7]{MK04}, for all $j=1,\ldots, r$, we have
\beq\label{eq:KUMAR}
\frac{\mu_r(\bm\Lambda^\prime\bm\Lambda)}n\mu_j(\bm\Gamma^F) \le \frac{\mu_{j}^\chi}n \le \frac{\mu_j(\bm\Lambda^\prime\bm\Lambda)}n\ \mu_1(\bm\Gamma^F).
\eeq
The proof then follows from Assumption \ref{ass:common}(a) which, by continuity of eigenvalues, implies that, for any $j=1,\ldots, r$, as $n\to\infty$
\[
\lim_{n\to\infty}\frac{\mu_j(\bm\Lambda^\prime\bm\Lambda)}n = \mu_j(\bm\Sigma_\Lambda).
\]
with 
$$
0<m_\Lambda^2\le \mu_r(\bm\Sigma_\Lambda)\le\mu_1(\bm\Sigma_\Lambda)\le M_\Lambda^2<\infty ,$$ 
and by Assumption \ref{ass:common}(b) and Assumption \ref{ass:common}(c) which imply 
$$
0<m_F\le \mu_r(\bm\Gamma^F)\le\mu_1(\bm\Gamma^F)\le M_F<\infty.
$$ 

For part (v), by Assumption \ref{ass:idio}(b):
\begin{align}
\Vert\bm\Gamma^\xi\Vert\le \max_{i=1,\ldots,n}\sum_{j=1}^n \vert \E[\xi_{it}\xi_{jt}]\vert \le 
\max_{i=1,\ldots, n}\sigma_i^2+ \max_{i=1,\ldots,n}\sum_{j=1,j\ne i}^n M_{ij}\le C_\xi^\prime + M_{\xi}.\nn
\end{align}
Part (vi) follows from parts (iv) and (v) and Weyl's inequality. This completes the proof. 
\end{proof}
\begin{lem}\label{lem:FTLN}
Under Assumptions \ref{ass:common} through \ref{ass:ind},  for all $t=1,\ldots, T$ and all $n,T\in\mathbb N$
\begin{compactenum}[(i)]
\item $\Vert\frac{\bm\Lambda}{\sqrt n}\Vert=O(1)$;
\item $\Vert \mbf F_t\Vert = O_{\mathrm {ms}}(1)$ and $\Vert\frac{\bm F}{\sqrt T}\Vert=O_{\mathrm {ms}}(1)$.
\end{compactenum}
\end{lem}

\begin{proof}
By Assumption \ref{ass:common}(a), which holds for all $n\in\mathbb N$,
\beq
\sup_{n\in\mathbb N}\l\Vert\frac{\bm\Lambda}{\sqrt n}\r\Vert^2\le \sup_{n\in\mathbb N}\l\Vert\frac{\bm\Lambda}{\sqrt n}\r\Vert^2_F =\sup_{n\in\mathbb N}\frac 1n\sum_{j=1}^r\sum_{i=1}^n \lambda_{ij}^2 \le\sup_{n\in\mathbb N} \max_{i=1,\ldots, n} \Vert\bm\lambda_i\Vert^2\le M_\Lambda^2,\nn
\eeq
since $M_\Lambda$ is independent of $i$. This proves part (i).\smallskip

By Assumption \ref{ass:common}(b), which holds for all $T\in\mathbb N$ because of stationarity (see also part (i) of Lemma \ref{lem:covarianze}),
\begin{align}\label{eq:vettoreF}
\sup_{T\in\mathbb N}\max_{t=1,\ldots ,T}\E[\Vert\mbf F_t\Vert^2]&=
\sup_{T\in\mathbb N} \max_{t=1,\ldots ,T}\sum_{j=1}^r \E[F_{jt}^2]\le
r\sup_{T\in\mathbb N}\max_{t=1,\ldots ,T}\max_{j=1,\ldots,r} \E[F_{jt}^2] \nn\\
&\le r\sup_{T\in\mathbb N}\max_{t=1,\ldots ,T}\max_{j=1,\ldots,r} \bm\eta_j^\prime\bm\Gamma^F\bm\eta_j\le r \Vert\bm\Gamma^F\Vert\le
r M_F,
\end{align}
since $M_F$ is independent of $t$ and where $\bm\eta_j$ is an $r$-dimensional vector with one in the $j$th entry and zero elsewhere. Therefore, from \eqref{eq:vettoreF}:
\begin{align}
\sup_{T\in\mathbb N}\E\l[\l\Vert\frac{\bm F}{\sqrt T}\r\Vert^2\r]\le\sup_{T\in\mathbb N} \E\l[\l\Vert\frac{\bm F}{\sqrt T}\r\Vert^2_F\r] =\sup_{T\in\mathbb N}\frac 1T\sum_{j=1}^r\sum_{t=1}^T \E[[\bm F]_{tj}^2] \le\sup_{T\in\mathbb N} \max_{t=1,\ldots, T} \E[\Vert\mbf F_t\Vert^2]
\le r M_F.\nn
\end{align}
This proves part (ii) and it completes the proof. 
\end{proof}
\begin{lem}\label{lem:LLN}
Under Assumptions \ref{ass:common} through \ref{ass:ind}, for all $n,T\in\mathbb N$ 
$\sqrt T\l\Vert\frac{\bm F^\prime\bm \Xi}{\sqrt n T} \r\Vert = O_{\mathrm{ms}}(1)$.
\end{lem}

\begin{proof}  By Assumption \ref{ass:ind}, Lemma \ref{lem:Gxi}(iii) and Cauchy-Schwarz inequality
\begin{align}
\E\l[ \l\Vert\frac {\bm F^\prime\bm \Xi}{\sqrt n T} \r\Vert^2\r] &=\E\l[ \l\Vert\frac 1{\sqrt nT}\sum_{t=1}^T \mbf F_t\bm\xi_t^\prime \r\Vert^2\r]\le \E\l[ \l\Vert\frac 1{\sqrt nT}\sum_{t=1}^T \mbf F_t\bm\xi_t^\prime \r\Vert_F^2\r]\nn\\
&=\frac 1{nT^2}\sum_{j=1}^r\sum_{i=1}^n \E\l[\l(\sum_{t=1}^T F_{jt}\xi_{it}\r)^2\r]\nn\\
&\le\frac r{T^2}\max_{j=1\ldots, r} \max_{i=1\ldots, n} \sum_{t=1}^T\sum_{s=1}^T \E[\xi_{it}F_{jt}\xi_{is}F_{js}]\nn\\
&=\frac r{T^2}\max_{j=1\ldots, r} \max_{i=1\ldots, n} \sum_{t=1}^T\sum_{s=1}^T \E[F_{jt}F_{js}]\, \E[\xi_{it}\xi_{is}]\nn\\
&\le \l\{\frac{r}{T}\max_{j=1,\ldots, r}\max_{t,s=1,\ldots, n} \vert\E[F_{jt}F_{js}]\vert\r\}
\l\{\frac 1T  \sum_{t=1}^T\sum_{s=1}^T\l\vert\E[\xi_{it}\xi_{is}]\r\vert\r\}\nn\\
&\le \frac r T \max_{j=1,\ldots, r}\max_{t,s=1,\ldots, n} \E[F_{jt}^2] \frac{M_\xi(1+\rho)}{1-\rho}\le\frac{rM_F M_{3\xi}}T,\nn
\end{align}
since $M_F$ is independent of $t$ and $s$ and $M_{3\xi}$ is independent of $i$. This completes the proof. 
\end{proof}

\begin{lem}\label{lem:aiuto}
Under Assumptions \ref{ass:common} through \ref{ass:ind}, for all $n,T\in\mathbb N$ 
\begin{compactenum}[(i)]
\item $\min(n,\sqrt T)\l\Vert\frac{\bm \Xi^\prime\bm \Xi}{nT} \r\Vert = O_{\mathrm{ms}}(1)$;
\item $\min(n,\sqrt {nT})\l\Vert\frac{\bm\Lambda^\prime\bm \Xi^\prime\bm \Xi}{n^{3/2}T} \r\Vert = O_{\mathrm{ms}}(1)$;
\item $\l\Vert\frac{\bm \Lambda^\prime\bm \Lambda}{n} \r\Vert = O(1)$.
\end{compactenum}
\end{lem}

\begin{proof}
For part (i), first notice that, by Lemmas \ref{lem:Gxi}(v) and \ref{lem:covarianzeF}(ii)
\beq
\l\Vert\frac{\bm \Xi^\prime\bm \Xi}{nT} \r\Vert \le \l\Vert\frac{\bm \Xi^\prime\bm \Xi}{nT}-\frac{\bm\Gamma^\xi}{n} \r\Vert +\l\Vert\frac{\bm\Gamma^\xi}{n}\r\Vert
= \l\Vert\frac{\bm \Xi^\prime\bm \Xi}{nT}-\frac{\bm\Gamma^\xi}{n} \r\Vert +\frac{\mu_1^\xi}{n}=O_{\mathrm P}\l(\frac 1{\sqrt T}\r)+O\l(\frac 1{n}\r).
\eeq
Similarly, for part (ii) by Lemmas \ref{lem:FTLN}(i) and \ref{lem:Gxi}(v) 
\beq
\l\Vert\frac{\bm\Lambda^\prime\bm \Xi^\prime\bm \Xi}{n^{3/2}T} \r\Vert \le \l\Vert\frac{\bm\Lambda^\prime\bm \Xi^\prime\bm \Xi}{n^{3/2}T}-\frac{\bm\Lambda^\prime\bm\Gamma^\xi}{n^{3/2}} \r\Vert +\l\Vert \frac{\bm\Lambda}{\sqrt n}\r\Vert\,\l\Vert\frac{\bm\Gamma^\xi}{n}\r\Vert= 
\l\Vert\frac{\bm\Lambda^\prime\bm \Xi^\prime\bm \Xi}{n^{3/2}T}-\frac{\bm\Lambda^\prime\bm\Gamma^\xi}{n^{3/2}} \r\Vert +O\l(\frac 1{n}\r).\label{eq:rifacciamo2}
\eeq
Then, because of Assumption \ref{ass:idio}(c), 
\begin{align}
\E\l[\l\Vert\frac{\bm\Lambda^\prime\bm \Xi^\prime\bm \Xi}{n^{3/2}T}-\frac{\bm\Lambda^\prime\bm\Gamma^\xi}{n^{3/2}} \r\Vert ^2\r]&\le 
\E\l[\l\Vert\frac{\bm\Lambda^\prime\bm \Xi^\prime\bm \Xi}{n^{3/2}T}-\frac{\bm\Lambda^\prime\bm\Gamma^\xi}{n^{3/2}} \r\Vert ^2_F\r]\nn\\
&= \sum_{k=1}^r\sum_{j=1}^n\E\l[\l\vert\frac 1{n^{3/2}T}
\sum_{i=1}^n  \sum_{t=1}^T \l\{\lambda_{ik}\xi_{it}\xi_{jt}-\lambda_{ik}\E[\xi_{it}\xi_{jt}]\r\}
\r\vert^2
\r]\nn\\
&\le \frac{rM_\Lambda n}{n^2T}
\max_{j=1,\ldots, n}
\E\l[\l\vert\frac 1{\sqrt{nT}}
\sum_{i=1}^n  \sum_{t=1}^T \l\{\xi_{it}\xi_{jt}-\E[\xi_{it}\xi_{jt}]\r\}
\r\vert^2
\r]\le \frac{rM_\Lambda K_\xi}{nT},\label{eq:rifacciamo}
\end{align}
since $K_\xi$ is independent of $j$. By using \eqref{eq:rifacciamo} into \eqref{eq:rifacciamo2},
we prove part (ii).\smallskip 

Part (iii) follows from Lemma \ref{lem:FTLN}(i) since
\[
\l\Vert\frac{\bm \Lambda^\prime\bm \Lambda}{n} \r\Vert\le \l\Vert\frac{\bm \Lambda}{\sqrt n} \r\Vert^2\le M_\Lambda^2.
\]
This completes the proof. 
\end{proof}


\begin{lem}\label{lem:covarianzeF}
$\,$
\begin{compactenum}[(i)]
\item Under Assumption \ref{ass:common}, for all $T\in\mathbb N$, $\sqrt T\l\Vert\frac{\bm F^\prime\bm F}{T} -\bm\Gamma^F\r\Vert = O_{\mathrm{ms}}(1)$;
\item Under Assumption \ref{ass:idio}, for all $n,T\in\mathbb N$, $\sqrt T \l\Vert\frac{\bm \Xi^\prime\bm \Xi}{nT}-\frac{\bm\Gamma^\xi}{n} \r\Vert= O_{\mathrm{ms}}(1)$.
\end{compactenum}
\end{lem}

\begin{proof}
Part (i) is direct consequence of Assumption \ref{ass:common}(c-ii), since
\begin{align}
\E\l[\l\Vert\frac{\bm F^\prime\bm F}{T}-\bm\Gamma^F\r\Vert^2\r]&=
\E\l[\l\Vert \frac 1T\sum_{t=1}^T \mbf F_t\mbf F_t^\prime-\bm\Gamma^F\r\Vert^2\r]\le 
\E\l[\l\Vert \frac 1T\sum_{t=1}^T \l\{\mbf F_t\mbf F_t^\prime-\bm\Gamma^F\r\}\r\Vert^2_F\r]\nn\\
&=\frac 1{T^2} \sum_{j=1}^r\sum_{k=1}^r\E\l[\l(\sum_{t=1}^T \l\{F_{jt}F_{kt}-\E[F_{jt}F_{kt}] \r\}\r)^2\r]
\le \frac {r^2C_F}{T},\nn
\end{align}
since $C_F$ is independent of $j$ and $k$. \smallskip

For part (ii),  by Assumption \ref{ass:idio}(c), letting $\gamma_{ij}^\xi$ be the $(i,j)$the entry of $\bm\Gamma^\xi$, we have
\begin{align}
\E\l[\l\Vert\frac{\bm \Xi^\prime\bm \Xi}{nT}-\frac{\bm\Gamma^\xi}{n} \r\Vert ^2\r]&=\E\l[\l\Vert \frac 1{nT} \sum_{t=1}^T\bm\xi_t\bm\xi_t^\prime -\frac{\bm\Gamma^\xi}n \r\Vert^2\r]\le  \E\l[\l\Vert \frac 1{nT} \sum_{t=1}^T\l\{\bm\xi_t\bm\xi_t^\prime -\bm\Gamma^\xi\r\} \r\Vert^2_F\r]\nn\\
&=\frac 1{n^2T^2}\sum_{i,j=1}^n\E\l[\l(\sum_{t=1}^T \l\{\xi_{it}\xi_{jt}-\gamma_{ij}^\xi\r\}\r)^2\r]\nn\\
&\le \max_{i,j=1,\ldots, n}\frac 1{T^2} \E\l[\l(\sum_{t=1}^T \l\{\xi_{it}\xi_{jt}-\gamma_{ij}^\xi\r\}\r)^2\r]
\le \frac {K_\xi}{T},\nn
\end{align}
since $K_\xi$ is independent of $i$ and $j$. This completes the proof. 
\end{proof}

\begin{lem}\label{lem:covarianze}
Under Assumptions \ref{ass:common} through \ref{ass:ind}, 
for all $n,T\in\mathbb N$ 
\begin{compactenum}
\item [(i)]  $\frac {\sqrt T} n\Vert \wh{\bm\Gamma}^x-\bm\Gamma^x\Vert = O_{\mathrm {ms}}(1)$;
\item [(ii)] $\frac {\min(n,\sqrt T)} n\Vert \wh{\bm\Gamma}^x-\bm\Gamma^\chi\Vert = O_{\mathrm {ms}}(1)$;
\item [(iii)] $\frac  {\min(n,\sqrt T)} n \Vert \wh{\mbf M}^x-\mbf M^\chi\Vert=O_{\mathrm {ms}}(1)$;
\item [(iv)] as $n\to\infty$, ${\min(n,\sqrt T)} \Vert\wh{\mbf V}^x-\mbf V^\chi\mbf J\Vert = O_{\mathrm {ms}}(1)$, where $\mbf J$ is $r\times r$ diagonal with entries $\pm 1$.
\end{compactenum}
\end{lem}

\begin{proof}
For part (i),  from Lemma \ref{lem:covarianzeF}(i), \ref{lem:covarianzeF}(ii) and \ref{lem:FTLN}(i): 
\begin{align}
\E\l[\l\Vert\frac  1n\l( \wh{\bm\Gamma}^x-\bm\Gamma^x\r)\r\Vert^2\r]
& =
 \E\l[\l\Vert \frac 1n\l\{ \bm\Lambda\l( \frac 1T\sum_{t=1}^T \mbf F_t\mbf F_t^\prime-\bm\Gamma^F\r)\bm\Lambda^\prime 
 + \frac 1T\sum_{t=1}^T\bm\xi_t\bm\xi_t^\prime-\bm\Gamma^\xi\r\}\r\Vert^2
 \r]\nn\\
& \le 
\E\l[\l\Vert \frac 1n\l\{ \bm\Lambda\l( \frac 1T\sum_{t=1}^T \mbf F_t\mbf F_t^\prime-\bm\Gamma^F\r)\bm\Lambda^\prime 
\r\}\r\Vert^2\r]
 +\E\l[\l\Vert \frac 1n\l( \frac 1T\sum_{t=1}^T\bm\xi_t\bm\xi_t^\prime-\bm\Gamma^\xi\r)\r\Vert^2
 \r]\nn\\
&\le \l\Vert\frac{\bm\Lambda}{\sqrt n}\r\Vert^4 \, \E\l[\l\Vert \frac 1T\sum_{t=1}^T \mbf F_t\mbf F_t^\prime-\bm\Gamma^F\r\Vert^2\r]+\E\l[\l\Vert \frac 1n\l( \frac 1T\sum_{t=1}^T\bm\xi_t\bm\xi_t^\prime-\bm\Gamma^\xi\r)\r\Vert^2
 \r]\nn\\
&\le \frac{M_\Lambda^4 r^2 C_F+K_\xi}{T}.\nn
\end{align}
This proves part (i).\smallskip

Part (ii), follows from
\begin{align}
\E\l[\l\Vert\frac  1n\l( \wh{\bm\Gamma}^x-\bm\Gamma^\chi\r)\r\Vert^2\r] &\le \E\l[\l\Vert\frac  1n\l( \wh{\bm\Gamma}^x-\bm\Gamma^x\r)\r\Vert^2\r] +\l\Vert\frac  {\bm\Gamma^\xi}n\r\Vert^2 \le  \frac{M_\Lambda^4K_F+K_\xi}{T}+
\frac{\mu_1^\xi}{n^2}\le \frac{M_\Lambda^4K_F+K_\xi}{T}+\frac{M_{2\xi}}{n^2}\nn\\
&\le M_1\max \l(\frac 1{n^2},\frac 1T\r),\;\text{say,}\nn
\end{align}
 because of part (i) and Lemma \ref{lem:Gxi}(v) and where $M_1$ is a finite positive real independent of $n$ and $T$.\smallskip
 
For part (iii), for any $j=1,\ldots, r$, because of Weyl's inequality and part (ii), it holds that
\beq\label{eq:weyl}
\vert\wh\mu_j^x-\mu_j^\chi \vert\le \mu_1(\wh{\bm\Gamma}^x-\bm\Gamma^\chi)
=\l\Vert\wh{\bm\Gamma}^x-\bm\Gamma^\chi\r\Vert.
\eeq
Hence, from \eqref{eq:weyl} and part (ii),
\begin{align}
\E\l[ \l\Vert\frac 1n\l( \wh{\mbf M}^x-\mbf M^\chi\r)\r\Vert^2\r] &\le  
\E\l[ \l\Vert\frac 1n\l( \wh{\mbf M}^x-\mbf M^\chi\r)\r\Vert^2_F\r] =
\frac 1{n^2} \sum_{j=1}^r \E\l[\l(\wh\mu_j^x-\mu_j^\chi \r)^2\r]\nn\\
&\le \frac{r}{n^2} \E\l[\l(\wh\mu_1^x-\mu_1^\chi \r)^2\r]\le r\E\l[\l\Vert\frac 1n\l(\wh{\bm\Gamma}^x-\bm\Gamma^\chi\r)\r\Vert^2\r]\le \frac{rM_\Lambda^4K_F+K_\xi}{T}+\frac{rM_\xi}{n^2}\nn\\
&\le rM_1\max\l(\frac 1{n^2},\frac 1T\r).
\end{align}
This proves part (iii).\smallskip

For part (iv), because of  \citet[Theorem 2]{yu15}, which is a special case of Davis Kahn Theorem,  there exists an $r\times r$ diagonal matrix $\mbf J$  with entries $\pm 1$ such that
\beq
\Vert\wh{\mbf V}^x-\mbf V^\chi\mbf J\Vert\le \frac{2^{3/2}\sqrt r\l\Vert \wh{\bm\Gamma}^x-\bm\Gamma^\chi\r\Vert }{\min(\vert \mu_{0}^\chi-\mu_{1}^\chi \vert,\vert \mu_{r}^\chi-\mu_{r+1}^\chi \vert)},\label{eq:DK}
\eeq
where $\mu_0^\chi=\infty$. This holds provided the eigenvalues $\mu_j^\chi$ are distinct as required by Assumption \ref{ass:eval}.
Therefore, from \eqref{eq:DK},  part (ii) and Lemma \ref{lem:Gxi}(iv), and since $\mu_{r+1}^\chi =0$, as $n\to\infty$
\begin{align}
\min(n^2,T) \E\l[\Vert\wh{\mbf V}^x-\mbf V^\chi\mbf J\Vert^2\r]&\le \frac
{\min(n^2,T)  2^3 \frac r{n^2} \E\l[\l\Vert\wh{\bm\Gamma}^x-\bm\Gamma^\chi \r\Vert^2\r]}
{\frac 1{n^2}\l\{\min(\vert \mu_{0}^\chi-\mu_{1}^\chi \vert,\vert \mu_{r}^\chi-\mu_{r+1}^\chi \vert)\r\}^2} \le \frac{8rM_1}{\underline C_r^2}= M_2,\;\text{say,}\nn
\end{align}
where $M_2$ is a finite positive real independent of $n$ and $T$.
This proves part (iv). This completes the proof. 
\end{proof}

\begin{lem}\label{lem:covarianzerighe}
Let $\bm\varepsilon_i$ be the $n$-dimensional vector with one in entry $i$ and zero elsewhere. Under Assumptions \ref{ass:common} through \ref{ass:ind}, for all $i=1,\ldots, n$ and $T\in\mathbb N$ and as $n\to\infty$ 
\begin{compactenum}[(i)]
\item $\frac {\min(\sqrt n,\sqrt T)}{\sqrt n} \Vert\bm\varepsilon^\prime_i (\wh{\bm\Gamma}^x-\bm\Gamma^\chi)\Vert = O_{\mathrm {ms}}(1)$;
\item $\sqrt n\Vert\mbf v_i^\chi\Vert = O(1)$;
\item ${\min(\sqrt n,\sqrt T)} \sqrt n \Vert\wh{\mbf v}^{x\prime}_i-\mbf v^{\chi\prime}\mbf J\Vert = O_{\mathrm {ms}}(1)$.
\end{compactenum}
\end{lem}

\begin{proof}
 First notice that
\begin{align}
\max_{i=1,\ldots, n}\E\l[\l\Vert \frac 1{\sqrt n}\bm\varepsilon^\prime_i\l( \sum_{t=1}^T\bm\xi_t\bm\xi_t^\prime -\bm\Gamma^\xi\r) \r\Vert^2\r]
&\le\max_{i=1,\ldots, n}\frac 1{n}\sum_{j=1}^n\E\l[\l(\frac 1T\sum_{t=1}^T \xi_{it}\xi_{jt}-[\bm\Gamma^\xi]_{ij}\r)^2\r]\nn\\
&\le \max_{i,j=1,\ldots, n} \E\l[\l(\frac 1T\sum_{t=1}^T \xi_{it}\xi_{jt}-[\bm\Gamma^\xi]_{ij}\r)^2\r]
\le \frac {K_\xi}{T}\label{eq:Gammaxiappriga},
\end{align}
since $K_\xi$ is independent of $i$ and $j$. Therefore, from Lemma \ref{lem:covarianzeF}(i) and Lemma \ref{lem:FTLN}(i), and using \eqref{eq:Gammaxiappriga}, 
\begin{align}
\max_{i=1,\ldots, n} \E\l[\l\Vert\frac  1{\sqrt n}\bm\varepsilon^\prime_i\l( \wh{\bm\Gamma}^x-\bm\Gamma^x\r)\r\Vert^2\r]
=&\,
\max_{i=1,\ldots, n} \E\l[\l\Vert \frac 1{\sqrt n}\l\{ \bm\lambda_i^\prime \l( \frac 1T\sum_{t=1}^T \mbf F_t\mbf F_t^\prime-\bm\Gamma^F\r)\bm\Lambda^\prime 
 +\bm\varepsilon^\prime_i\l( \frac 1T\sum_{t=1}^T\bm\xi_t\bm\xi_t^\prime-\bm\Gamma^\xi\r)\r\}\r\Vert^2
 \r]\nn\\
\le&\,\max_{i=1,\ldots, n}\Vert \bm\lambda_i\Vert^2\, \l\Vert\frac{\bm\Lambda}{\sqrt n}\r\Vert^2  \, \E\l[\l\Vert \frac 1T\sum_{t=1}^T \mbf F_t\mbf F_t^\prime-\bm\Gamma^F\r\Vert^2\r]\nn\\
&+\max_{i=1,\ldots, n}\E\l[\l\Vert \frac 1{\sqrt n}\bm\varepsilon^\prime_i\l( \frac 1T\sum_{t=1}^T\bm\xi_t\bm\xi_t^\prime-\bm\Gamma^\xi\r)\r\Vert^2
 \r]\nn\\
\le&\, \frac{M_\Lambda^4 r^2 C_F+K_\xi}{T},\label{eq:rigacov}
\end{align}
since $M_\Lambda$ and $C_F$ are independent of $i$. Then, following the same arguments as Lemma \ref{lem:covarianze}(ii), because of \eqref{eq:rigacov} and Lemma \ref{lem:Gxi}(v):
\begin{align}
\max_{i=1,\ldots, n}\E\l[\l\Vert\frac  1{\sqrt n}\bm\varepsilon^\prime_i\l( \wh{\bm\Gamma}^x-\bm\Gamma^\chi\r)\r\Vert^2\r]& \le\max_{i=1,\ldots, n} \E\l[\l\Vert\frac  1{\sqrt n}\bm\varepsilon^\prime_i\l( \wh{\bm\Gamma}^x-\bm\Gamma^x\r)\r\Vert^2\r] +\max_{i=1,\ldots, n}\l\Vert\bm\varepsilon^\prime_i\frac  {\bm\Gamma^\xi}{\sqrt n}\r\Vert^2\nn\\
 &\le  \frac{M_\Lambda^4r^2C_F+K_\xi}{T}+\max_{i=1,\ldots, n}  \Vert \bm\varepsilon_i\Vert^2\, \l\Vert\frac{\bm\Gamma^\xi}{\sqrt n}\r\Vert^2\nn\\
 &=\frac{M_\Lambda^4r^2C_F+K_\xi}{T}+ \frac{\mu_1^\xi}{n} \le \frac{M_\Lambda^4r^2C_F+K_\xi}{T}+ \frac{M_{2\xi}}{n}\le M_1\max\l(\frac 1T,\frac 1n\r),\; \text{say,}\nn
 \end{align}
 since $\Vert \bm\varepsilon_i\Vert=1$ and where $M_1$ is a finite positive real independent of $n$ and $T$ defined in Lemma \ref{lem:covarianze}(ii). This proves part (i).\smallskip

For part (ii), notice that for all $i=1,\ldots, n$ we must have:
\beq\label{eqfinitevarchi}
\Var(\chi_{it}) =\bm\lambda_i^\prime\bm\Gamma^F\bm\lambda_i\le \Vert\bm\lambda_i\Vert^2\,\l\Vert\bm\Gamma^F\r\Vert
\le M_\Lambda^2 M_F,
\eeq
which is finite for all $i$ and $t$. So, since by Lemma \ref{lem:Gxi}(iv)
\[
\liminf_{n\to\infty}\max_{i=1,\ldots, n}\Var(\chi_{it})=\liminf_{n\to\infty}\max_{i=1,\ldots, n}\sum_{j=1}^r \mu_j^\chi [\mbf V^\chi]_{ij}^2\ge 
\liminf_{n\to\infty}\mu_r^\chi \max_{i=1,\ldots, n}\sum_{j=1}^r  [\mbf V^\chi]_{ij}^2 \ge \underline C_r n \max_{i=1,\ldots, n}\Vert\mbf v_i^\chi\Vert^2,
\]
then, because of \eqref{eqfinitevarchi}, we must have, as $n\to\infty$, 
\[
\underline C_r n \max_{i=1,\ldots, n}\Vert\mbf v_i^\chi\Vert^2\le  M_\Lambda^2 M_F
\]
which implies that, as $n\to\infty$,  
\[
n \max_{i=1,\ldots, n}\Vert\mbf v_i^\chi\Vert^2 \le M_V, 
\]
for some finite positive real $M_V$ independent of $n$. This proves part (ii).\smallskip

Finally, using the same arguments in Lemma \ref{lem:covarianze}(iv), from \eqref{eq:DK},  part (i) and Lemma \ref{lem:Gxi}(iv), and since $\mu_0^\chi=\infty$ and $\mu_{r+1}^\chi =0$, as $n\to\infty$
\begin{align}
\max_{i=1,\ldots, n}
\min(n,T)\E\l[\Vert\sqrt n(\wh{\mbf v}_i^{x\prime}-\mbf v_i^{\chi\prime}\mbf J)\Vert^2\r]
&=\max_{i=1,\ldots, n}\min(n,T)\E\l[\Vert\sqrt n\bm\varepsilon_i^\prime(\wh{\mbf V}^x-\mbf V^\chi\mbf J)\Vert^2\r]\nn\\
&\le\max_{i=1,\ldots, n}\ \frac
{\min(n,T) 2^3 \frac r{n^2}  n \E\l[\l\Vert\bm\varepsilon_i^\prime(\wh{\bm\Gamma}^x-\bm\Gamma^\chi) \r\Vert^2\r]}
{\frac 1{n^2}\l\{\min(\vert \mu_{0}^\chi-\mu_{1}^\chi \vert,\vert \mu_{r}^\chi-\mu_{r+1}^\chi \vert)\r\}^2}\nn\\
&=\max_{i=1,\ldots, n}\ \frac
{\min(n,T) 2^3 \frac r{n} \E\l[\l\Vert\bm\varepsilon_i^\prime(\wh{\bm\Gamma}^x-\bm\Gamma^\chi) \r\Vert^2\r]}
{\frac 1{n^2}\l\{\min(\vert \mu_{0}^\chi-\mu_{1}^\chi \vert,\vert \mu_{r}^\chi-\mu_{r+1}^\chi \vert)\r\}^2}\nn\\
&\le \frac{8r M_1}{\underline C_r^2}= M_2,\nn
\end{align}
where $M_2$ is a finite positive real independent of $n$ and $T$  defined in Lemma \ref{lem:covarianze}(iv).
This proves part (iii) and completes the proof. 
\end{proof}

\begin{lem}\label{lem:MO1} Under Assumptions \ref{ass:common} through \ref{ass:ind}, for all $n,T\in\mathbb N$,
\begin{compactenum}[(i)]
\item $\Vert\frac{{\mbf M}^\chi}{n}\Vert =O(1)$;
\item $\Vert(\frac{{\mbf M}^\chi}{n})^{-1}\Vert=O(1)$;
\item $\Vert\frac{\wh{\mbf M}^x}{n}\Vert=O_{\mathrm{P}}(1)$;
\item $\Vert(\frac{\wh{\mbf M}^x}{n})^{-1}\Vert=O_{\mathrm{P}}(1)$.
\end{compactenum}
\end{lem}

\begin{proof}
 Parts (i) and (ii) follow directly from Lemma \ref{lem:Gxi}(iv), indeed,
\[
\l\Vert\frac{{\mbf M}^\chi}{n}\r\Vert =\frac{\mu_1^\chi}{n}\le \overline C_1,
\]
and
\[
\l\Vert\l(\frac{{\mbf M}^\chi}{n}\r)^{-1}\r\Vert = \frac n{\mu_r^\chi}\le \frac 1{\underline C_r}.
\]
Both statements hold for all $n\in\mathbb N$ since the eigenvalues are an increasing sequence in $n$.

For part (iii), because of part (i) and Lemma \ref{lem:covarianze}(iii),
\begin{align}
\l\Vert\frac{\wh{\mbf M}^x}{n}\r\Vert\le 
\l\Vert\frac{{\mbf M}^\chi}{n}\r\Vert+
\l\Vert\frac{\wh{\mbf M}^x}{n}-\frac{{\mbf M}^\chi}{n}\r\Vert\le \overline C_1+O_{\mathrm P}\l(\max\l(\frac1{\sqrt n},\frac 1{\sqrt T}\r)\r).\nn
\end{align}
For part (iv) just notice that, because of Lemma \ref{lem:covarianze}(iii) and part (ii), then  $\frac{\wh{\mbf M}^x}{n}$ is positive definite with probability tending to one as $n,T\to\infty$.
This completes the proof. 
\end{proof}

\begin{lem}\label{lem:Vzero} Under Assumptions \ref{ass:common} through \ref{ass:ind}, 
\begin{compactenum}[(i)]
\item $\Vert\frac{\mbf M^\chi}{n}-\bm V_0 \Vert=o(1)$, as $n\to\infty$ and $\Vert \bm V_0\Vert = O(1)$;
\item $\Vert\frac{\wh{\mbf M}^x}{n}-\bm V_0 \Vert=o_{\mathrm {ms}}(1)$, as $n,T\to\infty$;
\item $\Vert(\frac{\mbf M^\chi}{n})^{-1}-\bm V_0^{-1} \Vert=o(1)$, as $n\to\infty$ and $\Vert \bm V_0^{-1}\Vert = O(1)$
\item $\Vert(\frac{\wh{\mbf M}^x}{n})^{-1}-\bm V_0^{-1} \Vert=o_{\mathrm {ms}}(1)$, as $n,T\to\infty$,
\end{compactenum}
where $\bm V_0$ is $r\times r$ diagonal with entries the eigenvalues of $\bm\Sigma_\Lambda\bm\Gamma^F$ sorted in descending order.
\end{lem}

\begin{proof}
 For part (i), first notice that the $r$ non-zero eigenvalues of $\frac{\bm\Gamma^\chi}{n}$ are also the $r$ eigenvalues of $(\bm\Gamma^F)^{1/2}\frac{\bm\Lambda^\prime\bm\Lambda}{n}(\bm\Gamma^F)^{1/2}$ which in turn are also the entries of $\bm V_0$. The proof follows from continuity of eigenvalues and since, because of Assumption \ref{ass:common}(a), as $n\to\infty$, 
\[
\l\Vert(\bm\Gamma^F)^{1/2}\frac{\bm\Lambda^\prime\bm\Lambda}{n}(\bm\Gamma^F)^{1/2}
-(\bm\Gamma^F)^{1/2}\bm\Sigma_\Lambda(\bm\Gamma^F)^{1/2}
\r\Vert=o(1).\nn
\]
Part (ii) is a consequence of part (i) and Lemma \ref{lem:MO1}(iii). \smallskip
 For part (iii) notice that
 \[
 \l\Vert\l(\frac{\mbf M^\chi}{n}\r)^{-1}-\bm V_0^{-1} \r\Vert \le\l \Vert\l(\frac{\mbf M^\chi}{n}\r)^{-1}\r\Vert\,
 \l\Vert\frac{\mbf M^\chi}{n}-\bm V_0 \r\Vert
\, \l\Vert
\bm V_0^{-1}
 \r\Vert,
 \] 
 then the proof follows from part (i), Lemma \ref{lem:MO1}(ii), and since  $\bm V_0$ is positive definite since $\bm\Sigma_\Lambda$ and $\bm\Gamma^F$ are positive definite by Assumptions \ref{ass:common}(a) and \ref{ass:common}(b), respectively. This completes the proof. 
 \end{proof}
 
\begin{lem}\label{lem:KO1} Under Assumptions  \ref{ass:common} through \ref{ass:ind}, as $n\to\infty$,
\begin{compactenum}[(i)]
\item $\Vert{\mbf K}\Vert=O(1)$;
\item $\Vert{\mbf K}^{-1}\Vert=O(1)$.
\end{compactenum}
\end{lem}

\begin{proof}
 For part (a), from \eqref{eq:ups} in the proof of Proposition \ref{prop:KKK}(a) it follows that
\beq\label{eq:ups9}
 \l\Vert\mbf K-\mbf J\bm\Upsilon_0\r\Vert =o(1),
\eeq
where $\bm\Upsilon_0$ is the $r\times r$ matrix of normalized eigenvectors of $(\bm\Gamma^F)^{1/2}\bm\Sigma_\Lambda(\bm\Gamma^F)^{1/2}$ and $\mbf J$ is a diagonal matrix with entries $\pm 1$. Part (i) follows from the fact that $\Vert \mbf J\bm\Upsilon_0\Vert = O(1)$, since $(\bm\Gamma^F)^{1/2}\bm\Sigma_\Lambda(\bm\Gamma^F)^{1/2}$ is finite. Likewise part (ii) follows from the fact that $\mbf J$ is obviously positive definite and $\bm\Upsilon_0$ is also positive definite because the eigenvalues of $(\bm\Gamma^F)^{1/2}\bm\Sigma_\Lambda(\bm\Gamma^F)^{1/2}$ are distinct by Assumption \ref{ass:eval}. This completes the proof. 
\end{proof}

\begin{lem}\label{lem:HO1bis} Under Assumptions  \ref{ass:common} through \ref{ass:ind}, as $n\to\infty$,
$\Vert{\bm{\mathcal H}}\Vert=O(1)$.
\end{lem}

\begin{proof}
 From  \eqref{eq:VKH} and \eqref{eq:mcH}  in the proof of Proposition \ref{prop:L} $\bm{\mathcal H}=(\bm\Gamma^F )^{1/2}\mbf K\mbf J$. Then, the proof follows immediately from Assumption \ref{ass:common}(b), Lemma \ref{lem:KO1}(i), and since $\mbf J$ is obviously finite and positive definite.
 \end{proof}

\begin{lem}\label{lem:HO1} Under Assumptions  \ref{ass:common} through \ref{ass:ind}, as $n,T\to\infty$,
$\Vert\wh{\mbf H}\Vert=O_{\mathrm {P}}(1)$.
\end{lem}

\begin{proof}
From \eqref{eq:acca}, by Proposition \ref{prop:L}(a), Lemma \ref{lem:LLN}(v), \ref{lem:MO1}(iv) \ref{lem:HO1bis}(i), we have
\begin{align}
\l\Vert\wh{\mbf H}\r\Vert &\le \l\Vert\frac{\bm F^\prime\bm F}{T}\r\Vert\,
\l\Vert\frac{\bm\Lambda^\prime\wh{\bm\Lambda}}{n}\r\Vert\,
\l\Vert\l(\frac{\wh{\mbf M}^x}{n}\r)^{-1}\r\Vert\le \l\Vert\frac{\bm F}{\sqrt T}\r\Vert^2\, \l\Vert\frac{\bm\Lambda}{\sqrt n}\r\Vert\, \l\Vert\frac{\wh{\bm\Lambda}}{\sqrt n}\r\Vert\,\l\Vert\l(\frac{\wh{\mbf M}^x}{n}\r)^{-1}\r\Vert\nn\\
&\le\l\Vert\frac{\bm F}{\sqrt T}\r\Vert^2\, \l\Vert\frac{\bm\Lambda}{\sqrt n}\r\Vert^2\, \l\Vert\bm{\mathcal H}\r\Vert\, \l\Vert\l(\frac{\wh{\mbf M}^x}{n}\r)^{-1}\r\Vert+ \l\Vert\frac{\bm F}{\sqrt T}\r\Vert^2\, \l\Vert\frac{\bm\Lambda}{\sqrt n}\r\Vert\, \l\Vert\frac{\wh{\bm\Lambda}-\bm\Lambda\bm{\mathcal H}}{\sqrt n}\r\Vert\,\l\Vert\l(\frac{\wh{\mbf M}^x}{n}\r)^{-1}\r\Vert\nn\\
&= 
O_{\mathrm P}(1)+ O_{\mathrm P}\l(\max\l(\frac 1{\sqrt n},\frac 1{\sqrt T}\r)\r),\nn
\end{align}
which  completes the proof.
\end{proof}

\begin{lem}\label{lem:woodbury}
Under Assumptions \ref{ass:common} through \ref{ass:ind} and letting $\bm\Sigma^\xi$ be the $n\times n$ diagonal matrix with diagonal entries the diagonal entries of $\bm\Gamma^\xi$, as $n\to\infty$,
$$
n^{-1}\l\Vert\{\mbf I_r+{\bm\Lambda}^\prime({\bm\Sigma}^\xi)^{-1} {\bm\Lambda}\}^{-1}\l\{{\bm\Lambda}^\prime({\bm\Sigma}^\xi)^{-1}{\bm\Lambda}\r\}-\mbf I_r\r\Vert= O\l(1\r).
$$
\end{lem}

\begin{proof}
First notice that for any two invertible matrices $\bm K$ and $\bm H$ we have
\begin{align}
(\bm H+ \bm K)^{-1}&= (\bm H + \bm K)^{-1}- \bm K^{-1} + \bm K^{-1}= (\bm H+ \bm K)^{-1}(\bm K - (\bm H + \bm K))\bm K^{-1}+ \bm K^{-1}\nn\\
&= (\bm H + \bm K)^{-1}(-\bm H)\bm K^{-1} + \bm K^{-1}= \bm K^{-1}- (\bm H + \bm K)^{-1}\bm H\bm K^{-1}.\label{eq:inverse}
\end{align}
Then, setting $\bm K=\bm \Lambda^\prime(\bm\Sigma^\xi)^{-1}\bm\Lambda $ and $\bm H=\mbf I_r$ from \eqref{eq:inverse} we have
\begin{align}
\{\mbf I_r+{\bm\Lambda}^\prime({\bm\Sigma}^\xi)^{-1} {\bm\Lambda}\}^{-1}=\{\bm \Lambda^\prime(\bm\Sigma^\xi)^{-1}\bm\Lambda \}^{-1}-\{\mbf I_r+{\bm\Lambda}^\prime({\bm\Sigma}^\xi)^{-1} {\bm\Lambda}\}^{-1}\{\bm \Lambda^\prime(\bm\Sigma^\xi)^{-1}\bm\Lambda \}^{-1}
.\label{eq:abcinv00}
\end{align}
which implies
\beq\label{eq:abcinv}
\{\mbf I_r+{\bm\Lambda}^\prime({\bm\Sigma}^\xi)^{-1} {\bm\Lambda}\}^{-1}\{\bm \Lambda^\prime(\bm\Sigma^\xi)^{-1}\bm\Lambda \} = \mbf I_r-\{\mbf I_r+{\bm\Lambda}^\prime({\bm\Sigma}^\xi)^{-1} {\bm\Lambda}\}^{-1}.
\eeq
Then, by Weyl's inequality:
\beq
\nu_r(\mbf I_r+{\bm\Lambda}^\prime({\bm\Sigma}^\xi)^{-1} {\bm\Lambda}) \ge 1+\nu_r({\bm\Lambda}^\prime({\bm\Sigma}^\xi)^{-1} {\bm\Lambda})\ge \nu_r({\bm\Lambda}^\prime({\bm\Sigma}^\xi)^{-1} {\bm\Lambda}).\label{eq:weyllunga}
\eeq
From, \eqref{eq:weyllunga}, we have
\begin{align}
\Vert \{\mbf I_r+{\bm\Lambda}^\prime({\bm\Sigma}^\xi)^{-1} {\bm\Lambda}\}^{-1}\Vert&=\frac 1{\nu_r(\mbf I_r+{\bm\Lambda}^\prime({\bm\Sigma}^\xi)^{-1} {\bm\Lambda})}\le \frac 1{\nu_r({\bm\Lambda}^\prime({\bm\Sigma}^\xi)^{-1} {\bm\Lambda})}.
\label{eq:abcinv2}
\end{align}
Now, the $r$ eigenvalues of ${\bm\Lambda}^\prime({\bm\Sigma}^\xi)^{-1} {\bm\Lambda}$ are also the $r$ non-zero eigenvalues of $\bm \Lambda\bm \Lambda^\prime \bm ({\bm\Sigma}^\xi)^{-1}$, and the $r$ non-zero eigenvalues of $\bm \Lambda\bm \Lambda^\prime$ are the $r$ eigenvalues of $\bm \Lambda^\prime\bm \Lambda$.
Therefore, because of \citet[Theorem 7]{MK04} and Proposition \ref{prop:K00}(a), and Lemma \ref{lem:Gxi}(iv):
\begin{align}
\nu_r({\bm\Lambda}^\prime({\bm\Sigma}^\xi)^{-1} {\bm\Lambda})&\ge   \nu_r (\bm \Lambda\bm \Lambda^\prime)\nu_r(({\bm\Sigma}^\xi)^{-1})=\frac {\nu_r (\bm \Lambda^\prime\bm \Lambda)}{\nu_1({\bm\Sigma}^\xi)} =  \frac {\mu_r^\chi}{\max_{i=1,\ldots, n}\sigma_i^2}\ge \frac{n\underline C_r}{C_\xi}.\label{eq:THEEND}
\end{align}
By substituting \eqref{eq:THEEND} into \eqref{eq:abcinv2} we have
\beq\label{eq:abcinv3}
\Vert \{\mbf I_r+{\bm\Lambda}^\prime({\bm\Sigma}^\xi)^{-1} {\bm\Lambda}\}^{-1}\Vert\le  \frac{C_\xi}{n\underline C_r}.
\eeq
Therefore, from \eqref{eq:abcinv} and \eqref{eq:abcinv3}, 
\[
\l\Vert \{\mbf I_r+{\bm\Lambda}^\prime({\bm\Sigma}^\xi)^{-1} {\bm\Lambda}\}^{-1}\{\bm \Lambda^\prime(\bm\Sigma^\xi)^{-1}\bm\Lambda \} - \mbf I_r\r\Vert =\l\Vert \{\mbf I_r+{\bm\Lambda}^\prime({\bm\Sigma}^\xi)^{-1} {\bm\Lambda}\}^{-1}\r\Vert\le \frac{C_\xi}{n\underline C_r},
\]
which completes the proof. 
\end{proof}
$\,$
\end{appendix}

\end{document}